\documentclass[11pt]{amsart}

\usepackage[english]{babel}
\usepackage[T1]{fontenc}
\usepackage[utf8]{inputenc}
\usepackage{libertine}
\usepackage{enumitem}
\setlist[itemize]{nosep}
\setlist[enumerate]{nosep}
\setlist[itemize]{topsep=3pt}
\setlist[enumerate]{topsep=3pt}

\usepackage[dvipsnames]{xcolor}
\definecolor{Eggplant}{RGB}{97, 64, 81}

\usepackage{url}

\usepackage{amsmath}
\usepackage{amsfonts}
\usepackage{amsthm}
\usepackage{amssymb}
\usepackage{upgreek}
\usepackage{nicefrac}
\usepackage{faktor}
\usepackage{xspace}

\newtheorem{theorem}{Theorem}
\newtheorem{proposition}{Proposition}
\newtheorem{corollary}{Corollary}
\newtheorem{lemma}{Lemma}
\newtheorem{remark}{Remark}
\newtheorem{definition}{Definition}
\newtheorem{heuristic}{Heuristic}
\newtheorem*{problem}{Problem}
\usepackage[linesnumbered, longend]{algorithm2e}
\makeatletter
\newcommand{\algorithmstyle}[1]{\renewcommand{\algocf@style}{#1}}
\makeatother

\DontPrintSemicolon

\SetKwInOut{Input}{Input}
\SetKwInOut{Output}{Output}
\SetKwInOut{Parameter}{Parameter}

\SetKwSty{kwfont}

\SetKw{And}{and}
\SetKw{Goto}{goto}
\SetKw{Bystep}{by step of}
\SetKw{True}{true}
\SetKw{Break}{break}
\SetKw{False}{false}
\SetKw{Downto}{downto}
\SetKw{To}{to}

\SetKwRepeat{Do}{do}{while}

\SetCommentSty{commfont}

\SetNlSty{}{}{}
\let\oldnl\nl
\newcommand\nonl{
  \renewcommand{\nl}{\let\nl\oldnl}}

\usepackage[most]{tcolorbox}
\tcbset{
    algotitle/.style={title={\strut
        Algorithm~\thetcbcounter\ifstrempty{#1}{\ignorespaces}{~---~#1}}}}

\newtcolorbox[auto counter]{boxedAlgorithm}[1][]{
    colback=white,
    colframe=Eggplant,
    boxrule=1pt,
    titlerule=0pt,
    sharp corners=all,
    colbacktitle=white,enhanced,
    attach boxed title to top center={yshift=-10pt},
    boxed title style={boxrule=-1pt},
    fonttitle=\sffamily,
    coltitle=Eggplant,
    algotitle={},
    #1
}

\usepackage{bm}
\usepackage{wrapfig}
\usepackage{multirow}
\usepackage{graphicx}
\usepackage{tabularx}
\usepackage{colortbl}
\usepackage{booktabs}

\usepackage{caption}
\usepackage{subcaption}

\usepackage{rotating}
\usepackage{pdflscape}
\usepackage{lscape}
\usepackage{tikz}
\usepackage{tikz-cd}
\usetikzlibrary{matrix}
\usetikzlibrary{patterns}
\usetikzlibrary{chains,fit,shapes}
\usetikzlibrary{arrows}
\usetikzlibrary{decorations.pathreplacing}

\usepackage{todonotes}

\makeatletter
\def\tagform@#1{\maketag@@@{\ignorespaces#1\unskip\@@italiccorr}}
\let\orgtheequation\theequation
\def\theequation{(\orgtheequation)}
\makeatother

\addto\extrasenglish{

}


\let\originalleft\left
\let\originalright\right
\renewcommand{\left}{\mathopen{}\mathclose\bgroup\originalleft}
\renewcommand{\right}{\aftergroup\egroup\originalright}

\newcommand{\GSO}{\mbox{\textsc{gso}\xspace}}
\newcommand{\LLL}{\mbox{\textsc{lll}}}
\newcommand{\BKZ}{\mbox{\textsc{bkz}}}
\newcommand{\DBKZ}{\mbox{\textsc{dbkz}}}

\newcommand{\GCD}{\mbox{\textsc{gcd}}}

\DeclareMathOperator{\id}{Id}
\DeclareMathOperator{\tr}{tr}

\DeclareMathOperator{\degree}{deg}
\DeclareMathOperator{\covol}{vol}

\newcommand{\Vol}[1]{\covol\left({#1}\right)}
\newcommand{\inner}[2]{\langle {#1},{#2} \rangle}

\renewcommand{\deg}[1]{{\degree}\ {#1}}

\newcommand{\E}{\mathbb{E}}

\newcommand{\order}{\mathcal{O}}
\newcommand{\ideal}[1]{\mathfrak{#1}}

\newcommand{\norm}{\mathcal{N}}

\newcommand{\NN}{\mathbf{N}}

\newcommand{\ZZ}{\mathbf{Z}}
\newcommand{\RR}{\mathbf{R}}

\newcommand{\KK}{\mathbf{K}}
\newcommand{\QQ}{\mathbf{Q}}
\newcommand{\CC}{\mathbf{C}}

\newcommand{\lL}{\mathbf{L}}

\newcommand{\Gl}{\textrm{GL}}

\newcommand{\Lat}{\Lambda}
\newcommand{\module}{\mathcal{M}}

\newcommand{\me}{\mathrm{e}}
\newcommand{\algName}[1]{{\color{Eggplant}\textbf{\textsf{#1}}\color{black}}}

\newcommand{\bigO}[1]{\textrm O\left(#1\right)}
\newcommand{\bigOtilde}[1]{\tilde{\textrm O}\left(#1\right)}
\newcommand{\littleO}[1]{o\left(#1\right)}

\newcommand{\profile}{\mu}

\newcommand{\Mod}{\bmod}
\newcommand{\Log}{\textrm{Log}}
\newcommand{\rad}{\textrm{rad}}
\newcommand{\Id}{\textrm{Id}}

\let\conj\overline

\usepackage{hyperref}
\hypersetup{
  bookmarks=true,
  colorlinks = true,
  linkcolor  = Eggplant,
  citecolor  = Eggplant,
  filecolor  = magenta,
  urlcolor   = cyan
}
\usepackage[nameinlink]{cleveref}

\crefname{tcb@cnt@boxedAlgorithm}{algorithm}{algorithms}
\Crefname{tcb@cnt@boxedAlgorithm}{Algorithm}{Algorithms}
\crefname{heuristic}{heuristic}{heuristics}
\Crefname{heuristic}{Heuristic}{Heuristics}

\title{Algebraic and Euclidean Lattices: Optimal Lattice Reduction and Beyond}

 \author{
 Paul Kirchner, Thomas Espitau and Pierre-Alain Fouque }
 \address{Sorbonne universit\'es, lip6, paris, france}
 \email{t.espitau@gmail.com}
 \address{Rennes Univ, irisa}
 \email{pa.fouque@gmail.com}
 \email{paul.kirchner@irisa.fr}
 \thanks{This work has been supported in part by the European Union H2020
 Programme under grant agreement number \textsc{ERC}-669891 and
 \textsc{Prometheus Project}-780701.}

\begin{document}

\maketitle

\begin{abstract}

We introduce a framework generalizing lattice reduction algorithms to module
lattices in order to practically and efficiently solve the $\gamma$-Hermite Module-SVP
problem over arbitrary cyclotomic fields. The core idea is to exploit the
structure of the subfields for designing a doubly-recursive strategy of
reduction: both recursive in the rank of the module and in the field we are
working in. Besides, we demonstrate how to leverage the inherent symplectic
geometry existing in the tower of fields to provide a significant speed-up of the
reduction for rank two modules. The recursive strategy over the rank can also
be applied to the reduction of Euclidean lattices, and we can
perform a reduction in asymptotically almost the same time as matrix
multiplication.  As a byproduct of the design of these fast reductions, we
also generalize to all cyclotomic fields and provide speedups for many
previous number theoretical algorithms.

Quantitatively, we show that a module of
rank 2 over a cyclotomic field of
degree $n$ can be heuristically reduced within approximation factor
$2^{\bigOtilde{n}}$ in time $\bigOtilde{n^2B}$, where $B$ is the bitlength of
the entries. For $B$ large enough, this complexity shrinks to
$\bigOtilde{n^{\log_2 3}B}$. This last result is particularly striking as it goes below
the estimate of $n^2B$ swaps given by the classical analysis of the \LLL{}
algorithm using the so-called potential.

Finally, all this framework is fully parallelizable, and we provide a full
implementation. We apply it to break multilinear cryptographic candidates on
concrete proposed parameters. We were able to reduce matrices of dimension
4096 with 6675-bit integers in 4 days, which is more than a million times
faster than previous state-of-the-art implementations. Eventually, we
demonstrate a quasicubic time for the Gentry-Szydlo algorithm which finds a
generator given the relative norm and a basis of an ideal. This algorithm is
important in cryptanalysis and requires efficient ideal multiplications and
lattice reductions; as such we can practically use it in dimension 1024.

 \end{abstract}

\newpage

\section{Introduction}
Lattice-based cryptography increasingly uses ideal and module lattices for
efficiency reasons as the NTRU cryptosystem since 1996. This
achieves quasilinear key size, encryption/decryption and signature time
complexities instead of quadratic. Consequently, it is of utmost importance
to reduce such lattices very efficiently. Peikert in~\cite{Peikert16} asked
the following question: \emph{For worst-case problems on ideal
lattices, especially in cyclotomic rings, are there (possibly quantum)
algorithms that substantially outperform the known ones for general
lattices? If so, do these attacks also extend to work against the
ring-SIS and ring-LWE problems themselves ?} So far, there is no result
in this direction and the security parameters are chosen so that these
lattices are as hard to reduce as random lattices.

The classical way of reducing algebraic lattices starts by \emph{descending}
the algebraic lattice over the integers $\ZZ$. This corresponds to forgetting the
algebraic structure of the module and running a reduction algorithm on it. But
the image over $\ZZ$ of a rank $d$ algebraic lattice is of rank $d\times n$,
where $n$ is the degree of field inside which we are working initially.
Hence, even in the case where the lattice is of small rank, the reduction
can be very costly as the actual dimension over $\ZZ$ might be large.
This process is forgetful of the algebraic specificities of the
base ring. But these properties translate into symmetries over
modules, as they are very structured. Consequently, the above-mentioned
reduction cannot take these symmetries into account. Thus, it is natural to
wonder if it is possible to \emph{exploit} the algebraic structure of the
fields to speed up the reduction.

 In this paper, we present several optimal and heuristic algorithms for
 \LLL\nobreakdash-reducing lattices defined over $\ZZ$ and more generally over module
 lattices defined over cyclotomic fields~\cite{DCC:LanSte15}.
 In the special case of rank-2 module, which is the
 case in the cryptanalysis of the NTRU cryptosystem~\cite{NTRU}, we describe more
 specific algorithms. One of them takes into account the symplectic structure
 of these lattices. Since recent advanced cryptographic constructions such as
 multilinear maps~\cite{AC:ACLL15} and fully homomorphic encryption
 schemes~\cite{EC:VGHV10,C:CorLepTib13} increasingly use
 lattices with high dimension and very large numbers,
 our goal is to give \emph{very efficient} and \emph{parallel}
 algorithms to reduce them.  Consequently, we depart from the current research
 line of proved worst-case lattice reductions to present heuristic
 algorithms with high performance. However, the introduced heuristics
 are practically verified and a large part of the algorithms is proven.

\subsection{Technical framework}
We introduce a framework of techniques to provide fast polynomial-time
algorithms for reducing algebraic lattices defined over cyclotomic fields.
The core design principles of our reductions are:
\begin{description}
  \item[A recursive strategy on the rank] The reduction of a rank $d$ lattice
    is performed recursively on large blocks. Instead of relying on a local
    (\LLL-like) strategy consisting in choosing the first (or an arbitrary)
    block for which some progress can be made, we systematically perform the
    reduction of the blocks. This global process is somewhat similar to the
    ironing out strategies of \BKZ{}-like reductions or to the fast variant of
    \LLL{} of Neumaier and Stehlé~\cite{ISSAC:NeuSte16}, where successive passes of local
    reductions are made on the \emph{whole} basis to gradually improve its
    reduceness. However, we differ from the iterative design \`a la \BKZ{} as
    we shift the blocks between odd and even steps to mix all basis vectors as
    in the early parallelized versions of \LLL{} of Villard~\cite{ISSAC:Villard92}. A
    generic instance of two successive passes of our strategy is given in the following:\\
    \hspace{-1.4cm}
    \includegraphics[scale=0.45]{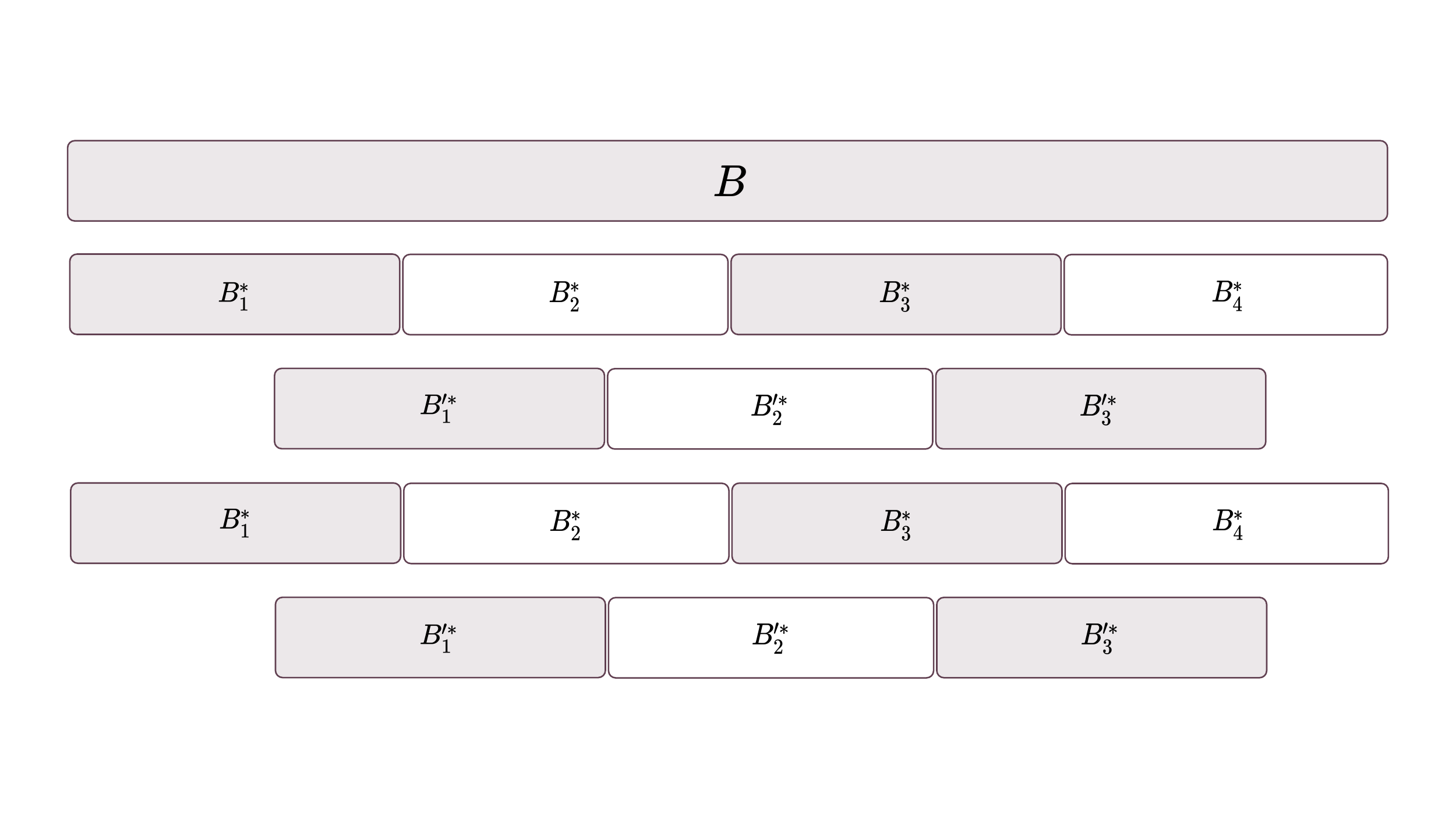}

    The basis $B$ is here sundered in four chunks $B_1, B_2, B_3, B_4$ of
    length $|B|/4$. The reduction process will start by reducing (possibly at
    the same time) the first chunk $B_1^*=B_1$, the projection $B_2^*$ of the second
    one orthogonally to $B_1$, the projection $B_3^*$ of the third one
    orthogonally of $B_1\|B_2$ and so on. When this pass is over, the
    same process starts again, but this time on shifted blocks (i.e.\ the first
    block $B_1'$ starts with the vector $|B|/8$ and is of length $|B|/4$).
    Hence, the rank of the lattices which are called recursively decreases
    until we reach rank 2 lattices, where we can use a fast reduction like
    Sch\"onhage's algorithm~\cite{Schonhage91}.

  \item[A recursive strategy on the degree of the field] Suppose that we are
    given a tower of number fields $\KK_0\subset \KK_1 \subset \ldots\subset
    \KK_h$. Let $\Lat$ be an algebraic lattice defined over the ring of
    integers of the upper field $\KK_h$. We can look at $\Lat$ as an algebraic
    lattice defined over the field right under, that is $\KK_{h-1}$.

    \begin{center}
      \begin{tikzpicture}
        \matrix (m) [matrix of math nodes,row sep=3em,column sep=4em,minimum width=2em]
        {
          \KK_h & \order_{\KK_h} && \Lat\\
          \KK_{h-1} & \order_{\KK_{h-1}}&& \\
          \vdots & \vdots&&\\
        \QQ & \ZZ&&\\};
        \path[-stealth]
          (m-2-1) edge node [left] {} (m-1-1)
          (m-3-1) edge node [left] {} (m-2-1)
          (m-4-1) edge node [left] {} (m-3-1)
          (m-2-2) edge node [left] {} (m-1-2)
          (m-3-2) edge node [left] {} (m-2-2)
          (m-4-2) edge node [left] {} (m-3-2)
          (m-1-4) edge [dashed,-] (m-1-2)
          edge [dashed,-] (m-2-2)
          edge [dashed,-] (m-4-2);
      \end{tikzpicture}
    \end{center}

    Such an identification is possible at the cost of increasing the rank of
    the lattice: the rank of $\Lat$ seen over $\KK_{h-1}$ is exactly
    $[\KK_h:\KK_{h-1}]$ times its rank over $\KK_h$. Then we make use of the
    recursive design  over the rank, introduced above, to reduce this problem
    into numerous instances of reduction of rank two lattices over
    $\KK_{h-1}$. Each of them can be seen over $\KK_{h-2}$, inviting us to
    pursue this descent until we get to the bottom of the tower and are now
    reducing lattices over $\ZZ$, that is, Euclidean lattices.

  \item[A generic use of symplectic structures in number fields] A Euclidean
    space is a vector space endowed with a positive definite symmetric
    bilinear form acting on it. Replacing this form by an antisymmetric one
    yields the notion of \emph{symplectic space}. Lattices embedded in
    symplectic spaces have additional symmetries that can be exploited to
    (roughly) halve the cost of the reduction. We prove that we can define a
    recursive symplectic structure over a tower of number fields. As a
    consequence we can halve the running time of the reduction at \emph{each} level
    of the recursion tree, yielding significant asymptotic speedups on the
    overall reduction.

  \item[A (controlled) low precision reduction]
    We use approximations instead of exact computations, which corresponds to
    reducing the projected sublattices with only the most significant bits of
    their basis.  A careful analysis of the precision required to ensure a
    global reduction gains a factor up to $d$ depending on
    the condition number of the initial basis, where $d$ is the rank of the
    lattice we want to reduce. Furthermore, we can show that the
    precision needed will significantly decrease during \emph{some}
    recursive calls, up to a factor of $d$ once again.

  \item[A fast and generic algorithmic for the log-unit lattice] During the
    reduction of an algebraic lattice, we need to balance the size of the
    Archimedean embeddings of elements to avoid a blow-up of the precision
    used. This can be done by carefully multiplying the considered quantities
    by units of the field, yielding a decoding problem in the so-called
    \emph{log-unit lattice} of cyclotomic fields. We generalize the work of
    Cramer, Ducas, Peikert, and Regev~\cite{EC:CDPR16}, which proved two
    different
    results.  The first is that, given a point, we can find a unit nearby with
    prime-power cyclotomics\footnote{This was later extended by
      Wesolowski~\cite{wesolowski2018arithmetic} to all cyclotomics, however the
    running time is still superquadratic.}. The second one is that, given a
    log-unit lattice point plus some large subgaussian noise, we can find the
    lattice point in polynomial time. We prove that these results can be
    achieved within quasilinear running time, and for any cyclotomic field.
\end{description}

\subsection{Results and practical considerations}
We now discuss the practical implication of the techniques above-mentioned.
Using the recursion on the rank with the low precision technique yields a fast
heuristic reduction algorithm for Euclidean lattices. More precisely we prove
that for a Euclidean lattice given by a matrix $M$ of dimension $d$ with
entries in $\ZZ$ of bitsize at most $B$, with \emph{condition number} bounded by
$2^B$, our reduction algorithm finds a lattice vector $v$ such that $\|v\|
\leq 2^{\frac{d}{2}} |\det M|^{1/d}$ (that is the $ 2^\frac{d}{2}$-Hermite
\textsc{SVP}) in time: \[\bigO{
\frac{d^\omega}{(\omega-2)^2} \cdot \frac{B}{\log B} + d^2B\log B},\] where
$\omega$ is the exponent of matrix
multiplication.
We give in~\cref{app:reduction} a reduction from lattice reduction
to modular linear algebra which suggests that this complexity is almost
optimal.
We also show that for the ubiquitous ``knapsack-like'' matrices, we can further
reduce by a factor of $d$ the complexity.

Combining the recursion over the degree of the number fields yields
a reduction algorithm for module lattices over cyclotomic fields.
Over a cyclotomic field of degree  $n$ and sufficiently smooth conductor, we
can reduce a rank two module
represented as a $2\times 2$ matrix $M$ whose
  number of bits in the input coefficients is uniformly bounded by
  $B>n$, in time \[\bigOtilde{n^2B}.\]
   The first column of the reduced matrix
  has its coefficients uniformly bounded by
  $2^{\bigOtilde{n}}\left(\covol M\right)^{\frac{1}{2n}}$.
Using the symplectic technique gives the
fastest heuristic reduction algorithm over cyclotomic fields, achieving the
same approximation factor of $2^{\bigOtilde{n}}$ in time:
\[
  \bigOtilde{ n^{2+ \frac{\log(1/2+1/2q)}{\log q}} B} + n^{\bigO{\log \log n}}
\]
where $q$ is a prime, and the conductor is a power of $q$.

\subsubsection*{A note on the approximation factor.}
It is noticeable that the approximation factor increases quickly with
the height of the tower. If we can perform a reduction over a number field
above $\QQ$ directly, then there is no need to descend to a
$\ZZ$-basis and we can instead stop at this intermediate level. Actually, the
larger the ring is, the more
efficient the whole routine is. It is well-known that it is possible to come
up with a direct reduction algorithm for an algebraic lattice when the
underlying ring of integer is norm-Euclidean, as first mentioned by
Napias~\cite{Napias}. The reduction algorithm over such a ring
$\mathcal{O}_{\KK}$ can be done exactly as for the classical \LLL~algorithm,
by replacing the norm over $\QQ$ by the algebraic norm over $\KK$. Hence a
natural choice would be
$\ZZ[x]/(x^n + 1)$ with $n \leq 8$ as these rings are proved to be
norm-Euclidean. We explain in~\cref{sec:implementation} how we can in fact
deal with the larger ring $\ZZ[x]/(x^{16} + 1)$ even though it is not
norm-Euclidean.  In several applications, it is interesting to decrease the
approximation factor.  Our technique is, at the lowest level of recursion, and
when the number of bits is low, to use a \LLL-type algorithm. Each time the
reduction is finished, we descend the matrix to a lower level where the
approximation factor is lower.

\subsubsection{Practical impact in cryptography}

We test our algorithm on a large instance coming from multilinear map
candidates based on ideal lattices proposed in~\cite{AC:ACLL15} where
$q\approx 2^{6675}$ and $N=2^{16}$.  We solve this instance over the smaller
field $n=2^{11}$ in 13 core-days. If we compare this computation with the
previous large computation with \textsf{fplll}, Albrecht \emph{et al.}\ were able to
compute with $n=2^8$, $q\approx 2^{240}$ in 120 hours. As the complexity of
their code is about $n^4\log(q)^2$ we can estimate our improvement factor to 4
million.

As a byproduct of our reduction we were also able to drastically enhance the
Gentry-Szydlo algorithm~\cite{EC:GenSzy02}.
The key in this algorithm is to quicken the ideal
arithmetic. Instead of the classical $\ZZ$-basis representation, we choose to
represent ideals with a small family of elements over the order of a subfield
of ${\KK}$. Then, one can represent the product of two ideals using the family
of all products of generators.  However, this leads to a blow-up in the size
of the family. A reasonable approach is then to sample a bit more than $[\lL
: \KK]$ random elements in the product so that with overwhelming probability
the ideal generated by these elements is the product ideal itself.  It then
suffices to reduce the corresponding module with the fast reduction process to
go back to a representation with few generators.

An important piece is then the reduction of an ideal itself. Our practical
approach is here to reduce a square matrix of dimension $[\lL : \KK]$, and
every two rounds to add a new random element with a small Gram-Schmidt in the
ideal at the last position. We show in~\cref{subsec:gs} that the overall
complexity is $\bigOtilde{n^3}$, while the previous implementation was in
$\bigO{n^6}$.  The running time of the first practical implementation
published~\cite{EC:BEFGK17} in dimension 256 was 20 hours while we were able
to do it in 30 minutes. Assuming it is proportional to $n^6$ leads to an
estimate of 10 years for $n = 1024$ while we were able to compute it in 103
hours.

\subsection{Related Work}
Recently some independent line of research started to tackle the problem of
reduction of algebraic lattices~\cite{ASIACRYPT:LPSW19,EPRINT:MukSteDav19}.
These works actually provide polynomial time reduction from
$\gamma$-module-\textsc{svp} (or $\gamma$-Hermite-\textsc{svp}) in small rank
to the same problem in arbitrary rank.
However, an implementation would rely on an actual oracle for this problem,
yielding algorithms whose running time would be exponential in the degree of
the field. We emphasize here that while our techniques rely on \emph{many} heuristics,
the resulting algorithms are implemented and enable a fast reduction of high-dimensional
lattices.

The fastest (theoretical) asymptotic variant of the \LLL{} reduction is
the recursive strategy of Neumaier and Stehl\'e~\cite{ISSAC:NeuSte16}, whose
running time is \[d^4B^{1+\littleO{1}}\] for an integer lattice of rank $d$
with coefficients of bitsize smaller than $B$.\\
In all applications of \LLL{} known to the authors, the condition number of a
matrix is barely larger than the matrix entries; however, we underscore
that it can be much larger.
Also, even though both their and our algorithms are not proven to return
an \LLL{}-reduced basis, but a basis starting with a short vector; in practice the basis
returned is in fact \LLL{}-reduced.

We give an example of a round in Neumaier-Stehlé's algorithm:

\hspace{-1.4cm}
\includegraphics[scale=0.45]{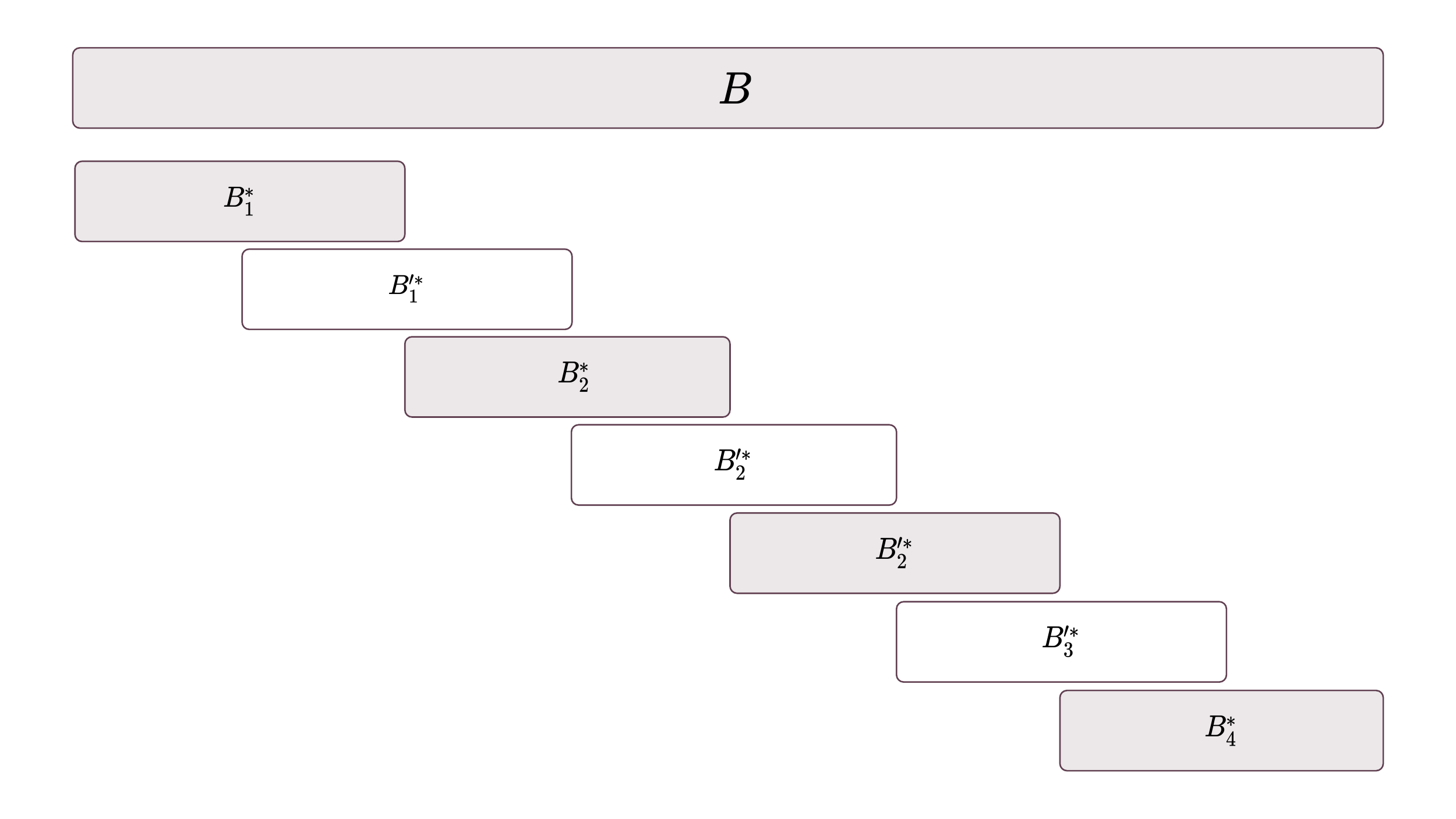}

The main difference is that their round prevent any parallelism.

It is well-known since the work of Schnorr~\cite{Schnorr88} that we can provably use
approximations in the computations to reduce the needed bitsize. However, previous
papers were limited to the ``well-conditioned'' \LLL{}-reduced part of the matrix,
which prevented the use of fast matrix multiplications.
In contrast, we give a framework able to work with approximations, such that the
the number of bits needed is within a small constant factor of the optimal.
This, in turn, enables a reduction in the precision used on the
partially-reduced intermediary bases.

\subsection{Organization of the paper.}
In the next section, we present the mathematical objects we need in the paper
and the \LLL{} algorithm.
In~\cref{sec:fast_LLL_NF} we present the algorithm which reduces rank 2 modules,
whose complexity is analyzed in~\cref{sec:algorithm_complexity}.
In~\cref{sec:fastlll} we show how to efficiently reduce high-rank modules, and
its impact on the reduction of knapsack-like bases.
Then in~\cref{sec:symplectic}, we explain how to use the symplectic structure to
obtain an even faster reduction of rank 2 modules.
We describe tricks for a faster implementation in~\cref{sec:implementation} and
detail applications and compare with a previous implementation
in~\cref{sec:applications}.

The~\cref{app:precision} is dedicated to fast approximate algorithms, as well as
bounding the precision needed.
\Cref{app:unites} explains how to round efficiently with respect to the cyclotomic units.
\Cref{app:sympallnf} indicates ways to obtain a symplectic structure with all
number fields.
Finally, \cref{app:reduction} reduces lattice reduction to modular linear algebra.
 \tableofcontents

\section{Background}
\label{sec:mathematical_background}
We describe the mathematical definitions and lattice reduction
algorithm. For algebraic number theory results, a comprehensive
reference can be found in~\cite{Neukirch}.

\subsection{Notations and conventions}
The bold capitals $\ZZ$, $\QQ$, $\RR$ refer as usual to the
ring of integers and respectively the field of rational and real.
Given a real number $x$, its integral rounding denoted by
$\lfloor x \rceil$ returns its closest.  Its fractional part is the
excess beyond that number's integer part and denoted by $\{x\}$.

These operators are extended to operate on vectors and matrices by point-wise
composition. The complex conjugation of $z\in\CC$ is denoted by the
usual bar $\bar{z}$. The logarithm functions are used as $\log$ for
the binary logarithm and $\ln$ for the natural one.

We say that an integer $n\in\ZZ$ is \textbf{log-smooth} if all the prime
factors of $n$ are bounded by $\log(n)$.

\medskip\noindent
\textbf{Matrix and norms.}
For a field $\KK$, let us denote by $\KK^{d\times d}$ the space of
square matrices of size $d$ over $\KK$, $\Gl_d(\KK)$ its group of
invertibles. Denote classically the elementary matrices by
$T_{i,j}(\lambda)$ and $D_i(\lambda)$ for respectively the transvection
(or shear mapping) and the dilatation of parameter $\lambda$.

We extend the definition of the product for any pair of
matrices $(A,B)$: for every matrix $C$ with compatible size with
$A$ and $B$, we set: $(A,B)\cdot C = (AC,BC)$.

For a vector $v$ (resp. matrix $A$), we denote by $\|v\|_\infty$ (resp.
$\|A\|_{\max}$) its absolute (resp. max) norm, that is the maximum of
the absolute value of its coefficients.

We adopt the following conventions for submatrix extraction: for any
matrix $M = (m_{i,j})\in\KK^{n\times n}$ and $1\leq a<b\leq n, 1\leq
c<d\leq n$, define the extracted submatrix
\[
  M[a:b,c:d] = \left(m_{i,j}\right)_{a\leq i\leq b, c\leq j\leq d},
\]
while $M_i$ refers to the $i$th column of $M$.

\medskip\noindent
\textbf{Computational setting.}
We use the standard model in algorithmic theory, i.e.\ the word-RAM with
unit cost and logarithmic size register (see for instance~\cite[Section
2.2]{mehlhorn2008algorithms} for a comprehensive reference).
The number of bits in the register is $w$.\medskip

For a non-negative integer $d$, we set $\omega(d)$ to be the exponent
of matrix multiplication of $d\times d$ matrices. If the dimension $d$
is clear from context we might omit it and write simply
$\bigO{d^\omega}$ for this complexity. We can assume that this exponent
is not too close to 2, in particular $\omega(d) > 2+1/\log(d)$, so that
complexities with terms in $(\omega-2)^{-1}$ make sense.
Also, we assume that $\omega$ is non-increasing.
Note the conflict with Landau's notations.

\subsection{Background on Algebraic Number Theory}

\medskip\noindent
\textbf{Number fields.}
\label{sec:nf_definition}
A number field $\KK$ is an algebraic extension of $\QQ$ such that: \[\KK
\cong \QQ[X]/(P)=\QQ(\alpha),\]
where $P$ is a monic irreducible polynomial of degree $n$ over $\ZZ$
and $\alpha$ is the image of $X$ in the quotient. For a number field
$\lL$ containing $\KK$ denote by $[\lL:\KK]$ the dimension of $\lL$
seen as a $\KK$-vector space. This integer is called the relative
degree of $\lL$ to $\KK$.
Any element $\gamma$ of $\KK$ has a minimal polynomial, i.e.\ the
unique monic polynomial of least degree among all polynomials of
$\QQ[X]$ vanishing at $\gamma$. An \emph{algebraic integer} has
its minimal polynomial in $\ZZ[X]$.
The set of all integers in $\KK$ forms a ring, called the {\it ring of
integers} or \emph{maximal order} of $\KK$, $\order_\KK$.

Let $\left(\alpha_1,\dotsc,\alpha_n\right) \in \CC^n$ be the
distinct complex roots of $P$. Then, there are $n$ distinct
embeddings, field homomorphisms, of $\KK$ in $\CC$. We
define the $i$-th embedding $\sigma_i:\KK \to
\CC$ as the morphism mapping $\alpha$ to $\alpha_i$.
We distinguish embeddings induced by real roots,
{\it real embeddings} from embeddings coming from
complex roots, {\it complex embeddings}.
Assume that $P$ has $r_1$ real roots and $r_2$ complex roots,
$n=r_1+r_2$.
This leads to the Archimedean \emph{embedding} $\sigma$:
\[
  \begin{array}{cccl}
    \sigma:& \KK & \longrightarrow & \RR^{r_1}\times \CC^{r_2}\\
           &x & \longmapsto & \left(\sigma_1(x),\ldots,\sigma_{r_1}(x),
    \sigma_{r_1+1}(x), \ldots \sigma_{r_1+r_2}(x)\right) .
\end{array}
\]

This embedding can be used to define a Hermitian symmetric bilinear form on
$\KK$, which is positive definite and endows $\KK$ with a natural
Hermitian structure:
\[\inner{a}{b}_\sigma= \sum_{i=1}^{n} \sigma_i(a)\overline{\sigma_i(b)}.\]

\medskip\noindent
\textbf{Modules and Ideals.}
Let fix $R$ be a ring with multiplicative identity $1_R$.
A $R$-module $\module$ consists of an abelian group
$(\module, +)$ and a composition law $\cdot: R\times \module \rightarrow
\module$ which is bilinear and associative.
Suppose $\module$ is a $R$-module and $\mathcal{N}$ is a
subgroup of $\module$. Then $\mathcal{N}$ is a $R$-submodule if, for any
$v$
in $\mathcal{N}$ and any $r$ in $R$, the product $r\cdot v$ is in
$\mathcal{N}$.  A $R$-module $\module$ is  said to be \emph{free} if it
is isomorphic to $R^d$ for some positive integer $d$. Consequently,
there
exists a set of elements $v_1, \ldots, v_d\in \module$ so that
every element in $\module$ can be uniquely written as an $R$-linear
combination of the $v_i$'s. Such a family is called a basis of the
module.

An ideal of $\order_\KK$ is as an $\order_\KK$-submodule of
$\order_\KK$. Every ideal $\ideal{a}$ in number fields are finitely
generated modules that is it can be described by a finite family of
generators i.e.\ expressed as $\alpha_1 \order_\KK + \cdots +\alpha_k
\order_\KK$, for some integer $k$ with the $(\alpha_i)$ belongings to
$\order_\KK$. Since the ring $\order_\KK$ is
Dedekind, any ideal can be generated by two elements.
The product of two ideals $\mathfrak{a}$ and $\mathfrak{b}$
is defined as follows
\[
  \mathfrak{a} \mathfrak{b}:=\left\{a_1v_1+ \dots + a_mv_m \mid a_i \in
    \mathfrak{a} \mbox{ and } v_i \in \mathfrak{b}, i\in\{1,\ldots, n\};
  \mbox{ for } m\in\NN\right\},
\]
i.e., the product is the ideal generated by all products $ab$ with $a
\in\mathfrak{a}$ and $b \in \mathfrak{b}$.

\medskip\noindent
\textbf{Trace and norm in $\KK$.}
Let $\KK\subset\lL$ a number field extension and $n=[\lL:\KK]$.
Let $\sigma^\KK_i:\lL\rightarrow\CC$ the $n$ field embeddings fixing
$\KK$.  For any element $\alpha\in \lL$ define its (relative) algebraic
norm $\norm_{\lL/\KK}(\alpha)$ to be the determinant of the $\KK$-linear
map $x\mapsto x\alpha$.  One can describe this norm
using the $\sigma^\KK_i$ embeddings as: $\norm_{\lL/\KK}(\alpha) =
\prod_{1 \leq i \leq n} \sigma^\KK_i(\alpha),$ showing in particular
that the relative norm is multiplicative.  Similarly define its
(relative) trace $\tr_{\lL/\KK}(\alpha)$ to be the trace of the
$\KK$-linear map $x\mapsto x\alpha$.  This
trace is described using the $\sigma^\KK_i$ embeddings as:
$\tr_{\lL/\KK}(\alpha) = \sum_{1 \leq i \leq n} \sigma^\KK_i(\alpha),$
showing in particular that the relative trace is additive. It is clear
from these definitions that the for any $\alpha\in\lL$, its relative
trace and norm are elements of $\KK$.  Remark that by definition
of the Archimedean structure of $\KK$, we have $\inner{a}{b}_\sigma =
\tr_{\KK/\QQ}\left(a\overline{b}\right)$ for any elements $a,b\in\KK$.
We define the (relative) canonical norm of an element over $\KK$ to be
\[
  \|\alpha\|_{\lL/\KK} = \left(\tr_{\lL/\KK}\left(\alpha
  \overline{\alpha}\right)\right)^\frac{1}{2}.
\]

We easily derive a relation between the algebraic norm of an
integer and its canonical norm,  based on the inequality
of arithmetic and geometric means.

\begin{lemma}[Inequality between relative arithmetic and geometric
  norms]
  \label{lem:inequality}
  Let $\QQ\subset\KK\subset\lL$ a tower of number field.
  For every $\alpha \in \lL$:
  \[
    |\norm_{\lL/\QQ}(\alpha)|\leq
    \left(\frac{\norm_{\KK/\QQ}(\|\alpha\|_{\lL/\KK})}{\sqrt{[\lL:\KK]}}\right)^{[\lL:\KK]}.
  \]
\end{lemma}
\subsection{Cyclotomic fields and Modules over $\ZZ[\zeta_f]$}
We denote by $\Phi_f \in\ZZ[X]$ the $f$-th cyclotomic polynomial, that
is the unique monic polynomial whose roots $\zeta_f^k=\exp(2ik\pi/f)$ with
$\gcd(k,f)=1$ are the $f$-th primitive roots of the unity.
The $f$-th cyclotomic polynomial can be written as:
$\Phi_f = \prod_{k\in\ZZ_f^{\times}} (X-\zeta_f^k)$ and the
cyclotomic field $\QQ(\zeta_f)$ is obtained by adjoining a primitive
root $\zeta_f$ to the rational numbers. As such, $\QQ(\zeta_f)$ is
isomorphic to the field $\QQ[X]/(\Phi_f)$. Its degree over $\QQ$ is
$\deg(\Phi_f) = \varphi(f)$, the Euler totient of $f$. In this specific
class of number fields, the ring of integers is precisely
$\ZZ[X]/(\Phi_f) \cong \ZZ[\zeta_f]$ (see~\cite[Proposition
10.2]{Neukirch} ).

\medskip\noindent
\textbf{Canonical Hermitian structure.}
Let $\module$ be a free module of rank $d$
over the cyclotomic ring of integers $\ZZ[\zeta_f]$. It is isomorphic to
$\bigoplus_{i=1}^d \alpha_i \ZZ[\zeta_f]$, for some
linearly independent vectors $\alpha_i\in\QQ(\zeta_f)^d$. The
Hermitian structure of $\QQ(\zeta_f)^d$ naturally lifts to
$\module{}$ as defined to $\inner{\alpha_i}{\alpha_j} = \sum_{t=1}^d
\tr \left(\alpha_i^{(t)}\overline{\alpha_j^{(t)}}\right)$ on the basis
elements and extended by linearity. We denote by $\|\cdot\|$ the
corresponding norm. More generically we also use this notation to denote
the associated induced norm on endomorphisms (or matrices) over this
$\QQ(\zeta_f)^d$.

\medskip\noindent
\textbf{Relative structure of ring of integers in a tower.}
\label{sec:descending}
Let $\KK\subseteq \lL$ be a subfield of $\lL$ of index $n$.
Then $\order_{\KK}$ is a subring of $\order_\lL$, so that
$\order_\lL$ is a module over $\order_{\KK}$. In whole generality, it is
not necessarily free over $\order_{\KK}$, but by the Steinitz theorem it is
isomorphic to $\order_\KK^{n-1}\oplus \ideal{a}$ for a fractional ideal
$\ideal{a}$ of $\KK$ (see for instance~\cite[Theorem 7.23]{Bourbaki85} ).
Nonetheless, in our case, we only consider the case where both
$\KK$ and $\lL$ are \emph{both} cyclotomic fields.
In this precise situation, $\order_\lL$ is a free $\order_\KK$ module of
rank $n$ over $\order_{\lL}$.
Henceforth, the module $\module$ can itself be viewed as a free module over
$\order_\KK$ of rank $dn$. Indeed, consider $(\xi_1, \ldots,
\xi_n)$ a basis of $\order_\KK$ over $\order_{\lL}$ and $(v_1,
\ldots, v_d)$ a basis of $\module$ over $\order_\KK$. For any $1\leq
i\leq d$, each coefficient of the vector $v_i$ decomposes uniquely in
the basis $(\xi_j)$. Grouping the corresponding coefficients
accordingly yields a decomposition \[ v_i = v_i^{(1)}\xi_1 + \cdots +
v_i^{(d)} \xi_n,\] where $v_i^{(j)}\in\order_{\lL}^{dn}$.  The
family $(v_i^{(j)}\xi_j)_{\substack{1\leq i \leq d, \\1\leq j\leq
n}}$ is a basis of $\module$ viewed as $\order_\KK$-module.

\medskip\noindent
\textbf{Unit rounding in cyclotomic fields.}
The group of units of a number field is the group of invertible
elements of its ring of integers. Giving the complete description of the
units of a generic number field is a computationally hard problem in
algorithmic number theory.
It is possible to describe a subgroup of finite index of the unit
group, called the \emph{cyclotomic units}. This subgroup contains all
the units that are products of elements\footnote{One should notice that
$\zeta_f^i-1$ is not a unit for $f$ a prime-power.} of the form $\zeta_f^i - 1$ for any $1\leq
i\leq f$.

As these units are dense, structured and explicit we can use them to
round an element. The following theorem is a fast variant
of~\cite[Theorem 6.3]{EC:CDPR16}, and is fully proved in
\cref{app:unites}.

\begin{theorem}
  \label{thm:unites}
  Let $\KK$ be the cyclotomic field of conductor $f$.
  There is a quasi-linear randomized algorithm that given any element in
  $x\in (\RR \otimes \KK)^\times$ finds a unit $u\in \order_\KK^\times$
  such that for any field embedding $\sigma:\KK \rightarrow\CC$ we have
  \[ \sigma\left({x}{u}^{-1}\right)= 2^{\bigO{\sqrt{f\log
  f}}}\norm_{\KK/\QQ}(x)^{\frac{1}{\varphi(f)}}.\]
\end{theorem}
\begin{remark}
  Recall that $\frac{f}{\varphi(f)}=\bigO{\log \log f}$, then denoting by
  $n=\varphi(n)$ the dimension of $\KK$, we then shall use the
  bound \[2^{\bigO{\sqrt{n\log n\log \log
  n}}}\norm_{\KK/\QQ}(x)^{\frac{1}{n}},\]
  in the result of \cref{thm:unites}.
\end{remark}
We call \algName{Unit}~the corresponding program.

\subsection{Lattice}
\label{sec:Hermitian_spaces}

\begin{definition}[Lattice]
  A \emph{lattice}  $\Lat$ is a finitely generated free $\ZZ$-module,
  endowed with a Euclidean norm on $\|.\|$ on the rational vector space
  $\Lat\otimes_\ZZ \QQ$.
\end{definition} We may omit to write down the norm to refer to a
lattice $\Lat$ when there is
no ambiguity. By definition of a
finitely-generated free module, there exists a finite family $(v_1,
\ldots, v_d) \in \Lat^d$ such that $\Lat = \bigoplus_{i=1}^d v_i \ZZ$,
called a \emph{basis} of $\Lat$. Every basis has the same number of
elements called the rank of the lattice.

Two different bases of the same lattice $\Lat$ are related by a
unimodular transformation, which is a linear transformation represented
by an element of $\Gl_d(\ZZ)$, set of $d\times d$ integer-valued
matrices of determinant $\pm 1$.  Thus, algorithms acting on lattice
bases can be seen as sequences of unimodular transformations. Among
these procedures, reduction algorithms are of the utmost importance.
They aim at congenial classes of bases,
proving that for any lattice, one can efficiently find
\emph{quasi-orthogonal} bases with controlled norm vectors.
The \emph{volume} of a lattice is defined to be the square root
of the Gram-matrix of any basis, that is:
\[
  \covol \Lat = \sqrt{\det \left(\inner{v_i}{v_j}\right)_{i,j}}
\]

\medskip\noindent
\textbf{Orthogonalization of vectors in Hermitian space.}
\label{sec:gram_schmidt}

Let $\mathcal{S} = (v_1,\dots, v_d)$ a family of linearly independent
vectors of a space $E$.  The orthogonal complement $\mathcal{S}^\bot$ is
the subspace $\{ x\in E~|~ \forall i,~\inner{x}{v_i} = 0\}$. Denote by
$\pi_i$ the orthogonal projection on $(v_{1}, \dotsc, v_{i-1})^\bot$,
with the convention that $\pi_1 = \id$. The Gram-Schmidt
orthogonalization process (\GSO) is an algorithmic method for
orthogonalizing $\mathcal{S}$ while preserving the increasing chain of
subspaces $(\bigoplus_{j=1}^i v_j\RR)_i$. It constructs the orthogonal
set $\mathcal{S}^* = \left(\pi_1(v_1), \ldots, \pi_d(v_d)\right)$. For
notational
simplicity we refer generically to the orthogonalized vectors of such
family by $v_i^*$ for $\pi_i(v_i)$. The computation of
$\mathcal{S}^*$ can be done inductively as follows: for all
$1\leq i\leq d$,
\[
  v_i^* = v_i - \sum_{j=1}^{i-1}
  \frac{\inner{v_i}{v_j^*}}{\inner{v_j^*}{v_j^*}}v_j.
\]
Collect the family $\mathcal{S}$ in a matrix $S$; the Gram-Schmidt transformation
corresponds to the QR decomposition of $S$. Namely we have $S=QR$ for
an orthogonal matrix $Q$ and an upper triangular matrix $R$, where
$R_{i,j} = \frac{\inner{v_i}{v_j^*}}{\|v_j^*\|}$ and $Q = \left[
  \frac{v^*_1}{\|v^*_1\|}, \ldots, \frac{v^*_d}{\|v_d^*\|}\right]$

The volume of the parallelepiped spanned by the vectors of $S$ can be
computed from the Gram-Schmidt vectors $S^*$
as: $\Vol{S} = \prod_{i=1}^d \|v_i^*\|$.

\medskip\noindent
\textbf{Size-reduction of a family of vectors.}
\label{sec:size_reduction}
Let $\Lat$ be a rank $d$ lattice given by a basis $(v_1, \ldots, v_d)$,
we might want to use the Gram-Schmidt process. However since the
quotients
$\frac{\inner{v_i}{v_j^*}}{\inner{v_j^*}{v_j^*}}$ are not
integral in general, the vectors $v_i^*$ may not lie in
$\Lat$. However, we can approximate the result of this process by
taking a rounding to a nearest integer. This process is
called \emph{Size-reduction} and corresponds to the simple iterative
algorithm, where $v_j^*$ refers to the current value of $v_j$:
\[
  \begin{array}{ll}
    &\textrm{\textbf{for~}} i = 2 \textrm{~\textbf{to}~} n \textrm{\textbf{~do}}\\
    &\qquad\textrm{\textbf{for~}} j = i-1 \textrm{~\textbf{to}~} 1 \textrm{\textbf{~do}}\\
    &\qquad\qquad v_i \gets v_i -
    \left\lceil\frac{\inner{v_i}{v_j^*}}{\inner{v_j^*}{v_j^*}}\right\rfloor
    v_j\\
    &\qquad\textrm{\textbf{end}}\\
    &\textrm{\textbf{end}}\\
  \end{array}
\]

\subsection{The LLL~reduction algorithm}
\label{sec:textbook_LLL}
Lenstra, Lenstra, and Lovász~\cite{LLL82} proposed a notion
called \emph{\LLL-reduction} and a polynomial time algorithm
that computes an {\LLL}-reduced basis from an arbitrary basis
of the same lattice. Their reduction notion is formally defined as follows:

\begin{definition}[LLL reduction]
  A basis $\mathcal{B}$ of a lattice is said to be
  $\delta$-\LLL-reduced for  certain parameters $1/4<\delta\leq
  1$, if the following two conditions are satisfied:
  \[
    \forall i <j, \quad
    \left|\langle{{v_j}},{v_i^*}\rangle \right|
    \leq \frac{1}{2}{\|v_i^*\|^2}
    \quad \textrm{(Size-Reduction condition)}
  \]
  \[
    \forall i, \quad \delta \|v_i^*\|^2 \leq
    \|v_{i+1}^*\|^2 +
    \frac{\inner{v_{i+1}}{v_i^*}^2}{\|v_i^*\|^2}
    \quad \textrm{(Lovász condition)}.
  \]
\end{definition}

To find a basis satisfying these conditions, it suffices to
iteratively modify the current basis at any point where one of these
conditions is violated. This yields the simplest version of the
\LLL~algorithm as described in \cref{alg:original_lll}. The method
can be extended to lattices described by a generating family rather than
by a basis~\cite{Pohst87}.

\begin{boxedAlgorithm}[algotitle=Textbook LLL reduction,label=alg:original_lll]
  \begin{algorithm}[H]
    \BlankLine
    \KwIn{Initial basis $B = ({b_1}, \ldots, {b_d})$}
    \KwResult{A $\delta$-\LLL-reduced basis}
    \BlankLine
    $k\gets 1$\;
    \While{$k< d$}{
      Compute the $R$ part of the QR-decomposition of $B$\;
      \For{$j = k-1$ \Downto $1$}{
        $b_k \gets b_k - \left\lceil R_{k,j} \right\rfloor\cdot
        b_j$ \;
        $R_k \gets R_k - \left\lceil R_{k,j} \right\rfloor\cdot
        R_j$
      }
      \uIf{\small{$\delta \|(R_{k,k},0)\|^2\leq \|(R_{k+1,k},
        R_{k+1,k+1})\|^2$}}{
        $k \gets k+1$\;
      }
      \uElse{
        Swap $b_k$ and $b_{k+1}$\;
        $k \gets \max(k-1,1)$\;
      }
    }
    \Return (${b_1}, \ldots, {b_d}$)
  \end{algorithm}
\end{boxedAlgorithm}

\medskip\noindent
\textbf{Decrease of the potential and complexity.}
The algorithm can only terminate when the current lattice basis is
{\LLL}-reduced. Moreover, as shown in~\cite{LLL82}, it terminates in
polynomial time when $\delta < 1$. Indeed, consider the (square of
the) product of the volumes of the flag associated with the basis:
$\prod_{i=1}^d \|v_i^*\|^{2(n-i+1)},$ which is often called its
\emph{potential}. This value decreases by a factor at least
$\delta^{-1}$ in each exchange step and is left unchanged by other
operations. Indeed:
\begin{itemize}
  \item The flag is not modified by any operation other than swaps.
  \item A swap between ${v_k}$ and ${v_{k-1}}$ only changes the sublattice
    spanned by the $k-1$ first vectors. The corresponding volume
    $\prod_{i=1}^{k-1} \|v_i^*\|^{2}$ decreases by a factor at least
    $\delta^{-1}$ and so does the potential.
\end{itemize}
Since the total number of iterations can be bounded by
twice the number of swaps plus the dimension of the lattice, this
suffices to conclude that it is
bounded by $\bigO{d^2B }$ where $B$ is a bound on the size of the
coefficients of the matrix
of the initial basis.  As the cost of a loop iteration is of $\bigO{d^2}$
arithmetic operations on \emph{rational} coefficients of length at most
$\bigO{dB}$, the total cost in term of arithmetic
operations is loosely bounded by $\bigO{d^6 B^3}$.

\medskip\noindent
\textbf{Reduceness of LLL-reduced bases and approximation factor.}
Let $\Lat$ be a rank $d$ lattice and $v_1, \ldots, v_d$ a
$\delta$-\LLL~ reduced basis of $\Lat$. The length of
vectors and orthogonality defect of this basis is related
to the reduction parameter $\delta$:

\begin{proposition}
  \label{prop:norm_bound}
  Let $1/4< \delta< 1$ be an admissible \LLL~parameter.
  Let $(v_1, \ldots, v_d)$ a $\delta$-\LLL~reduced basis of rank-$d$ lattice
  $(\Lat, \inner{\cdot}{\cdot})$. Then for any $1\leq k\leq d$:
  \[
    \Vol{v_1,\ldots,v_k}\leq (\delta-1/4)^{-\frac{(d-k)k}{4}}
    \Vol{\Lat}^{\frac{k}{d}} .
  \]
\end{proposition}

\subsection{$\order_\KK$-lattices}

We now generalize the notion of Euclidean lattice to the higher-degree
context. As a lattice is a finitely
generated free $\ZZ$-module $\Lat$ endowed with a Euclidean structure
on its real ambient space $\Lat\otimes_\ZZ\RR$. To extend this
definition we want to replace the base-ring $\ZZ$ by the ring of integer
$\order_\KK$ of a number field $\KK$. In the present context we will
keep the freeness condition of the module, even if this setting is
slightly too restrictive in general\footnote{Indeed in a tower of
field $\QQ\subseteq \KK\subseteq\lL$, the module $\order_\lL$ seen
over the Dedekind domain $\order_\KK$ is not necessarily free. Hence
using as definition for such a generalized lattice $\Lat$ to be a free
$\order_\lL$-module would forbid $\Lat$ to be a lattice over
$\order_\KK$. Relaxing the freeness into projectiveness is however
sufficient as $\order_\lL$ is always a projective
$\order_\KK$-module.}.

\begin{definition}[$\order_\KK$-lattice]
  \label{def:algebraic_lattice}
  Let $\KK$ be a cyclotomic number field. An $\order_\KK$-lattice---or
  algebraic lattice over $\order_\KK$---is a free
  $\order_\KK$-module $\Lat$ endowed with a $\KK\otimes\RR$-linear
  positive definite self-adjoint\footnote{The definition of such a form
    is done in the usual manner: it is represented by a matrix $A$ such
    that $A=A^*$ for $*$ being the composition of the transposition and
  conjugation operator of $\KK_\RR$.} form on the ambient vector space
  $\Lat\otimes_{\order_\KK} \RR$.
\end{definition}

As for Euclidean lattices without loss of generality we can only look at
the case where the inner product is the one derived from the
polarization of the canonical norm introduced in
\cref{sec:nf_definition}. As for the Euclidean case we now study the
orthogonalization process in such space and devise the equivalent notion
of the volume of an algebraic lattice.

\medskip\noindent
Taking the basis $(m_1, \ldots, m_d)$ of $\module$, one can construct
an orthogonal family $(m_1^*, \ldots, m_d^*)$ such that the
flag of subspaces $(\oplus_{i=1}^k b_i\KK )_{1\leq k\leq d}$ is
preserved. This routine is exactly the same as for Euclidean lattices
and is given in \cref{alg:orthogonalize}, \algName{Orthogonalize}. We
present it here in its matrix form, which generalizes the so-called
$QR$-decomposition.

\begin{boxedAlgorithm}[algotitle={Orthogonalize},
  label=alg:orthogonalize]
  \begin{algorithm}[H]
    \Input{Basis $M\in\order_{\KK_h}^{d\times d}$ of an
    $\order_{\KK_h}-$module $\module$}
    \Output{$R$ part of the QR-decomposition of $M$}
    \BlankLine
    \For{$j=1$ \To $d$}
    {
      $Q_j \gets M_j - \sum_{i=1}^{j-1}
      \frac{\inner{M_j}{Q_i}}{\inner{Q_i}{Q_i}}Q_i$
    }
  \Return $R=\left(\frac{\inner{Q_i}{M_j}}{\|Q_i\|}\right)_{1\leq i<j \leq d}$\;
  \end{algorithm}
\end{boxedAlgorithm}

The volume of $S$ can be computed from the Gram-Schmidt vectors
collected in the matrix $R$ as:
$\Vol{\module} = \norm_{\KK/\QQ}\left(\prod_{i=1}^d R_{i,i}\right)$.

 \section{Reduction of low-rank $\order_\KK$-modules in cyclotomic fields}
\label{sec:fast_LLL_NF}

Let $h$ be a non-negative integer.  In the following of this section we
fix a tower of log-smooth conductor cyclotomic fields $\KK_h^\uparrow =
\left(\QQ = \KK_0 \subset \KK_1 \subset \cdots \subset \KK_h\right)$
and denote by $1 = n_0 < n_1 < \cdots < n_h$ their respective degrees
over $\QQ$. Then we consider a free module $\module$ of rank $d$ over
the upper field $\KK_h$, which is represented by a basis $(m_1,
\ldots, m_d)$ given as the columns of a matrix
$M\in\order_{\KK_h}^{d\times d}$. For notational simplicity, in this section,
we shall denote by $\langle a,b\rangle$ the $\order_\lL$-module
$a\order_\lL\oplus b\order_\lL$.

\subsection{In-depth description of the algorithm}
\subsubsection{Outer iteration}
\label{sec:algorithm_LL_outer_reduction}

To reduce the module $\module$ we adopt an iterative
strategy to progressively modify the basis:
for $\rho$ steps a reduction pass over the current basis is
performed, $\rho$ being a parameter whose value is computed
to optimize the complexity of the whole algorithm while still ensuring the
reduceness of the basis; we defer the precise computation of this
constant to \cref{sec:algorithm_complexity}. As in the \LLL~algorithm a
size-reduction operation is conducted to control the size of the
coefficients of the basis and ensure that the running time of the
reduction is polynomial. Note that for number fields this subroutine
needs to be adapted to deal with units of $\order_{\KK_h}$ when
rounding. The specificities of this size-reduction are the matter
of~\cref{sec:algorithm_LL_size_reduction}.

\subsubsection{Step reduction subroutine}
\label{sec:algorithm_LL_step_reduction}

We now take a look at the step reduction pass, once the size-reduction
has occurred. As observed in~\cref{sec:textbook_LLL}, the textbook
\LLL~algorithm epitomizes a natural idea: make the reduction process
boiling down to the treatment of rank two modules and more precisely
to iteratively reduce \emph{orthogonally projected} rank two modules.
We are using the same paradigm here and this step reduction pass over
the current basis is a sequence of reduction of projected rank 2
$\order_{\KK_h}-$modules.  However on the contrary to the
\LLL~algorithm we do not proceed progressively along the basis, but
reduce $\lfloor d/2\rfloor$ independent rank 2 modules at each
step.
This design enables an efficient parallel implementation which
reduces submodules simultaneously, in the same way that the classical
\LLL{} algorithm can be parallelized~\cite{ISSAC:Villard92,heckler1998complexity}.\medskip

Formally, given the basis of $\module$ collected in the matrix
$M$, let us denote by $r_j$ the vector $\left(R_{j,j},
R_{j+1,j}=0\right)$, and $r'_{j}$ the vector $\left(R_{j+1,j},
R_{j+1,j+1}\right)$ where $R$ is the $R$-part of the QR-decomposition
of $M$. The module $\mathcal{R}_{i}$ encodes exactly the
projection of $\module_i = \big\langle m_{i-1}, m_{i} \big\rangle$
over the orthogonal space to the first $i-1$ vectors $(m_1, \ldots,
m_{i-1})$. In order to recursively call the reduction algorithm on
$\mathcal{R}_i$ we need to \emph{descend} it to the subfield
$\KK_{h-1}$.

\subsubsection{Interlude: descending to cyclotomic subfields}
\label{sec:descent}
Remark now that since $\KK_h$ is a cyclotomic extension of
the cyclotomic field $\KK_{h-1}$, there exists a root of unity
$\xi$ such that
\[
  \order_{\KK_h} = \order_{\KK_{h-1}} \oplus \xi \order_{\KK_{h-1}}
  \oplus \cdots \oplus \xi^{q_h-1} \order_{\KK_{h-1}}.
\]
for $q_h=n_h/n_{h-1}$ being the relative degree of $\KK_h$ over
$\KK_{h-1}$. As a consequence, the module $\mathcal{R}_i$ decomposes over
$\order_{\KK_{h-1}}$ as:
\[
  \begin{aligned}
    \mathcal{R}_i   =~ &r_i\order_{\KK_h} \oplus r_{i+1}'\order_{\KK_h}
    \\
                    =~ &r_i\order_{\KK_{h-1}} \oplus \xi r_i
    \order_{\KK_{h-1}} \oplus \cdots \oplus \xi^{q_h-1} r_i
    \order_{\KK_{h-1}} \oplus \\
    &r'_{i+1}\order_{\KK_{h-1}} \oplus \xi
    r_{i+1}'
    \order_{\KK_{h-1}} \oplus \cdots \oplus \xi^{q_h-1} r'_{i+1}
    \order_{\KK_{h-1}},
  \end{aligned}
\]
yielding a basis of $\mathcal{R}_i$ viewed as a free
$\order_{\KK_{h-1}}$-module of rank $2\times q_h$.
This module can then recursively reduced, this time over a tower of
height $h-1$. This conversion from an $\order_{\KK_h}$-module to an
$\order_{\KK_{h-1}}$ module is referred as the function
\algName{Descend}. Conversely, any vector $u\in \order_{\KK_{h-1}}^{2q_h}$
can be seen with this decomposition as a vector of $\order_{\KK_h}^2$ by
grouping the coefficients as $\left(\sum_{i=1}^{q_h} u[i]\xi^i ,\sum_{i=1}^{q_h}
u[q_h+1+i]\xi^i\right)$. We denote by \algName{Ascend}~this conversion.

\subsubsection{Back on the step reduction}
\label{sec:two_by_two_reduction}

As mentioned in \cref{sec:algorithm_LL_step_reduction},
we start by reducing---with a recursive call after descending---all the
modules $\mathcal{R}_{2i} = \big\langle r_{2i-1}, r'_{2i} \big\rangle$
for $1\leq i \leq \lfloor d/2 \rfloor$, so that each of these
reductions yields a small element of the submodule $\module_{2i} =
\big\langle m_{2i-1}, m_{2i} \big\rangle$; which is then
\emph{completed}\footnote{The precise definition of this completion
and lifting is given in \cref{sec:algorithm_LL_lifting}.} in a
basis of $\module_{2i}$.   But on the contrary of the classical
\LLL~reduction, this sequence of pairwise independent reductions does
not make interact the elements $m_{2i}$ and $m_{2i+1}$, in the sense
that no reduction of the module projected from $\langle m_{2i},
m_{2i+1}\rangle$ is performed.  To do so, we  then perform the same
sequence of pairwise reductions but with all indices shifted by 1: we
reduce the planes $\big\langle r_{2i}, r'_{2i+1} \big\rangle$ for each
$1\leq i \leq \lfloor d/2 \rfloor$, as depicted in the following
diagram:
\medskip

\begin{center}
  \begin{tikzpicture}
    \tikzstyle{every path}=[thin]
    \edef\sizetape{0.7cm}
    \tikzstyle{tape}=[fill=black!5!white, draw, minimum size=\sizetape]
    \tikzstyle{tape1}=[fill=white, draw, minimum size=\sizetape]

    \begin{scope}[start chain=1 going right,node distance=+0.20mm]
      \node [on chain=1,tape] (anchor1) {$\,\,\,m_1\,\,\,$};
      \node [on chain=1,tape1] {$\,\,\,m_2\,\,\,$};
      \node [on chain=1,tape] {$\,\,\,m_3\,\,\,$};
      \node [on chain=1,tape1] {$\,\,\,m_4\,\,\,$};
      \node [on chain=1,tape,draw=none, fill=none] {$\ldots$};
      \node [on chain=1,tape] {$m_{i-1}$};
      \node [on chain=1,tape1] {$\,\,\,m_i\,\,\,$};
      \node [on chain=1,tape] {$m_{i+1}$};
      \node [on chain=1,tape,draw=none, fill=none] {$\ldots$};
      \node [on chain=1,tape1] {$m_{n-1}$};
      \node [on chain=1,tape] {$\,\,\,m_n\,\,\,$};
      \node [on chain=1,tape,draw=none, fill=none] {Basis};

      \node [on chain=1,tape, yshift=-1cm, xshift=-0.55cm] at (anchor1)
        {$\hspace{0.33165cm}\big\langle
        r_1\,,\,r'_{2}\big\rangle\hspace{0.3316cm}$};
      \node [on chain=1,tape1]
        {$\hspace{0.3317cm}\big\langle
        r_3\,,\,r'_{4}\big\rangle\hspace{0.3317cm}$};
      \node [on chain=1,tape,draw=none, fill=none] {$\ldots$};
      \node [on chain=1,tape]
        {$\hspace{1.60mm}\big\langle
        r_{i-1}\,,\,r'_{i}\big\rangle\hspace{1.58mm}$};
      \node [on chain=1,tape,draw=none, fill=none] {$\ldots$};
      \node [on chain=1,tape1]
        {$\hspace{-0.58mm}\big\langle
        r_{n-2}\,,\,r'_{n-1}\big\rangle\hspace{-.59mm}$};
      \node [on chain=1,tape,draw=none, fill=none]
        {\hspace{1.11cm}Odd steps};

      \node [on chain=1,tape, yshift=-2cm, xshift=+0.55cm] at (anchor1)
        {$\hspace{0.33165cm}\big\langle
        r_2\,,\,r'_{3}\big\rangle\hspace{0.3316cm}$};
      \node [on chain=1,tape1] {$\hspace{0.3317cm}\big\langle
        r_4\,,\,r'_{5}\big\rangle\hspace{0.3316cm}$};
      \node [on chain=1,tape,draw=none, fill=none] {$\ldots$};
      \node [on chain=1,tape]
        {$\hspace{1.58mm}\big\langle
        r_{i}\,,\,r'_{i+1}\big\rangle\hspace{1.60mm}$};
      \node [on chain=1,tape,draw=none, fill=none] {$\ldots$};
      \node [on chain=1,tape1]
        {$\hspace{1.38mm}\big\langle
        r_{n-1}\,,\,r'_{n}\big\rangle\hspace{1.39mm}$};
        \node [on chain=1,tape,draw=none, fill=none] {Even steps};
      \end{scope}
    \end{tikzpicture}
  \end{center}

  \subsubsection{Unit-size-reduction for $\order_{\KK_h}$-modules}
  \label{sec:algorithm_LL_size_reduction}

  As mentioned in~\cref{sec:algorithm_LL_outer_reduction} in order to
  adapt the size-reduction process to the module setting, one needs to
  adjust the rounding function. When $\KK_h=\QQ$, the rounding boils
  down to finding the closest element in $\order_\KK=\ZZ$, which is
  encompassed by the round function $\lceil \cdot \rfloor$. In the
  higher-dimensional context, we need to approximate any element of
  $\KK_h$ by \emph{a} close element of $\order_{\KK_h}$.

  Note that finding \emph{the} closest integral element is not
  efficiently doable. The naive approach to this problem consists in
  reducing the problem to the resolution of the closest integer problem
  in the Euclidean lattice of rank $n_h$ given by $\order_{\KK_h}$ under
  the Archimedean embedding. However, up to our knowledge, no
  exponential speedup exists using its particular structure compared to
  sieving or enumeration in this lattice.

  Nonetheless, finding a target \emph{close enough} to the target
  suffices for our application. As such we simply define the rounding of
  an element $\alpha\in\KK_h$ as the integral rounding on each of its
  coefficients when represented in the power base of $\KK_h$.

  We add here an important and necessary modification: before the actual
  size-reduction occurred, we compute a unit $u$
  using~\cref{thm:unites} close to $R_{i,i}$. This routine is denoted
  by \algName{Unit}. The vector $M_i$ is then divided by $u$. While not
  changing the algebraic norms of the elements, this technicality forces
  the Archimedean embeddings of the coefficients to be balanced and
  helps the reduced matrix to be well-conditioned. This avoids
  a blow-up of the precision required during the computation.  This
  modified size-reduction is fully described
  in \cref{alg:size_reduce}, \algName{Size-Reduce}.

  \begin{boxedAlgorithm}[algotitle={Size-Reduce}, label=alg:size_reduce]
    \begin{algorithm}[H]
      \Input{$R$-factor of the QR-decomposition of
      $M\in\order_{\KK_h}^{d\times d}$}
      \Output{A unimodular transformation $U$ representing the
      size-reduced basis obtained from $M$.}
      \BlankLine
      $U \gets \id_{d,d}$\\
      \For{$i=1$ \To d}
      {
        $D \gets D_i(\textrm{\algName{Unit}}(R_{i,i}))$\tcp{$D_i$ is
        a dilation matrix}
        $(U,R)\gets (U,R)\cdot D^{-1}$ \\
        \For{$j=i-1$ \Downto $1$} {
          $\sum_{\ell=0}^{n-1} r_\ell X^\ell \gets R_{i,j}$
          \tcp{\footnotesize Extraction as a polynomial}
      $\mu \gets \sum_{\ell=0}^{n-1} \lfloor r_\ell\rceil X^\ell$
          \tcp{\footnotesize Approximate rounding of $R_{i,j}$ in
          $\order_{\KK_h}$}
          $(U,R) \gets (U,R)\cdot T_{i,j}(-\mu)$ \tcp{$T_{i,j}$ is a
          shear matrix}
        }
      }
      \Return $U$
    \end{algorithm}
  \end{boxedAlgorithm}

  \subsubsection{Reduction of the leaves}
  \label{sec:algorithm_LL_leaves}

  As the recursive calls descend along the tower of number fields,
  the bottom of the recursion tree requires reducing $\order_{\KK_0} (=
  \order_\QQ = \ZZ)$-modules, that is Euclidean lattices.
  As a consequence, the step reduction performs calls to a reduction
  oracle for plane Euclidean lattices. For the sake of efficiency we
  adapt
  Sch\"onhage's algorithm~\cite{Schonhage91} to reduce these
  lattices, which is faster than the traditional Gauss' reduction.
  This algorithm is an extension to the bidimensional case of the
  half-GCD algorithm, in the same way, that Gauss' algorithm can be seen
  as a bidimensional generalization of the classical GCD computation.

  The original algorithm of Sch\"onhage
  only deals with the reduction of binary quadratic forms, but can be
  straightforwardly adapted to reduce rank 2 Euclidean lattices, and to
  return the corresponding unimodular transformation matrix. In all of
  the following, we denote by $\algName{Schonhage}$ this modified
  procedure.

  \subsubsection{The lifting phase}
  \label{sec:algorithm_LL_lifting}

  As explained in~\cref{sec:algorithm_LL_step_reduction}, we
  recursively call the reduction procedure to reduce the descent of
  projected modules of rank $2$ of the form $\mathcal{R}_i = \langle r_i,
  r'_{i+1}\rangle$, over $\KK_{h-1}$, yielding a unimodular
  transformation $U'\in\order_{\KK_{h-1}}^{2q_h\times 2q_h}$ where $q_h$
  is the relative degree of $\KK_h$ over $\KK_{h-1}$.

  From $U'$, we can find random short elements in the module by computing a
  small linear combination of the first columns.
  Applying \algName{Ascend}, we deduce some short $x=m_ia+m_{i+1}b$.
  But then to replace $m_i$ by
  $x$ in the current basis, we need to complete this vector into a basis
  $(x, y)$ of $\module_i$ over $\order_{\KK_h}$.  Doing so boils down to
  complete a vector of $\order_{\KK_h}^2$ into a unimodular
  transformation. Indeed, suppose that such a vector $y$ is found and
  denote by $(a, b)$ and $(v, u)$ the respective coordinates of $x$ and
  $y$ in the basis $(m_i, m_{i+1})$. By preservation of the volume we
  have without loss of generality:
  \[
    1 = \det \begin{pmatrix}
      a& v\\
      b& u\\
    \end{pmatrix} = au-bv.
  \]

  Therefore finding the element $y$ to complete $x$ reduces to
  solving the B\'ezout equation in the unknown $u$ and $v$
  \begin{equation}\label{eq:bezout}
    au - bv = 1
  \end{equation}
  over the ring $\order_{\KK_h}$. Since this ring is in general not
  Euclidean we can not apply directly the Euclidean algorithm to solve
  this equation as an instance of the extended \GCD{} problem. However, we
  can use the algebraic structure of the tower $\KK_h^\uparrow$ to
  recursively reduce the problem to the rational integers. This
  \emph{generalized} Euclidean algorithm works as follows:

  \begin{description}
    \item[If $\KK_h = \QQ$] then the problem is an instance of
      extended GCD search, which can be solved efficiently by
      the binary-GCD algorithm.
    \item[If the tower $\KK_h^\uparrow$ is not trivial]
      we make use of the structure of $\KK_h^\uparrow$ and first
      descend the
      problem to the subfield $\KK_{h-1}$
      by computing the relative norm $\norm_{\KK_h/\KK_{h-1}}$ of
      the elements $a$ and $b$; then by recursively calling
      the algorithm on these elements $\norm_{\KK_h/\KK_{h-1}}(a)$
      and $\norm_{\KK_h/\KK_{h-1}}(b)$, we get two algebraic
      integers $\mu$ and $\nu$ of $\order_{\KK_{h-1}}$ fulfilling the
      equation:
      \begin{equation}\label{eq:algorithm_extend_euclide_rec}
        \mu\norm_{\KK_h/\KK_{h-1}}(a)-\nu
        \norm_{\KK_h/\KK_{h-1}}(b) = 1.
      \end{equation}
      But then remark that for any element $\alpha\in\order_{\KK_h}$
      we have, using the comatrix formula and the definition of the
      norm as a determinant that: $\norm_{\KK_h/\KK_{h-1}}(\alpha)\in
      \alpha\order_{\KK_h}$, so that
      $\alpha^{-1}\norm_{\KK_h/\KK_{h-1}}(\alpha)\in \order_{\KK_h}$.
      Then, from~\cref{eq:algorithm_extend_euclide_rec}:
      \[
        a\cdot
        \underbrace{\mu\,a^{-1}\norm_{\KK_h/\KK_{h-1}}(a)}_{:=u\in\order_{\KK_h}}-
        b\cdot \underbrace{
        \nu\,b^{-1}\norm_{\KK_h/\KK_{h-1}}(b)}_{:=v\in\order_{\KK_h}}
        = 1,
      \]
      as desired.
    \item[Reduction of the size of solutions]
      The elements $u,v$ found by the algorithm are not necessarily
      the smallest possible elements satisfying~\cref{eq:bezout}.
      To avoid a blow-up in the size of the coefficients lifted, we do need to control the size of the solution at each
      step. Since the function \algName{Size-Reduce}~preserves the
      determinant by construction and
      reduces the norm of the coefficients, we can use it to reduce the
      bitsize of $u,v$ to (roughly) the bitsize of $a$ and $b$.
  \end{description}

  The translation of this method in pseudocode is given
  in~\cref{alg:generalized_euclide}, \algName{G-Euclide}.

  \begin{boxedAlgorithm}[algotitle={G-Euclide, Lift},
    label=alg:generalized_euclide]
    \begin{algorithm}[H]
      \BlankLine
      \SetKwProg{Fn}{Function}{:}{}
      \Fn{\algName{G-Euclide}}{
        \Input{Tower of number fields $\KK^{\uparrow}_h$,
        $a, b\in\KK_h$.}
        \Output{$u,v\in \KK_h$, such that $au+bv=1$}
        \BlankLine
        \lIf{$\KK_h = \QQ$}{\Return $\algName{ExGcd}(a,b)$}
        $\mu, \nu \gets
        \algName{G-Euclide}\left(\KK_{h-1}^\uparrow,
          \norm_{\KK_h/\KK_{h-1}}(a),
        \norm_{\KK_h/\KK_{h-1}}(b)\right)$\;
        $\mu', \nu' \gets \mu\,a^{-1}\norm_{\KK_h/\KK_{h-1}}(a),
        \nu\,b^{-1}\norm_{\KK_h/\KK_{h-1}}(b) $\;
        $W \gets \begin{pmatrix}
          a & \nu'\\
          b  & \mu'
        \end{pmatrix}$\;
        $V \gets \algName{Size-Reduce} \left(
        \algName{Orthogonalize}(W) \right)$\;
        \Return $W\cdot V[2]$
      }
      \BlankLine
      \BlankLine
      \Fn{\algName{Lift}}{
        \Input{Tower of number fields $\KK^{\uparrow}_h$, unimodular
        matrix $U'\in\order_{\KK_{h-1}}^{2q_h}$}
        \Output{Unimodular matrix $U\in\order_{\KK_{h}}^{2\times 2}$}
        \BlankLine
        $a,b \gets$\algName{Ascend}$(\KK_h, U[1])$\;
        $\mu, \nu \gets \algName{G-Euclide}\left(\KK_{h-1}^\uparrow,
        a,b \right)$\;
        $U \gets \begin{pmatrix}
          a & \nu\\
          b & \mu
        \end{pmatrix}$\;
        \Return $U$
      }
    \end{algorithm}
  \end{boxedAlgorithm}

  The number of bits needed to represent the relative norms does not
  depend on the subfield, and the size-reduction forces the output
  vector to have the same bitsize as the input one. This remark is
  the crux of the quasilinearity of the \algName{G-Euclide}, as
  stated in \cref{lem:Euclid_complexity}.

  Remark that the algorithm needs $\norm_{\KK_h/\QQ}(a)$ to be prime with
  $\norm_{\KK_h/\QQ}(b)$.
  We assume that we can always find quickly such $a,b$ with a short $x$.
  This will lead to~\cref{heur:lifting_size}, and the validity of
  this assumption is discussed in~\cref{sec:lifting_problems}.

  \subsection{Wrapping-up}

  The full outline of the reduction is given in \cref{alg:fast_lll}
  and a schematic overview of the recursive steps is provided in
  the diagram of \cref{fig:scheme}.

  \begin{boxedAlgorithm}[algotitle=Reduce, label=alg:fast_lll]
    \begin{algorithm}[H]
      \Input{Tower of cyclotomic fields $\KK_h^\uparrow$,
        Basis $M\in\order_{\KK_h}^{d\times d}$ of the $\order_{\KK_h}-$module
      $\module$}
      \Output{A unimodular transformation
        $U\in\order_{\KK_h}^{d\times d}$
      representing a reduced basis of $\module$.}
      \BlankLine
      \lIf{$d=2$ \And $\KK_h=\QQ$} { \Return \algName{~Schonhage}$(M)$ }
      \For{$i=1$ \To $\rho$}
      {
        $R \gets $ \algName{Orthogonalize}$(M)$\;
        $U_i \gets $ \algName{Size-Reduce}$(R)$\;
        $(M,R) \gets (M,R)\cdot U_i$ \;
        \For{$j= 1+(i \mod 2)$ \To $d$ \Bystep 2}
        {
	  \If{$\norm_{\KK_h/\QQ}(R_{j,j}) \leq 2^{2(1+\varepsilon)\alpha n_h^2} \norm_{\KK_h/\QQ}(R_{j+1,j+1})$	}
	    {
          $M'\gets $ \algName{Descend}$(\KK_{h-1}^\uparrow, R[j:j+1,
          j:j+1])$ \;
          $U' \gets $ \algName{Reduce}$(\KK_{h-1}^\uparrow, M')$ \;
          $(U_i, M) \gets (U_i, M)\cdot \textrm{\algName{Lift}}(U')$ \;
	    }
        }
      }
      \Return $\prod_{i=1}^\rho U_i$
    \end{algorithm}
  \end{boxedAlgorithm}

  \makeatletter
  \newcommand*{\centerfloat}{
    \parindent \z@
    \leftskip \z@ \@plus 1fil \@minus \textwidth
    \rightskip\leftskip
  \parfillskip \z@skip}
  \makeatother
  \begin{landscape}
  \begin{figure}
    \centerfloat
    \begin{tikzpicture}[baseline= (a).base]
      \node[scale=.7] (a) at (0,0){
          \begin{tikzcd}
            \KK_h &  & \order_{\KK_h} \arrow[ll, hook'] & M =
            m_1\order_{\KK_h}\oplus \cdots\oplus m_d\order_{\KK_h}
            \arrow[rrr, "\textrm{reduce rk 2 projected modules}" description, harpoon, dashed]
            &  &  & \alpha_1\order_{\KK_h}\oplus\alpha_2\order_{\KK_h} {}
            \arrow[lllddd, "\textrm{descend}" description, two heads, dashed,
            hook', bend right] \arrow[rr, "\cong" description, no head] &  &
            \upsilon\order_{\KK_h}\oplus\omega\order_{\KK_h} \\
            &  &  &  &  &  &  & \upsilon\in M  \textrm{ (short)}  \arrow[ru,
            "\textrm{Lift}" description, dashed, bend right] &  \\
                                                             &  &  &  &  &  &  &  &  \\
            \KK_{h-1} \arrow[uuu, hook'] &  & \order_{\KK_{h-1}}
            \arrow[uuu, "r"', hook'] \arrow[ll, hook'] & M' =
            m'_1\order_{\KK_{h-1}}\oplus \cdots\oplus
            m'_{2m_h}\order_{\KK_{h-1}} \arrow[rrr, "\textrm{recurse to rank 2}"
            description, no head, dashed] \arrow[rrrruu] &  &  &
            \alpha_1'\order_{\KK_{h-1}}\oplus\alpha_2'\order_{\KK_{h-1}} \arrow[llld,
            "\textrm{descend}" description, two heads, dotted, hook'] &  &
            \upsilon'\order_{\KK_{h-1}}\oplus\omega'\order_{\KK_{h-1}} \arrow[ll,
            "\cong" description, no head] \\
            \vdots \arrow[u, hook'] &  & \vdots \arrow[u, hook'] & \vdots
            \arrow[rrr, no head, dotted] &  &  & \vdots \arrow[llldd,
            "\textrm{descend}" description, dotted, hook', bend right] &
            \upsilon'\in M'  \textrm{ (short)}  \arrow[ru, "\textrm{Lift}"
            description, dashed, bend right] &  \\
                                             &  &  &  &  &  &  &  &  \\
            \mathbf{Q} \arrow[uu, hook'] &  & \mathbf{Z} \arrow[ll, hook']
            \arrow[uu, hook'] & M_\mathbf{Z} = \beta_1\mathbf{Z}\oplus
            \cdots\oplus\beta_k\mathbf{Z} \arrow[rrr, "\textrm{recurse to rank 2}"
            description, no head, dashed] \arrow[rrrruu, dashed] &  &  &
            \alpha_1\mathbf{Z}\oplus\alpha_2\mathbf{Z} \arrow[d,
            "\textrm{\textsc{schonhage}}"
            description, dashed] &  & \upsilon\mathbf{Z}\oplus\omega\mathbf{Z}
            \textrm{ (reduced)} \arrow[ll, "\cong" description, no head] \\
            &  &  &  &  &  & \upsilon \textrm{ (short)} \arrow[rru,
            "\textrm{complete}" description, dashed, bend right] &  &
        \end{tikzcd}};
    \end{tikzpicture}
    \caption{Schematic view of the recursive call of reductions.}
    \label{fig:scheme}
  \end{figure}
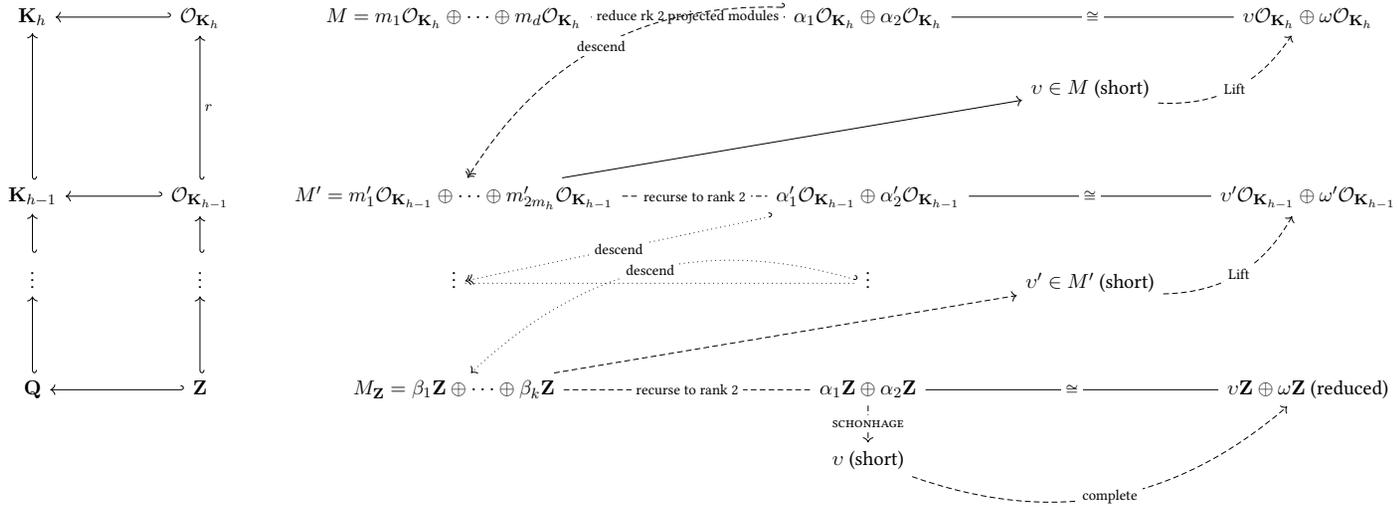
  \end{landscape}

 \section{Complexity analysis}
\label{sec:algorithm_complexity}

In this section, we devise the complexity of the~\cref{alg:fast_lll}
and of its approximation factor. More formally we prove the following
theorem:

\begin{theorem}\label{thm:complexity_fast_LLL}
  Let $f$ be a log-smooth integer. The complexity of
  the algorithm \algName{Reduce}{} on rank two modules over
  $\KK=\QQ[x]/\Phi_f(x)$, represented as a matrix $M$ whose
  number of bits in the input coefficients is uniformly bounded by
  $B>n$, is heuristically a $\bigOtilde{n^2B}$ with
  $n=\varphi(f)$. The first column of the reduced matrix
  has its coefficients uniformly bounded by
  $2^{\bigOtilde{n}}\left(\covol M\right)^{\frac{1}{2n}}$.
\end{theorem}

\subsection{Setting}
Let $h>0$ be a non-negative integer. In the following of this section we
fix a tower of cyclotomic fields $\KK_h^\uparrow = \left(\QQ
= \KK_0 \subset \KK_1 \subset \cdots \subset \KK_h\right)$ with
log-smooth conductors and denote by
$1 = n_0 < n_1 < \cdots < n_h$ their respective degrees over $\QQ$. We
consider a free module $\module$ of rank $d$ over the upper field
$\KK_h$, given by one of its basis, which is represented
as a matrix $M\in\order_{\KK_h}^{d\times d}$. In all of the following,
for any matrix $A$ with coefficients in $\KK_h$ we denote by $\|A\|$ the
2-norm for matrices.

We aim at studying the behavior of the reduction process given in
\cref{alg:fast_lll} on the module $\module$; as such we denote
generically by $X^{(\tau)}$ the value taken by any variable $X$
appearing in the algorithm at \emph{the beginning} of the step $i=\tau$, for
$1\leq\tau\leq \rho+1$. For instance $R^{(1)}$ denotes the $R$-part
of the orthogonalization of $M$ and $M^{(\rho+1)}$ represents the
reduced basis at the end of the algorithm.

Since the implementation of the algorithm is done using floating-point
arithmetic, we need to set a precision which is sufficient to handle
the internal values during the computation.
To do so we set:
\[
  p=\log\frac{\max_{\sigma:\KK_h\rightarrow\CC, R_{i,i}\in R}
  \sigma (R_{i,i})}{\min_{\sigma:\KK_h\rightarrow\CC, R_{i,i}\in R} \sigma (R_{i,i})},
\]
where the $\sigma$ runs over the possible field embeddings and the $R_{i,i}$
are the diagonal values of the $R$ part of the $QR$-decomposition of the input
matrix of the reduction procedure. We will prove as a byproduct of the
complexity analysis that taking a precision of $\bigO{p}$ suffices.

For technical reasons which will appear in the subsequent
proofs, we introduce a constant $\alpha>0$ which will be optimized at
the end of our analysis. It essentially encodes the approximation factor
of the reduction.
Eventually, we set the variable $\varepsilon$ to be equal to $1/2$. This
apparently odd choice allows us to state our theorems with sufficient
generality to reuse them in the enhanced proof of the reduction
algorithm with symplectic symmetries, as detailed in
\cref{sec:symplectic}, with a different value.

The whole set of notations used in the analysis is recalled in
\cref{table:notations}.
\begin{table}[]
  \small
  \sffamily
  \setlength{\tabcolsep}{6pt}
  \vspace{4pt}
  \centering
  \renewcommand*{\arraystretch}{1.5}
  \begin{tabular}{cc}
    \midrule
    \rowcolor{black!10}
    $h$ & Height of the tower \\
    $n_h$ & Absolute height $[\KK_h:\QQ]$ \\
    \rowcolor{black!10}
    $p$ & bound on the precision used by the reduction\\
    $\varepsilon$ & 1/2 \\
    \rowcolor{black!10}
    $i$ & Current outmost loop number ($1\leq i\leq \rho$)
    iteration \\
    $\alpha$ & Constant to be optimized\\
    \bottomrule
  \end{tabular}
  \vspace{2em}
  \caption{Notations used in the complexity analysis.
    $p$ is of course set to be larger than the bitsize
  of the input matrix.}
  \label{table:notations}
\end{table}

\subsection{Overview of the proof}
Before going into the details of the proof, we lay its blueprint.
We start by estimating the approximation factor of the reduction and
deduce a bound in $\bigO{d^2\log p}$ on the number of rounds $\rho$
required to achieve the reduction the module $\module$, where $p$ is
the precision needed to handle the full computation. We then prove that
the limiting factor for the precision is to be sufficiently large to
represent the shortest Archimedean embedding of the norm of the
Gram-Schmidt orthogonalization of the initial basis. We then
devise a bound by looking at the sum of all the bit sizes used in the
recursive calls and concludes on the complexity.
The critical part of the proof is to use the potential to show that
dividing the degrees by $\frac{d}{2}$ leads to a multiplication by a
factor at most in $\bigO{d^2}$ of the sum of all the precisions in the
recursive calls, instead of the obvious $\bigO{d^3\log p}$.

\subsection{A bound on the number of rounds and the approximation factor
of the reduction}

We define here a set of tools to study the approximation factor
of the reduction, by approximating it by an iterative linear operator
on the family of volumes of the submodules $\module_i = m_1\ZZ\oplus
\cdots\oplus m_i\ZZ$ for $1\leq i \leq d$. This method is quite similar to the
one used by Hanrot \emph{et al.}\ in~\cite{C:HanPujSte11} to analyze the
\BKZ{} algorithm by studying a dynamical system.

To ease the computation of the number of rounds,
we can without loss of generality, scale the input
matrix and suppose that:
\[
\covol\module = |\norm_{\KK_h/\QQ}(\det M)|^{\frac1d}=
2^{-(d+1)(1+\varepsilon)\alpha  n_h^2}.
\]
We only do so for this subsection.

\subsubsection{Potential and volumes of flags}

A \emph{global} measure of reduceness of a Euclidean lattice is its
potential. An $\order_{\KK_h}$-analog of this constant can be defined in
a similar manner by using the algebraic norm to replace the Euclidean
norm over $\RR^n$.

\begin{definition}[Potential]
\label{potential}
  Let $(m_1, \ldots, m_d)$ be a basis of the module $\module$ given as
  the columns of a matrix $M\in\order_{\KK_h}^{d\times d}$, and let
  $R$ be the $R$-part of its QR-decomposition.
  Its \emph{log-potential} is defined as:
  \[ \Pi(M) = \sum_{i=1}^d \log \covol{\module_i} = \sum_{i=1}^d
  (d-i)\log\norm_{\KK_h/\QQ}\left(R_{i,i}\right).\]
\end{definition}

As in the Euclidean case, a \emph{local} tool to analyze the evolution
of a basis $(m_1, \ldots, m_d)$ of a lattice $\Lat$, through a
reduction, is the \emph{profile} of the volumes associated with the flag
of a basis, namely the family:
\[
  \covol \left(\module_1\right), \ldots,
  \covol\left(\module_i\right), \ldots,
  \covol\Lat.
\]
As for the potential, we define the profile of the flag in a similar
way with the algebraic norm on $\KK_h$, but for technical reasons,
we quadratically twist it with the constant $\alpha>0$.

\begin{definition}[Flag profile]
  \label{def:flag_profile}
  Let $(m_1, \ldots, m_d)$ be a basis of the module $\module$ given as
  the columns of a matrix $M\in\order_{\KK_h}^{d\times d}$, and let
  $R$ be the $R$-part of its QR-decomposition.
  Its \emph{profile} is the vector $\profile(M) \in \RR^d$ defined by:
  \[
    \profile(M)_j=\sum_{k = 1}^j
    \left(\log\norm_{\KK_h/\QQ}\left(R_{k,k}\right)+
	2k(1+\varepsilon)\alpha n_h^2\right), \qquad \textrm{for~} 1\leq j \leq
    d.
  \]
\end{definition}

The following lemma gives an estimate of the norm of the
profile in terms of the parameters of the algorithm and
of the input bitsize.

\begin{lemma}
  \label{lem:bound_omega}
  With the same notations as in~\cref{def:flag_profile}, we have:
  \[
    \|\profile(M)\|_2 \leq
  (2+\varepsilon) \alpha d^{2}n_hp
  \]
\end{lemma}
\begin{proof}
  We have $|\norm_{\KK/\QQ}(\det(M))|\leq 1$ so for each $i$, and each embedding
  $\sigma$, we have that $|\sigma(R_{i,i})|\leq 2^p$.
  Now we compute:
  \[
    \begin{aligned}
      \frac{\|\profile(M)\|_2^2}{d} &\leq \max_{j=1,\ldots, d-1}
      \left\{\sum_{k\leq j}
        \left(\log\norm_{\KK_h/\QQ}\left(R_{k,k}\right)+
      2k(1+\varepsilon)\alpha n_h^2\right) \right\}\\
      &\leq dn_hp + d(d-1)(1+\varepsilon)\alpha n_h^2\\
\end{aligned}
 \]
 which implies the result.
\end{proof}

\subsubsection{A family of step operators}

To study the reduction steps, we define the following linear operators
\begin{equation}
  \delta_j: \left|
  \begin{array}{rcl}
    \RR^d & \longrightarrow & \RR^d \\
    v & \longmapsto & (w_\ell)_\ell =
    \left\{\begin{array}{cl}
        \frac{v_{j-1}+v_{j+1}}{2} & \textrm{if~} \ell = j\\
        v_{j} & \textrm{if~} \ell = j+1\\
        v_\ell & \textrm{else} \\
    \end{array}\right.
  \end{array}
\right.,
 \end{equation}
 for each $1\leq j\leq d$.
 These operators provide an upper bound on the profile of a
 basis after a reduction at index $j$. To encode the behavior
 of a full round of reduction we define the operators:
 \[
   \Delta_o = \prod_{i=1\,|\, i \textrm{~odd}} \delta_i, \quad
   \textrm{and}
   \quad
   \Delta_e = \prod_{i=2\,|\, i \textrm{~even}} \delta_i,
 \]
 to define inductively the sequence:
 \[
   \begin{aligned}
     \profile^{(1)} &= \profile(M^{(1)})\\
     \profile^{(i)} &= \Delta_o \left(\profile^{(i-1)}\right)
     \quad\textrm{~if~} i = 0 \pmod 2
     \quad \textrm{~else~}\quad \Delta_e \left(\profile^{(i-1)}\right)
   \end{aligned}
 \]

 \begin{remark}
   \label{rem:increase_alpha}
   By the constraint we set on the volume of $\module$ to be equal to
   $2^{-d(d+1)(1+\varepsilon)\alpha n_h^2}$, we have for all $1\leq
   i\leq \rho$, that $\profile^{(i)}_d=0$.
 \end{remark}

 \begin{proposition}[Exponential decay of $\|\profile^{(i)}\|_2$]
   \label{prop:bound_spectral_omega}
   For all odd $i$, we have,
   \[
      \left|\profile^{(i)}_1\right|\leq
  \me^{-\frac{\pi^2(i-1)}{2d^2}}\|\profile^{(1)}\|_2
   \]
   and
   \[
      \|\profile^{(i+1)}\|_2\leq
  2\me^{-\frac{\pi^2(i-1)}{2d^2}}\|\profile^{(1)}\|_2.
   \]
 \end{proposition}
 \begin{proof}
   Note that $\Delta_o\circ\Delta_e$ depends only on the
   odd coordinates, so let $\Delta$ be its restriction to them in the domain
   and codomain.
   Remark that for all $1\leq k \leq \lceil \frac{d-1}{2} \rceil $ the vector
 \[ \left(\sin\left(\frac{(2j-1)k\pi}{2\lfloor d/2 \rfloor}\right)\right)_j \] is
   an
   eigenvector of $\Delta$ of associated eigenvalue
   $\cos\left(\frac{k\pi}{2\lfloor d/2 \rfloor}\right)^2$. A direct
   computation ensures that the eigenvectors are orthogonal. Since $2\lfloor
   d/2\rfloor \leq d$, we use the trivial
   bound $\left|\cos\left(\frac{k\pi}{d}\right)\right|\leq
   \cos\left(\frac{\pi}{d}\right)$ in addition to
   the convexity bound \[\ln(\cos(\pi/d))<-\frac{\pi^2}{2d^2}\]
   to obtain:
   \[ \sum_{k=1\textrm{~odd}} \left(\profile^{(i)}\right)_k^2\leq
   \me^{-\frac{\pi^2(i-1)}{2d^2}}\|\profile^{(1)}\|_2^2. \]
   This implies the first statement and
   \[ \sum_{k=2\textrm{~even}} \left(\profile^{(i+1)}\right)_k^2\leq
   \sum_{k=1\textrm{~odd}} \left(\profile^{(i)}\right)_k^2 \]
   implies the second.
 \end{proof}
 \begin{remark}[A ``physical'' interpretation of $\Delta$] The operator
   $\Delta$ introduced in the proof of
   \cref{prop:bound_spectral_omega} acts as a \emph{discretized}
   Laplacian operator on the discrete space indexed by $\{1,\ldots,
   d\}$, for a metric where two consecutive integers are at distance
   $1$. Then, the action of $\Delta$ through the iterations $1\leq
   i\leq \rho$ are reminiscent of the diffusion property of the
   solution of the \emph{heat equation} ($\frac{\partial u}{\partial t}=
   \alpha\Delta u$), whose characteristic time is quadratic
   in the diameter of the space.
 \end{remark}

 \subsubsection{A computational heuristic}
 \label{sec:heuristic}
 We now relate the behavior of the sequences of $\profile$ to the
 values taken by $R^{(i)}$. In order to do so, we introduce a
 computational heuristic on the behavior of the \algName{Lift}
 function, asserting that the lifting phase does not blow up
 the size of the reduced vectors.

 \begin{heuristic}[Size of lifting]
   \label{heur:lifting_size}
   For any $1\leq i\leq\rho$ and any $1\leq j \leq d$ where a call to
   \algName{Lift} happened:
   \[
     \norm_{\KK_h/\QQ}\left(R^{(i+1)}_{j,j}\right) \leq
     \min\left(2^{\alpha n_h^2}
       \sqrt{\norm_{\KK_h/\QQ}
       \left(R^{(i)}_{j,j}R^{(i)}_{j+1,j+1}\right)},
     \norm_{\KK_h/\QQ}\left(R^{(i)}_{j,j}\right)\right).
   \]
 \end{heuristic}

 A discussion on the validity of this heuristic is done in
 \cref{sec:lifting_problems}. However, we recall that we {\em do not} perform a
 local reduction if the following condition is fulfilled, up to the
 approximation error due to the representation at finite
 precision\footnote{More precisely, if the precision used when
   performing this testing is $p$, then if we are certain that \[
   \norm_{\KK_h/\QQ}(R_{j,j}^{(i)})\geq \min\left(2^{(1+\varepsilon)\alpha n_h^2}
       \sqrt{\norm_{\KK_h/\QQ}
       \left(R^{(i)}_{j,j}R^{(i)}_{j+1,j+1}\right)},
     \norm_{\KK_h/\QQ}\left(R^{(i)}_{j,j}\right)\right)
,\] no local reduction is called,
   else we have \[
   \norm_{\KK_h/\QQ}(R_{j,j}^{(i)})\geq
       \min\left(2^{(1+\varepsilon)\alpha n_h^2}
       \sqrt{\norm_{\KK_h/\QQ}
       \left(R^{(i)}_{j,j}R^{(i)}_{j+1,j+1}\right)},
     \norm_{\KK_h/\QQ}\left(R^{(i)}_{j,j}\right)\right) (1-2^{-\Omega(p)})
   \] and a recursive local reduction is called, the multiplicative error
 term coming from the approximation error committed by the approximation
 of the values $R_{*,*}$ at precision $p$.}:
 \[
   \norm_{\KK_h/\QQ}(R^{(i)}_{j,j}) \leq
   \min\left(2^{(1+\varepsilon)\alpha n_h^2}
       \sqrt{\norm_{\KK_h/\QQ}
       \left(R^{(i)}_{j,j}R^{(i)}_{j+1,j+1}\right)},
     \norm_{\KK_h/\QQ}\left(R^{(i)}_{j,j}\right)\right).
 \]

 From \cref{heur:lifting_size} we can show by a direct induction on
 $i$ that the sequence of $\profile^{(i)}$ is an over-approximation of
 the flag profile at step $i$. More precisely we have:

 \begin{lemma}
   \label{lem:inequality_omega_kappa}
   Under \cref{heur:lifting_size}, for any $1\leq i\leq\rho$:
   \[
 \profile\left(M^{(i)}\right) \leq  \profile^{(i)},
   \]
   where the comparison on vectors is taken coefficient-wise.
 \end{lemma}

 \subsubsection{A bound on the approximation factor and number of rounds}
 We can now conclude this paragraph by giving a quasiquadratic bound on the
 number of rounds:

 \begin{theorem}
   \label{thm:approximation_factor}
   Assuming that $\rho$ is even and $\rho >
	 \frac{2d^2}{\pi^2}\ln((2+\epsilon)\alpha d^2n_hp)$, we have that
   \[ \norm_{\KK_h/\QQ}(R^{(\rho+1)}_{1,1})\leq
     2^{(d-1)(1+\varepsilon)\alpha
   n_h^2+1}|\norm_{\KK_h/\QQ}(\det M)|^\frac1d.\]
 \end{theorem}
 \begin{proof}
   By taking the exponential of both sides of the inequality
   of~\cref{lem:inequality_omega_kappa}, we have:
   \[
     \norm_{\KK_h/\QQ}(R^{(\rho+1)}_{1,1})
     \leq 2^{\profile^{(\rho+1)}_1-2(1+\varepsilon)\alpha}.
   \]
   Recall that we forced $\left|\norm_{\KK_h/\QQ}(\det
   M)\right|^{\frac1d}=2^{-(d+1)\alpha n_h^2(1+\varepsilon)}$, so that:
   \[
     \norm_{\KK_h/\QQ}(R^{(\rho+1)}_{1,1})
     \leq 2^{(d-1)(1+\varepsilon)\alpha n_h^2+\profile^{(\rho+1)}_1}
     \left|\norm_{\KK_h/\QQ}(\det
     M)\right|^{\frac1d}.
   \]
   By~\cref{prop:bound_spectral_omega}, we know that
   $\profile^{(\rho+1)}_1\leq
   e^{-\frac{\pi^2\rho}{2d^2}}\|\profile^{(1)}\|_2.$
   Since we have:
   \[
     \begin{aligned}
       \ln|\profile^{(\rho+1)}_1|
	 \leq & \ln\|\profile^{(1)}\|_2
       -\frac{\rho\pi^2}{2d^2}\\
	     \leq& \ln((2+\epsilon)\alpha d^2n_hp)
       -\frac{\rho\pi^2}{2d^2} \leq & 0,
     \end{aligned}
   \]
   using~\cref{lem:bound_omega} and the hypothesis on $\rho$
   together with the fact that $d>1$. All in all
   $|\profile^{(\rho)}_1|\leq 1$ and which entails the desired
   inequality.
 \end{proof}

 With mild assumptions on the relative size of the parameters $\alpha,
 n_h, d$ and $p$ we have the following rewriting of
 \cref{thm:approximation_factor}.
 \begin{corollary}
   \label{cor:approximation_factor}
   Suppose that $\alpha = {\log^{\bigO{1}}(n_h)}$ and that
   $p>n_h+d$, then taking
   $\rho = \bigO{d^2\log(p)}$ is sufficient to reduce the module
   $\module$ and such that the algebraic norm of the first vector is
   bounded by a
 \[2^{\bigOtilde{dn_h^2}}|\norm_{\KK_h/\QQ}(\det
   M)|^\frac1d.\]
 \end{corollary}

 \begin{remark}
	 If the caller makes a similar heuristic with a $\alpha'$, then
	 we need $\alpha'>\alpha\cdot 2(1+\epsilon)\frac{d-1}{d}$ and any
	 such value is plausible for large $n_h$.
 \end{remark}

 \subsection{Time complexity of the toplevel reduction}
 Now that we have an estimate of the number of rounds, we can
 aim at bounding the complexity of each round, without counting the
 recursive calls, in a first time. To do so we will look independently
 at each of the part of a round, namely at the complexity of
 \algName{Orthogonalize}, \algName{Reduce} and \algName{Lift}. Since
 the lifting algorithm performs a size-reduction, we first give a
 fine-grained look at the \algName{Size-Reduce} function.

 \subsubsection{Complexity and quality of \algName{Size-Reduce}}

 The quantitative behavior of the \algName{Size-Reduce}{} procedure is
 encoded by the following theorem, given in all generality for
 arbitrary matrices over a cyclotomic field.
 \begin{theorem}
   \label{thm:size_reduce}
   Let $A$ be a matrix of dimension $d$ whose coefficients lie in the
   cyclotomic field $\KK=\QQ[\zeta_f]$, and $n=\varphi(f)$.
   We are given a non-negative integer $p>0$, where
   $\|A\|,\|A^{-1}\|\leq 2^p$ and such that $\sqrt{n\log n\log\log n}+d\log n<p$.
   By calling the algorithm \algName{Orthogonalize}~and
   \algName{Size-Reduce}, we can find in time
   \[\bigO{d^2np\left(1+\frac{d}{\log p}\right)}\]
   an integral triangular matrix $U\in
   \left(\order_\KK^\times\right)^{n\times n}$, such that $\|U\|\leq
   2^{\bigO{p}}$, and a matrix $R+E$, such that $\|E\|\leq 2^{-p}$,
   with $R$ being the $R$-factor of the QR decomposition of $AU$ and
   \[
     \kappa(AU)\leq \left(\frac{\max_i
         \norm_{\KK/\QQ}(R_{i,i})}{\min_i
     \norm_{\KK/\QQ}(R_{i,i})}\right)^{\frac1n}
     2^{\bigO{\sqrt{n\log n\log\log n}+d\log n}},
   \]
   for $\kappa(X)=\|X\|\|X^{-1}\|$ being the condition number of $X$.
 \end{theorem}
 \begin{proof}
   See \cref{app:precision}.
 \end{proof}

 \begin{corollary}
   \label{cor:size_reduce}
   Suppose that: \[ \|M^{(0)}\|, \|{M^{(0)}}^{-1}\|\leq 2^p
   \qquad\textrm{and}\qquad d\log
   n_h+\sqrt{n_h\log n_h\log\log n_h}<p.\]
   Then, we have the following bound on the condition number of $M^{(i)}$,
   valid for any loop index $1\leq i \leq \rho$:
   \[
     \kappa\left(M^{(i)}\right) \leq
     2^{2p+\bigO{\sqrt{n_h\log n_h\log\log n_h}+d\log n_h}},
   \]
   and the call of the procedure $\algName{Size-Reduce}$ at this $i$-th
   round has complexity \[\bigO{d^2n_hp\left(1+\frac{d}{\log
   p}\right)}\]
   and requires a $\bigO{p}$ of precision
 \end{corollary}
 \begin{proof}
   We first remark that for any $1\leq j \leq d$, the map
   $i \mapsto \max_j \norm_{\KK_h/\QQ}\left(R^{(i)}_{j,j}\right)$ is
   non-increasing, and therefore that $i \mapsto \min_j
   \norm_{\KK_h/\QQ}\left(R^{(i)}_{j,j}\right)$ is
   non-decreasing.

   Now, \cref{thm:unites} implies that the Archimedean
   embeddings are balanced so that we have for all $i$:
   \[
     \frac{\max_{\sigma:\KK_h\rightarrow\CC, R^{(i)}_{j,j}\in R^{(i)}}
       \left|\sigma\left(R^{(i)}_{j,j}\right)\right|}{\min_{\sigma:\KK_h\rightarrow\CC,
       R^{(i)}_{j,j}\in R^{(i)}}
     \left|\sigma\left(R^{(i)}_{j,j}\right)\right|}
     \leq 2^{2p+\bigO{\sqrt{n_h\log n_h\log\log n_h}}},
   \]
   and so that
   \[
     \frac{\max_j
       \norm_{\KK_h/\QQ}(R_{j,j})}{\min_i
       \norm_{\KK_h/\QQ}(R_{j,j})} = 2^{n_h\left(2p+\bigO{\sqrt{n_h\log
     n_h\log\log n_h}}\right)}.
   \]
   Therefore, by combining this bound with the result of
   \cref{thm:size_reduce}, after the call to
   \algName{Size-Reduce}, the condition number of $M^{(i)}$ is bounded
   by \[2^{2p+\bigO{\sqrt{n_h\log n_h\log\log n_h}+d\log n_h}}\] and the computation
   requires a $\bigO{p}$ bits of precision, with error bounded by
   $2^{-p}$.
 \end{proof}

 \subsubsection{Complexity of the \algName{Lift}~procedure}
 With the bounds given by~\cref{thm:size_reduce} we are now able to
 bound the complexity of the lift procedure described in~\cref{alg:generalized_euclide}.

 \begin{lemma}[Quasilinearity of \algName{Lift}]
   \label{lem:Euclid_complexity}
   Let $\KK$ be the cyclotomic field of conductor $f>0$,
   of dimension $n = \varphi(f)$. Denote by $r$ the largest prime factor of
   $f$. Let $a,b\in\order_\KK$ and suppose that:
   \[ \gcd(\norm_{\KK/\QQ}(a),\norm_{\KK/\QQ}(b))=1 \qquad \textrm{and} \qquad
   \|a\|+\|b\| \leq 2^{p}.\]
   Then, the time complexity
    of the algorithm \algName{G-Euclide}~on the inuput $(a,b)$ is a
	 \[\bigO{r\log(r)np\log p}\] for
   $p\geq \sqrt{n\log n\log\log n}$. Consequently, it is
   quasilinear for $r\leq \log n$.
   The output $(u,v)$ verify:
   \[
     au+bv=1 \qquad \textrm{and} \qquad \|u\|+\|v\|\leq 2^{p+\bigO{\sqrt{n\log
     n\log\log n}}}.
   \]
 \end{lemma}
 \begin{proof}
   We use a tower of
   number fields\footnote{Note that this tower is not
   same as the one used in the whole reduction process. The
     two towers are indeed constructed independently to optimize the
   global running time.} $\lL^\uparrow_h$,
   where $\lL_{i}=\QQ[x]/\Phi_{f_i}(x)$ and $f_i/f_{i+1}\leq
   r$. By trivial induction and multiplicativity of the relative norm
   map,
   we know that the input of the recursive call at level $i$, that is,
   in $\lL_i$ is $\norm_{\lL_h/\lL_i}(a),\norm_{\lL_h/\lL_i}(b)$. As
   such, with $p_i$ being the number of bits of the coefficients of the
   input at level $i$ of the recursion, we have $n_ip_i=\bigO{n_hp}$.
   Since computing the automorphisms corresponds to permutation of
   evaluation of a polynomial, each norm can be computed in time
  $\bigO{r\log(r)n_ip_i}$ using a product tree~\cite{moenck1972fast}.
   \medskip

   Now, we have by induction that $1=\det W=\det V$.
   With $R$ being the $R$-part of the $QR$-decomposition of $V$ we have
   at any level $i$ in the tower $\lL_h^\uparrow$: \[
   \|R_{2,2}\|=\|1/R_{1,1}\|\leq 2^{\bigO{\sqrt{n_i\log n_i\log \log n_i}}}, \] so
   that the size-reduction implies that
   \[
     \begin{aligned}
       \| M\| &\leq
       \norm_{\lL_i/\QQ}(R_{1,1})^{\frac{1}{n_i}}2^{\bigO{\sqrt{n_i\log
       n_i\log \log n_i}}}\\
              & =(n_h\|a\|+n_h\|b\|)^{\frac{n_h}{n_i}}2^{\bigO{\sqrt{n_i\log
              n_i\log \log n_i}}}.
     \end{aligned}
   \]
 Hence, the output coefficients are also stored using
 $\bigO{n_hp/n_i}$ bits. The complexity when $n_0=1$, i.e.\ the
 \algName{ExGcd}~base case, is classically in $\bigO{p_0\log p_0}$.
 Summing along all complexities gives:
 \[ \bigO{n_hp\log(n_hp)+\sum_{i=1}^h r\log(r)n_ip}=\bigO{n_hp\log p+r\log(r)n_hp\log
 n_h}\]
 which simplifies to a $\bigO{r\log(r)np\log p}$.
 \end{proof}

 \subsubsection{Complexity of the top-level}

 Now that we have analyzed the complexity and the output quality of
 each ``atomic'' parts, we can examine the complexity of the top-level
 of the algorithm \algName{Reduce}---that is to say its complexity
 without counting the recursive calls.

 \begin{proposition}
   \label{prop:complexity_toplevel}
   Suppose that the following conditions are fulfilled:
   \[
     \min_{\sigma:\KK_h\rightarrow\CC,
     R^{(1)}_{i,i}\in R^{(1)}}\left|\sigma(R^{(1)}_{i,i})\right| \geq 2^{-p}, \quad
     \alpha = {\log^{\bigO{1}}(n_h)}
   \]
   \[
     d\log n_h+\sqrt{n_h\log n_h\log\log n_h}<p.
   \]
   Then, the complexity at the top-level of the algorithm is
   a $\bigO{d^5n_hp\log p}$.
 \end{proposition}
 \begin{proof}
   \begin{description}
     \item[Base case: $\KK_h = \QQ$]
       This is a consequence of the analysis of Schönhage's
       fast reduction~\cite{Schonhage91}.

     \item[General case]
       Using \cref{cor:approximation_factor}, the number of rounds
       is $\rho = \bigO{d^2\log p}$. By
       \cref{lem:Euclid_complexity} the complexity of
       \algName{Lift} is quasilinear. Thus, the
       complexity of each round is dominated by the computation of the
       $QR$ decomposition and the size-reduction. By
       \cref{thm:size_reduce}, this complexity is a
       $\bigO{d^3n_hp/\log p+d^2n_hp}$, yielding a global complexity of
       $\bigO{d^5n_hp+d^4n_hp\log p}=\bigO{d^5n_hp\log p}$.
   \end{description}
 \end{proof}

 \subsubsection{Bounding the precision at each level}

 We now bound the precision used in the recursive calls at the
 top-level of the \algName{Reduce} algorithm:
 \begin{lemma}
   \label{lem:sum_of_precisions}
   The sum of all bit sizes used in the recursive calls at the
   top-level is $\bigO{d^2p}$, when subjected to the conditions:
   \[
      \min_{\sigma:\KK_h\rightarrow\CC,
     R^{(1)}_{i,i}\in R^{(1)}}\left|\sigma(R^{(1)}_{i,i})\right| \geq 2^{-p}
     \qquad d\log n_h+\sqrt{n_h\log n_h\log\log n_h}<p.
   \]
 \end{lemma}
 \begin{proof}
   Recall that the potential of the basis is defined
   as \[\Pi=\sum_{j=1}^d (d-j)\log(\norm_{\KK_h/\QQ}(R_{j,j})),\]
   which is in $\bigO{n_hd^2p}$ by assumption on $p$. Let $1\leq j\leq d$,
   then the
   reduction algorithm is about to perform a local reduction of the
   projected sublattice $(r_j, r'_{j+1})$, as presented in
   \cref{sec:two_by_two_reduction}, two cases can occur:
   \begin{itemize}
     \item Either $\norm_{\KK_h/\QQ}(R^{(i)}_{j,j})\leq \min\left(2^{\alpha
         n_h^2}
       \sqrt{\norm_{\KK_h/\QQ}
       \left(R^{(i)}_{j,j}R^{(i)}_{j+1,j+1}\right)},
     \norm_{\KK_h/\QQ}\left(R^{(i)}_{j,j}\right)\right)$, and as
     mentioned in
       \cref{sec:heuristic} the local reduction is not performed.
       We can consider that we use here a zero precision call.
     \item Either a local reduction is actually performed and by
       the result of \cref{sec:well_conditioned}, we can
       use a precision in $\bigO{p_{i,j}}$ with:
       \[
         p_i={\log\left(\frac{\max_k \sigma_k (R^{(i)}_{j,j})}{\min_k
               \sigma_k (R^{(i)}_{j+1,j+1})}\right)}
       \]
       to represent the projected lattice. Let now set
  \[ L = \frac{\log(\norm_{\KK_h/\QQ}
  (R^{(i)}_{j,j}/R^{(i)}_{j+1,j+1}))}{n_h}.\]
  The precision $p_{i,j}$ is, thanks  to
       the unit rounding~\cref{thm:unites} a
       \[\bigO{L+\sqrt{n_h\log n_h\log\log n_h}} = \bigO{L},\]
       by hypothesis. The reduction of
       this truncated matrix yields a unimodular
       transformation, represented with precision $\bigO{p_{i,j}}$, which
       when applied to the actual basis matrix implies that $\Pi$
       decreases by a term at least:
       \[
       \delta_{i,j} = n_h\left[\frac{L}{2}-\alpha
     n_h\right]-{2^{-\Omega(p)}}
       \]
       by \cref{heur:lifting_size} and
       \cref{thm:well_conditioned}.
     Let us bound the ratio $p_{i,j}/\delta_{i,j}$:
       \[
         \frac{p_i}{\delta_i} =
	   \frac{L+\bigO{\sqrt{n_h\log n_h\log\log
	   n_h}}}{\left(\frac{L}{2}-\alpha
		   n_h\right)n_h-2^{-\Omega(p_{i,j})}}
		   =\frac{1+\bigO{\frac{\sqrt{n_h\log n_h\log\log n_h}}{L}}}
         {\frac{n_h}{2}-\frac{\alpha
         n_h^2}{L}-\frac{2^{-\Omega(p_{i,j})}}{2L}}.
       \]
       Now recall that
       $\norm_{\KK_h/\QQ}(R^{(i)}_{j,j})\geq
       2^{2(1+\varepsilon)\alpha
       n_h^2}\norm_{\KK_h/\QQ}(R^{(i)}_{j+1,j+1})(1-2^{-\Omega(p_{i,j})})$,
       the
       multiplicative error term coming from the precision at which the
       values of the $R^{(i)}_{j,j}$ and $R^{(i)}_{j+1,j+1}$ are approximated
       at
       runtime. Thus, we have:
        \[
          \sqrt{n_h\log n_h\log\log n_h}/L = \bigO{\sqrt{\frac{\log
          n_h\log\log n_h}{n_h}}},
        \]
        and
        \[
  \alpha n_h^2/L \leq \frac{n_h}{2(1+\varepsilon)}.
        \]
        As such we have:
       \[
         \frac{p_{i,j}}{\delta_{i,j}} \leq
         \frac{1+\bigO{\sqrt{\frac{\log
         n_h\log\log
 n_h}{n_h}}}}{\frac{n_h\varepsilon}{1+\varepsilon}+\littleO{1}}.
       \]
       But then, $\delta_{i,j} = \Omega(n_h\varepsilon p_{i,j})$.
  \end{itemize}
   The potential is always a sum of non-negative terms, so $\sum_{i,j}
   \delta_{i,j}\leq \Pi$.
   The sum of the precision for the calls can thus be bounded
   by $\bigO{\frac{\varepsilon}{(1+\epsilon)} \frac{\Pi}{n_h}}
   =\bigO{d^2p}$, since $\varepsilon=\frac12$,
   which
   concludes the proof.
 \end{proof}

 Eventually we can prove the general complexity of the algorithm:
 \begin{proof}[Proof of \cref{thm:complexity_fast_LLL}]
   The first step of the proof consists in selecting a suitable tower
   of subfields, for which the relative degrees are chosen to optimize
   the complexity of the whole reduction. We choose a tower of
   cyclotomic subfields $\KK_h^\uparrow = \left(\QQ = \KK_0 \subset
   \KK_1 \subset \cdots \subset \KK_h\right)$ with $[\KK_i:\QQ]=n_i$
   and $n_{i+1}/n_i=r_i$ which satisfies $r_i/n_{i+1}^{1/5}\in [1;\log f]$, so
   that $h=\bigO{\log\log n}$.  This always exists as $f$ is log-smooth.  We
   can set $\alpha_i=4^{h-i+1}$ to satisfy
   the conditions of \cref{lem:sum_of_precisions} while
   making~\cref{heur:lifting_size} practically possible.
   By definition of the value set for $p$ we have $p=\bigO{B}$. And it of
   course satisfies the requirements of \cref{prop:complexity_toplevel}.
   Note that by the choices of local precision made in the proof
   \cref{lem:sum_of_precisions}, a simple induction shows that at each
   level of the recursion the local precision fulfills the condition of
   \cref{lem:sum_of_precisions}, by the exact choice of the $p_{i,j}$'s.
   A by product of this induction asserts that the sum of the precision used
   in \emph{all} the recursive calls needed to reduce a projected lattice at
   level $i$ is a
   \[
   \bigO{p\prod_{j=1}^{i-1}
     \bigO{r_j^2}}=2^{\bigO{i}} B\left(\frac{n}{n_i}\right)^2.
   \]

   Then, since by \cref{prop:complexity_toplevel} the complexity of the
   top-level call at level $i$ is a
   $\bigO{r_i^5n_ip\log(p)}=\bigO{r_i^5n_iB\log(B)}$.
   Hence the total complexity at level $i$ is $r_i^5/m_i\cdot
   n^2B\log(Bn)2^{O(i)}=n^2B\log(B)\log^{\bigO{1}} n$.
   Summing over all the levels retrieves the announced result.
 \end{proof}

 An important point is that all recursive calls can be computed in
 parallel, and as most of the complexity is in the leaves, this leads
 to an important practical speed-up.
 We conjecture that when the number of processors is at most $n/\log^{\bigO{1}} n$,
 the speed-up is linear.
 \section{A fast reduction algorithm for high-rank lattices}
\label{sec:fastlll}

While the previous reduction was tailored to reduce small (typically rank 2)
rank lattices over cyclotomic fields, we now turn to the reduction of high
rank lattices.  It runs roughly in a constant number of matrix
multiplications.  It can also be used in the previous algorithm at each step to
reduce the hidden logarithmic powers; but is of course interesting on its own for
reducing rational lattices.

A bottleneck with \cref{alg:fast_lll} is that each round needs a
matrix multiplication, and there are at least $d^2$ rounds.
However, one can notice that each round only make local modifications.
As a result, we propose to use a small number $D$ of blocks,
typically 4 or 8 suffices, and a round will (recursively) reduce
consecutive pairs of dimension $d/D$.  The resulting number of rounds is
again $\bigO{D^2\log B}$, giving a top-level complexity of
$\bigO{D^2}$ (equivalent) multiplications. The corresponding algorithm
is given in \cref{alg:ultra_fast_lll}. In addition, the naive
\algName{Size-Reduce} procedure is replaced by a variant of Seysen
reduction, which is detailed in \cref{app:ultra_fast}. The complexity
analysis is exactly the same as in the previous section, with the flag
profile defined with respect to the volume of the blocks instead of
simply the vectors, that is:
\[
\profile(M)_j=\sum_{k = 1}^{jd/D}
\left(\log|\norm_{\KK/\QQ}\left(R_{k,k}\right)|+
2k(1+\varepsilon)\alpha n_h^2\right), \qquad \textrm{for~} 1\leq j \leq
d.
\]

We describe the algorithm with respect to an oracle \algName{Oracle}{}
which computes the base case. One can either use \algName{Schonhage}{},
the algorithms in the previous or current section, or a recursive call.

\begin{boxedAlgorithm}[algotitle=Reduce, label=alg:ultra_fast_lll]
  \begin{algorithm}[H]
    \Input{
      Basis $M\in\order_{\KK}^{d\times d}$ of the
      $\order_{\KK}-$module
    $\module$}
    \Output{A unimodular transformation $U\in\order_{\KK}^{d\times
      d}$
    representing a reduced basis of $\module$.}
    \BlankLine
    \lIf{$d=2$} { \Return \algName{~Oracle}$(M)$ }
    \For{$i=1$ \To $\rho$}
    {
      $R \gets $ \algName{Orthogonalize}$(M)$\;
      $U_i \gets $ \algName{Seysen-Size-Reduce}$(R)$\;
      $(M,R) \gets (M,R)\cdot U_i$ \;
      \For{$j=1+(i \mod 2)$ \To $d$ \Bystep $2d/D$}
      {
	      $V_1 \gets \covol R[j:j+d/D-1,j:j+d/D-1]$ \;
	      $V_2 \gets \covol R[j+d/D:j+2d/D-1,j+d/D:j+2d/D-1] $
	      \If{$ V_1 \leq 2^{2(1+\varepsilon)\alpha n_h^2 d/D} V_2$ }{
	$U' \gets $ \algName{Reduce}$(
	R[j:j+2d/D-1, j:j+2d/D-1] )$\;
	$(U_i, M) \gets (U_i,
	M)\cdot\textrm{Diag}\left(\Id_{j},
	  U', \Id_{2d-j-2}\right)$ \;
      }
}
    }
    \Return $\prod_{i=1}^\rho U_i$ \tcp{The product is computed from the
      end}
  \end{algorithm}
\end{boxedAlgorithm}

The analysis by Neumaier-Stehlé~\cite{ISSAC:NeuSte16} only bounded the
number of rounds, and as a result the complexity is
$d^{3}B^{1+\littleO{1}}$.  One can remark that even the simple
algorithm uses $\Omega(d^3\log B)$ local reductions, so that
significantly decreasing their complexity can only come
from a reduced precision in this local operation.

We, on the other hand, make the following heuristic:
\begin{heuristic}\label{heur:heurdeux}
  At any point in the
recursion, when reducing a lattice of rank $d$, if we use a precision of
$p\geq (1+\epsilon)\alpha dn$ then we
decrease the potential $\Pi$ by $\Omega(d^2p)$.
\end{heuristic}

It is justified by the fact that the $\norm_{\KK/\QQ}(R_{i,i})$
usually decrease roughly exponentially in $i$ both in the input and the
output matrices.

We need one last heuristic, which removes a $\log B$ factor:
\begin{heuristic}\label{heur:heurtrois}
  The number of bits needed decreases exponentially quickly, at the same
  speed as the
$\mu$ vector.
\end{heuristic}
Indeed, a standard assumption for random lattices is
that the upper-bound in~\cref{heur:lifting_size} is in fact an
approximation.  As a result, we expect
that~\cref{lem:inequality_omega_kappa} holds with the vectors
replaced by their forward differences, which implies this heuristic. The
same property also implies the previous heuristic, as the forward
difference of the eigenvector corresponding to the largest eigenvalue is
a cosine.

\begin{theorem}
  Let $A$ be a matrix of dimension $d$ with entries in $\KK$, with
  $\kappa(A)\leq 2^B$ such that $B\geq \sqrt{n\log n\log \log n}+\log n\log d$, $n$
  being the degree of $\KK$ over $\QQ$.
  Given $A$ and an oracle which obeys~\cref{heur:lifting_size}, our
  reduction algorithm finds an integer vector $x$ with
 \[ \|Ax\| \leq 2^{2(1+\epsilon+\littleO{1})\alpha dn}\covol^{1/nd} A, \]
  with $\alpha$ and $\epsilon$ defined as in the
  \cref{heur:heurdeux}.
  Further, the sum of the precision used in the oracle calls is
  $\bigO{d^2p}$ and the {\em heuristic} running time is
  \[ \bigO{\frac{d^\omega}{(\omega-2)^2} n\cdot B/\log B+d^2nB\log^2 d}
  \]
  for any constant $\epsilon$.
\end{theorem}
\begin{proof}
  Let $r_i$ be the rank of the matrix at the $i$-th recursive level (one
  is the top).
  We use $w=\lfloor \log(B) \rfloor$.
  Then, using our heuristic on the potential, the sum of the precision
  $p$ used in this level is $\bigO{(d/r_i)^2B}$.
  Using the complexity results presented in \cref{app:ultra_fast}
  each call with precision $p\geq \log B$ has a running-time of
  \[ \bigO{\left(r_{i+1}/r_i\right)^2\left(\frac{r_i^\omega}{\omega-2} n\cdot p/\log
  B+r_i^2 np\log r_i\right)} \]
  using \cref{heur:heurtrois} on the exponential decrease of the precision
  used.
  Thus, the complexity of the $i$-th level is
  \[ \bigO{(r_{i+1}/r_i)^2\left(d^2r_i^{\omega-2} \frac{n}{\omega-2}\cdot
  B/\log B+d^2 nB\log r_i\right)}. \]
  If $p<\log B$, then $r_in=\bigO{\log B}$ and the cost is bounded by
  \[ \bigO{(r_{i+1}/r_i)^2\left(\frac{r_i^\omega}{\omega-2} n+r_i^2 np\log
  r_i\right)} \]
  which in total is at most
  \[ \bigO{(r_{i+1}/r_i)^2\left(d^2r_i^{\omega-2} \frac{n}{\omega-2}\cdot
  \frac{B}{r_in}+d^2 nB\log r_i\right)}. \]
  As $r_i^{\omega(r_i)-3}/(\omega(r_i)-2)$ is bounded, this is always
  negligible.

	One possible instantiation is $r_{i}/r_{i+1}$ bounded.
	We then get $\sum_i d^2r_i^{\omega(r_i)-2}=\bigO{\frac{d^\omega}{\omega-2}}$.

	This gives an algorithm which finds a transition matrix such that the first block has a low volume:
  \begin{equation}\label{eq:final_cp}
  \prod_{i=1}^{r_1} |\norm_{\KK/\QQ}(R_{i,i})|^{1/r_1}\leq
  2^{2(1+\epsilon)\alpha (d-r_1)n^2}\covol^{1/d} A
\end{equation}
	One then recurses on the first block, which corresponds to taking the
  product of a family of formulas of the same shape as
  \cref{eq:final_cp} for which the $(d-r_1)$ is replaced by a
  $(r_{i}-r_{i+1})$. The results derives directly from a telescopic
  summation over the exponents. This recursion is done for a fraction of
  the global complexity.
\end{proof}
We emphasize that {\em in practice}, the entire basis is reduced at the
end of the algorithm.

If we instantiate on rational lattices, this gives:
\begin{corollary}
\label{cor:cor3}
  Let $A$ be a matrix of dimension $d$ with entries in $\ZZ$, with
  $\kappa(A)\leq 2^B$ such that $B\geq d$.
  Given $A$ and an oracle which obeys~\cref{heur:lifting_size}, our
  reduction algorithm finds an integer vector $x$ with
  \[ \|Ax\| \leq 2^{d/2}|\det A|^{1/d}. \]
  Further, the {\em heuristic} running time is
  \[ \bigO{\frac{d^\omega}{(\omega-2)^2} \cdot B/\log B+d^2B\log B}. \]
\end{corollary}
This is, up to the $1/(\omega-2)$ factor on the first term
the complexity of QR-decomposition, so the
algorithm is essentially optimal.

Almost always, the first term is dominant and one can use $r_{i}/r_{i+1}=(d/r_i)^{1/3}$.
The number of levels is then only $\bigO{\log \log d}$, and the larger $r_i/r_{i+1}$ makes the heuristics more plausible.

Once the matrix is \LLL{} reduced, we can reduce it further with a \BKZ{} algorithm.
We can use the same recursive structure, but when the dimension is less than $\beta \log(\beta)$,
we use a \BKZ{} reduction.
The total number of calls is $\bigO{d^3\log d}$~\cite{C:HanPujSte11}, and we also have an approximation factor
of $\beta^{\bigO{d/\beta}}$.
Hence, we can use $\beta=\Theta(\log(Bd^{\omega-3}))$, which for $\omega$ not too small is $\Omega(\log(d))$, without
increasing the running time.
This implies that we can remove a $\log d/\log \log d$ factor when solving vectorial knapsacks, such as the ones
for polynomial factoring~\cite{van2002factoring}.

We can instantiate this algorithm on roughly triangular matrices, and show that
for random ones, one can get a (heuristic) significant speed-up.
These matrices are widespread, as it corresponds to ``knapsack''
problems or searching integer
relations\footnote{While PSLQ~\cite{ferguson1998polynomial} also solves
  this problem on \emph{real RAM machines}, this model is an extremely
  poor approximation of computers~\cite{schonhage1979power}.
  See~\cite[Section 2]{ferguson1998polynomial} for what can go wrong,
e.g.}.
In particular, one can quickly search a putative minimal polynomial.
\begin{theorem}
\label{thm:thm6}
	Let $A$ be a ``random'' matrix of dimension with $d$ columns,
  $\bigO{d}$ rows and entries in $\order_{\KK}=\ZZ[x]/\phi_f(x)$.
	We define $B\geq d^2n$ such that \[ \|A\|+\norm_{\KK/\QQ}(\covol C)\leq 2^B \] for all matrices $C$ whose columns are a subset of $A$.
	For $R$ the R-factor of the QR-decomposition of $A$, we also assume that  $\|R^{-1}\|\leq 2^{B/d}$ and $\|R_{i,j}\|\leq 2^{B/i}$ for all $i,j$.
	We also require that $A_{i,j}=0$ for $i\geq \bigO{j}$ with a uniform constant.
	We can find an integer vector $x$ with
	\[ \|Ax\| \leq 2^{d\bigOtilde{n}}\covol^{1/nd} A. \]
	The heuristic complexity is
	\[ \bigO{\frac{d^{\omega-1}}{(\omega-2)^2} n\cdot B/\log B+dnB\log^2 d}+d\bigOtilde{n^2B}\]
\end{theorem}
\begin{proof}
	The algorithm consists in reducing the first $k=2^i$ columns of $A$ for successive powers of two until $d$.
	The result is stored in $A_i$.
	The volume of $A_i$ is bounded by $2^B$ so heuristically we expect, and will assume that
	\[ \|A_i\|,\|A_i\|^{-1}=2^{d\bigOtilde{n}+\bigO{B/2^i}} \]
	for all $i$.
	We also store $Q_iR_i$, the QR-decomposition of $A_i$, and $R_i^{-1}$.
	We now explain how to compute $A_{i+1}$.
	Let $x$ be a column of $A$ which is not in the span of $A_i$.
	In order to reduce its bit size, we replace it by $x-A_j\lfloor R_j^{-1} \conj{Q_j}^t x \rceil$ for increasing $j$.
	This reduces the size of the projection of $x$ orthogonally to $A_j$ to $2^{d\bigOtilde{n}+\bigO{B/2^j}}$; and by assumption on the input matrix, this is also true for the part orthogonal to $A_j$.
	At the end of this process, the length of $x$ is therefore at most $2^{d\bigOtilde{n}+\bigO{B/2^i}}$.
	For efficiency, this reduction is computed on all $d$ vectors at the same time.
	Now we concatenate to $A_i$ all the reductions of the needed vectors, and use our lattice reduction algorithm on the R-factor of the QR-decomposition of this matrix.

	We now show that this matrix is well-conditioned.
	This matrix is written as
	\[ \begin{pmatrix} R_i & W \\ 0 & Z \end{pmatrix}. \]
	We remark that the reduction process did not change $Z$, so that $\|Z^{-1}\|\leq \|R^{-1}\|\leq 2^{B/d}$.
	The inverse is
	\[ \begin{pmatrix} R_i^{-1} & -R_i^{-1}WZ^{-1} \\ 0 & Z^{-1} \end{pmatrix} \]
	so that its condition number is bounded by $2^{d\bigOtilde{n}+\bigO{B/2^i}}$.

	The lattice reduction calls cost in total
	\[ \bigO{\sum_{i=1}^{\log d} \frac{2^{\omega(2^i)i}}{(\omega-2)^2} \frac{B}{2^i\log B}+2^{2i}nBi/2^i} \]
	and the cost of the oracles are bounded using the fact that $\Pi=\bigO{dnB}$:
	\[ d\bigOtilde{n^2B}. \]
	The pre-reduction computed by the algorithm has a running time of:
	\[ \bigO{\sum_{i=1}^{\log d} \frac{2^{\omega(2^i)i}}{\omega-2}\cdot \frac{d}{2^i} \frac{B}{2^i\log B}+d2^inB/2^i}. \]
	Summing these complexities leads to the announced result.
\end{proof}

One can check than knapsack matrices, or Hermite Normal Form matrices with decreasing round pivots verify the assumptions with a small $B$.

 \section{Symplectic reduction}
\label{sec:symplectic}

\subsection{On symplectic spaces and symplectic groups}
In the following, we very briefly introduce the linear theory of
symplectic geometry and establish all along this presentation the parallel
between the Euclidean and Symplectic geometries.

\subsubsection{Definitions}
A \emph{symplectic space} is a finite dimensional vector space $E$
endowed it with an antisymmetric bilinear form
$J: E\times E\rightarrow E$. We can define a natural
orthogonality relation between vectors $x,y\in E$ as being $J(x,y)=0$.
The linear transformations of $E$ letting the symplectic structure $J$
invariant is a group, called the $J$-symplectic group (or symplectic
group if the context makes $J$ clear). This group plays a similar role
to the \emph{orthogonal group} for Euclidean spaces.

\subsubsection{Darboux bases}
However on the contrary to Euclidean spaces, a symplectic space does not
possess an orthogonal basis, but instead a basis $e_1, \ldots, e_d,
f_1, \ldots, f_d$, so that for any indices $i < j$ we have
$J(e_i, e_j) = 0, J(f_i, f_j) = 0, J(e_i, f_j) = 0$ and
$J(e_i, f_i) > 0$. It implies in particular that any symplectic space
has even dimension. We demonstrated in \cref{sec:gram_schmidt} that
it is easy to transform any basis of a Euclidean space in an orthogonal
basis. This iterative construction is easily adapted to the symplectic
case.

\subsubsection{Symplectic lattice, size reduction}
We can now easily adapt the definition of a lattice to the symplectic
setting:
\begin{definition}
  A \emph{symplectic lattice} $\Lat$ is a
  finitely generated free $\ZZ$\nobreakdash-module,
  endowed with a symplectic form $J$ on the rational vector space
  $\Lat\otimes_\ZZ \QQ$.
\end{definition}
As mentioned in \cref{sec:algorithm_LL_size_reduction}, an important
tool
to reduce lattices is the \emph{size-reduction} procedure, which
can be viewed as a discretization of the Gram-Schmidt orthogonalization.
It aims at reducing the size and the condition number of the
lattice basis.  When dealing with symplectic symmetries, we can also
discretize the process to obtain a basis which is close to a Darboux
basis.

As we generalized the lattice formalism to $\order_\KK$-modules in
number fields, we can generalize straightforwardly the notions
of symplectic lattices to the algebraic context. Using the work
presented in~\cref{sec:fast_LLL_NF}, we aim
at providing a fast reduction algorithm for  $\order_\KK$-modules
using these symplectic considerations.

\subsubsection{Towards an improved algorithmic size-reduction}

The specificities of the symplectic symmetry and of the evoked
symplectic size-reduction enable a faster algorithm.

Indeed, we will demonstrate that a
local reduction within the first half of the matrix can be applied
directly to the second half. This almost divides by two the overall
complexity {\em at each descent}.

In the rest of this section, we generalize the work
of Gama, Howgrave-Graham and Nguyen~\cite{EC:GamHowNgu06}
on the use of symplectic symmetries lattices within the reduction
process.  In particular, we show that such techniques can be used for
all towers of number fields, and instead of an overall constant factor
improvement, we can gain a constant factor at each floor of the tower
and then cumulate them. Lattice reduction algorithms hinge on the two following
facts:
\begin{description}
  \item[Size reduction] We can control the bit size without changing
    the Gram-Schmidt norms.
  \item[Local reduction] Any two consecutive Gram-Schmidt norms can
    be made similar.
\end{description}
We therefore have to show that these two parts can be done while
preserving the symplectic property.

\subsection{J-Symplectic group and compatibility with extensions}

In all the following we fix an \emph{arbitrary} tower of number
fields \[ \KK_h^\uparrow = \left(\QQ = \KK_0 \subset \KK_1 \subset
\cdots \subset \KK_h\right). \] For any $1\leq i \leq h$ we denote by
$d_h$ the relative degree of $\KK_{h}$ over $\KK_{h-1}$. On any of these
number fields, we can define a simple symplectic form, which derives
from the determinant
form:
\begin{definition}
  Let $\KK$ be a field, and set $J$ to be an antisymmetric bilinear form
  on $\KK^2$. A matrix $M\in\KK^{2\times 2}$ is said to be
  $J$-\emph{symplectic} (or simply symplectic if there is no ambiguity
  on $J$) if it lets the form $J$ invariant, that is if $J\circ M=J$.
\end{definition}

Let us instantiate this definition in one of the fields of the tower
$\KK_h^\uparrow$ on the $2\times2$-determinant form.
Let $J_{h}$ be the antisymmetric bilinear form on $\KK_h^2$ which is given
as the determinant of $2\times 2$ matrices in $\KK_h$, i.e.
\[
  J_{h}\left(\begin{pmatrix} x_0 \\ x_1
      \end{pmatrix},\begin{pmatrix} y_0 \\
    y_1
\end{pmatrix}\right)=x_0y_1-x_1y_0.
\]

\begin{remark}
  In the presented case, $M$ is $J_{h}$-symplectic iff $\det M=1$.
\end{remark}
Notice that we can always scale a basis so that this condition is verified.

We descend the form $J_{h}$ to $\KK_{h-1}$ by composition with a
non-trivial linear form $\KK_h \rightarrow \KK_{h-1}$, for instance by
using the relative trace, that is
$J_{h}'=\tr_{\KK_h/\KK_{h-1}}\!\! \circ\, J_{h}$.  We then
extend the definition of symplectism to $\KK_{h-1}^{2d_h}$ by stating that
a $2d_h\times 2d_h$ matrix $M'$ is symplectic if it preserves
the $J_{h}'$ form, that is if $J'_h \circ M'=J'_h$. This construction is
tailored to be compatible with the descent of a matrix to $\KK_{h-1}$ in
the following sense:

\begin{lemma}
  Let $M$ be a $2\times 2$ matrix over $\KK_h$ which is
  $J_{h}$-symplectic, then its descent $M'\in\KK_{h-1}^{2d_h\times
	2d_h}$ is $J_{h}'$-symplectic.
\end{lemma}

\subsection{Towards module transformations compatible with
$J$-symplectism}

Before exposing the transformation matrices in our
size-reduction process of symplectic lattices, we give an insight on
these techniques coming from the Iwasawa decomposition of Lie groups.

\subsubsection{On the Iwasawa decomposition}
The \emph{Iwasawa decomposition} is a factorization of any semisimple Lie
group in three components, which generalizes the decomposition of
$\textrm{GL}(n,\RR)$ in the product $KAN$ where $K = O(n,\RR)$ is the
orthogonal group, $A$ is the group of diagonal matrices with positive
coefficients and $N$ is the unipotent group consisting of upper triangular
matrices with 1s on the diagonal. This decomposition of $\textrm{GL}(n,\RR)$
arises directly from the Gram-Schmidt decomposition of any real matrix and
extracting the diagonal of its $R$ part.
The $J-$symplectic group defined here is a semisimple Lie group and thus is
subject to Iwasawa decomposition. We aim at using an \emph{effective} version
of the Iwasawa decomposition.
In order to compute effectively such a decomposition, we need to
find a generating set of \emph{elementary} transformations over
bases, which generalizes the operators of transvections and swaps in
the general linear case.

We start by treating a simpler case: the Kummer-like extensions.
The general case is covered in~\cref{app:sympallnf}.

\subsubsection{A simple case: Kummer-like extensions
$\KK[X]/(X^{d_h}+a)$}
We define $R_{d_h}$ as the reverse diagonal of $1$ in a square matrix of
dimension $d_h$.

In this section, we use the notation $A^s$ as a shorthand for
$R_{d_h}A^TR_{d_h}$, which corresponds to the reflection across the
antidiagonal, that is exchanging the coefficients $A_{i,j}$ with
$A_{d_h+1-i, d_h+1-j}$.
We proceed here by adapting the work of Sawyer~\cite{SAWYER}.
Suppose that the defining polynomial of $\KK_h/\KK_{h-1}$ is
$X^{d_h}+a$.
Recall that $J_{h}$ is the $2\times2$-determinant form over
$\KK_h^2$. We can compose it by the linear form
\[
  \left|
  \begin{array}{rcl}
    \KK_h \cong \KK_{h-1}[X]/(X^{d_h}+a) & \longrightarrow &
    \KK_{h-1}\\
	  y & \longmapsto & \tr_{\KK_h/\KK_{h-1}}(\frac{Xy}{d_h a})
  \end{array}
\right.,
\]
  to construct the
matrix $J'_{h}$, which now becomes
\[J_{h}' = \begin{pmatrix} 0 & R_{d_h} \\ -R_{d_h} &
0\end{pmatrix}\]
in the power basis. In this particular setting we
retrieve the instantiation of~\cite{EC:GamHowNgu06}.
In particular:
\begin{lemma}
	Fix a basis of the symplectic space where
	the matrix corresponding to $J'_h$ is $\begin{pmatrix} 0 & R_{d_h} \\ -R_{d_h} &
0\end{pmatrix}$.
	Then, for any $M$ a $J'_h$-symplectic matrix and $QR$ its QR decomposition,
	both $Q$ and $R$ are $J'_h$-symplectic.
\end{lemma}
\begin{proof}
	Direct from the explicit Iwasawa decomposition given by~\cite{SAWYER}.
\end{proof}

\begin{lemma}[Elementary $J_{h}'$-symplectic matrices]
  \label{lem:elementary_computations} \hfill \break
  \begin{itemize}
    \item  For any $A\in\textrm{GL}(d_h,\KK_h)$,
  \[ \begin{pmatrix} A & 0 \\ 0 & A^{-s} \end{pmatrix} \]
  is $J_{h}'$-symplectic.
\item For any $A\in\textrm{GL}(2,\KK_h)$ with $\det A=1$ the block matrix
  \[ \begin{pmatrix} \Id_{d_h-1} & 0 & 0 \\
      0 & A & 0 \\ 0 & 0 &
  \Id_{d_h-1} \end{pmatrix} \]
 is $J_{h}'$ symplectic.
  \end{itemize}
\end{lemma}
\begin{proof} By direct computation. \end{proof}

 We now turn to the shape of triangular $J_{h}'$ symplectic matrices.

 \begin{lemma}
   \label{lem:triang_struct}
   Block triangular symplectic matrices are exactly the matrices of the
   form
   \[ \begin{pmatrix} A & AU \\ 0 & A^{-s} \end{pmatrix} \]
   where $U=U^s$.
 \end{lemma}
 \begin{proof}
   Let $M=\begin{pmatrix} A & U \\ 0 & B \end{pmatrix}$ a block
   triangular
   matrix. By \cref{lem:elementary_computations}, the action of the
   block diagonal matrices $\begin{pmatrix} A & 0 \\ 0 & A^{-s}
   \end{pmatrix}$ by left multiplication preserves the
   $J_{h}'$-symplectic
   group, so that without loss of generality we can suppose that $A$ is
   the identity matrix.  Identifying the blocks of $M^TJ_{h}'M =
   J_{h}'$
   yields two relations:
   \begin{itemize} \item $R_{d_h}B = R_{d_h}$, entailing
   $B=\Id_{d_h}$,
      \item $B^TR_{d_h}U-U^TR_{d_h}B=0$, so that $R_{d_h}U =
        U^TR_{d_h}$, and as such
        $U = U^s$.
   \end{itemize}
 \end{proof}

\subsubsection{Size-reduction of a $J_{h}'$-symplectic matrix}
\label{sec:size_reduction_JR}
A direct consequence of \cref{lem:elementary_computations} is that
the local reductions occurring during the reduction, that is swaps and
transvections can preserve the $J_{h}'$-symplectism by using the
corresponding previous constructions.

Consider $X$ a $J_{h}'$-symplectic matrix, we want to efficiently
\emph{size-reduce} $X$ using the symmetries existing by symplectism.
Let first take the $R$ part of the QR-decomposition of $X$ and make
appear the factors $A$ and $U$ as in \cref{lem:triang_struct}.

Then we can focus on the left-upper matrix $A$ and size-reducing it into
a matrix $A'$. Each elementary operations performed is also symmetrically
performed on $A^s$ to retrieve $(A')^s$. Eventually the size reduction
is completed by dealing with the upper-right block, which is done by
performing a global multiplication by
\[
  \begin{pmatrix}
    \Id_{d_h} & -\lfloor U \rceil \\
    0 & \Id_{d_h}
  \end{pmatrix}.
\]
The corresponding algorithm is given in \cref{alg:symplectic_size_reduce},
and uses the ``classical'' \algName{Size-Reduce}~procedure as a subroutine.
The recursive reduction algorithm using the symplectic structure is then the
exact same algorithm as \cref{alg:fast_lll}, where the
size-reduction call of line 4 is replaced by \algName{Symplectic-Size-Reduce}.

\begin{boxedAlgorithm}[algotitle={Symplectic-Size-Reduce},
  label=alg:symplectic_size_reduce]
  \begin{algorithm}[H]
    \Input{$R$-factor of the QR decomposition of
    a  $J'_h$-symplectic matrix $M\in\order_{\KK_h}^{d\times d}$}
    \Output{A $J'_h$-symplectic unimodular transformation $U$ representing the size-reduced
    basis obtained from $M$.}
    \BlankLine
    Set $A,U$ such that $\begin{pmatrix} A & AU \\ 0 & A^{-s} \end{pmatrix} =
    R$\;
    $V \gets$ \algName{Size-Reduce}$(A)$\;
    \Return $\begin{pmatrix} V & -V\left\lfloor U\right\rceil \\ 0 & V^{-s} \end{pmatrix}$
  \end{algorithm}
\end{boxedAlgorithm}

The size reduction property on $A'$ implies that both $A'$ and $A'^{-1}$
are small, and therefore it is easy to check that the same is true for
the now reduced $R'$ and of course for the corresponding size reduction
of the matrix $X$ itself.

This approach admits several algorithmic optimizations:
\begin{itemize}
  \item Only the first half of the matrix $R$ is actually needed to
    perform the computation since we can retrieve the other parts.
    Indeed, with the equation $QR=X$, $R$ is upper triangular and it
    only depends on the first half of $Q$.
  \item Further, we compute only the part above the antidiagonal of
    $AU$.  This is actually enough to compute the part above the
    antidiagonal of $A^{-1}(AU)$, which is persymmetric.
  \item An interesting implication is that since we need to compute only
    half of the QR decomposition, we need (roughly) only half the
    precision.
\end{itemize}

\subsection{Improved complexity}

We analyze the algorithm of the previous section with the size-reduction
of \cref{sec:size_reduction_JR}.
Using \cref{lem:elementary_computations}, we can use the transition matrix found
after a reduction in the first half of the matrix to directly reduce the second half
of the matrix.
This means that in symplectic reduction, we have recursive calls only for the first $d_h$
steps of the tour.
These are the only modifications in our algorithm.
It is clear that, during the entire algorithm, the matrix $R$ is symplectic.

The notation used in this section
are the same as in~\cref{sec:algorithm_complexity}, with the notable
exception that we may use here a large $\varepsilon$---recall that it was fixed
to $1/2$ in all of \cref{sec:algorithm_complexity}. We also assume
that $\alpha>\sqrt{\log n_h\log \log n_h}$ for the sake of simplicity.
We use here the modified potential where we consider only the first half of
the matrix:
\[
\Pi=\sum_{i=1}^{d_h}
\left(d_h+1-i\right)\log\norm_{\KK_h/\QQ}(R_{i,i}).
\]

To complete the proof we need an \emph{experimentally validated
heuristic} on the repartition of the potential during the reduction.

\begin{heuristic}
  \label{heur:symplectic}
  The potential $\Pi$ is, at the end of \algName{Reduce}, always
  larger than the potential of an orthogonal matrix with the same
  volume.
\end{heuristic}

\begin{remark}
  This heuristic hinges on the fact the sequence of
  $\norm_{\KK_h/\QQ}(R_{i,i})$ is non-increasing, which is always the case
  in practice for random lattices.
\end{remark}

We now give a better bound on the increase in bit sizes, which is a
refinement of \cref{lem:sum_of_precisions}. The proof is done in the
exact same manner.
\begin{lemma}
\label{lem:lemma10}
  Suppose the input matrix $M$ is a descent of a $2\times 2$ triangular
  matrix $\begin{pmatrix} u & v \\ 0 & w \end{pmatrix}$, where the diagonal
  elements have been balanced in the sense of
  \cref{thm:unites}. Under \cref{heur:symplectic}, the sum of all bit
  sizes used in the
  recursive calls at the top-level is at most
  \[pd_h^2\left(1+\frac{1}{\varepsilon}\right)
  \left(\frac12+\frac{1}{d_h}+
\bigO{\sqrt{\frac{\log n_h\log \log
n_h}{n_h}}}\right),\]
with
\[
p=\log\frac{\max_{\sigma:\KK_h\rightarrow\CC, R_{i,i}\in R}
  \sigma (R_{i,i})}{\min_{\sigma:\KK_h\rightarrow\CC, R_{i,i}\in R} \sigma (R_{i,i})} \geq n_hd_h,
\]
where the $\sigma$ runs over the possible field embeddings and the $R_{i,i}$
are the diagonal values of the $R$ part of the $QR$-decomposition of $M$.
\end{lemma}
\begin{proof}
  Without loss of generality, up to scaling, we can assume that
  \[\norm_{\KK_{h+1}/\QQ}(u)\norm_{\KK_{h+1}/\QQ}(w) =
  \norm_{\KK_h/\QQ}\left(\prod_i R_{i,i}\right)=1.\]
  Therefore, with our choice of $p$, we have at the beginning
\[ \|R_{i,i}\|\leq \|u\|\in 2^{p/2+\bigO{\sqrt{n_hd_h\log(n_hd_h)\log \log(n_hd_h)}}}. \]
  Thus we have :
\begin{align*}
	\Pi= & \frac{n_hd_h(d_h+1)}{4}\left(p+\bigO{\sqrt{n_hd_h\log(n_hd_h)\log \log(n_hd_h)}}\right) \\
	= &\frac{n_hd_h(d_h+1)}{4}p\left(1+\bigO{\sqrt{\frac{\log n_h\log \log
  n_h}{n_h}}}\right),
\end{align*}
since by hypothesis, $p>n_hd_h$.
  And then by, \cref{heur:symplectic}, we have $\Pi\geq 0$ at the end of
  the calls. When performing local reductions, as in the proof of
  \cref{lem:sum_of_precisions}, two cases can
   occur:
   \begin{itemize}
	   \item Either $\norm_{\KK_h/\QQ}(R_{j,j}^{(i)})\leq 2^{2(1+\varepsilon)\alpha
   n_h^2}\norm_{\KK_h/\QQ}(R_{j+1,j+1}^{(i)})$, and as mentioned in
       \cref{sec:heuristic} the local reduction is not performed,
       so that we can consider that we use here a zero precision
       call.
     \item Either a local reduction is actually performed and by
       the result of \cref{sec:well_conditioned}, we can
	   use a precision in $\bigO{p_{i,j}}$ with:
       \[
        p_{i,j} = \log\left(\frac{\max_{k} \sigma_k(R_{j,j}^{(i)})}{\min_k
         \sigma_k(R_{j+1,j+1}^{(i)})}\right),
        \]
	Let now set
	\[ L = \frac{\log(\norm_{\KK_h/\QQ} (R^{(i)}_{j,j}/R^{(i)}_{j+1,j+1}))}{n_h}.\]
	The value $p_{i,j}$ is, thanks  to
       the unit rounding~\cref{thm:unites} a
	   \[ L+\bigO{\sqrt{n_h\log n_h\log\log n_h}} ,\]
       by hypothesis.
       The reduction of
       this truncated matrix yields a unimodular
       transformation, represented with precision $\bigO{p_{i,j}}$, which
       when applied to the actual basis matrix implies that $\Pi$
       decreases by a term at least:
       \[
       \delta_{i,j} = n_h\left[\frac{L}{2}-\alpha
	   n_h\right]-{2^{-\Omega(p)}}
       \]
       by \cref{heur:lifting_size} and
       \cref{thm:well_conditioned}.
	   Let us bound the ratio $p_{i,j}/\delta_{i,j}$:
       \[
         \frac{p_i}{\delta_i} =
	   \frac{L+\bigO{\sqrt{n_h\log n_h\log\log
	   n_h}}}{\left(\frac{L}{2}-\alpha
		   n_h\right)n_h-2^{-\Omega(p)}}
	   =\frac{1+\frac{\bigO{\sqrt{n_h\log n_h\log\log n_h}}}{L}}
         {\frac{n_h}{2}-\frac{\alpha
         n_h^2}{L}-\frac{2^{-\Omega(p)}}{2L}}.
       \]
       Now recall that
       $\norm_{\KK_h/\QQ}(R^{(i)}_{j,j})\geq
       2^{2(1+\varepsilon)\alpha
       n_h^2}\norm_{\KK_h/\QQ}(R^{(i)}_{j+1,j+1})(1-2^{-\Omega(p)})$, the
       multiplicative error term coming from the precision at which the
       values of the $R^{(i)}_{j,j}$ and $R^{(i)}_{j+1,j+1}$ are approximated at
       runtime. Thus we have:
        \[
          \sqrt{n_h\log n_h\log\log n_h}/L = \bigO{\sqrt{\frac{\log
          n_h\log\log n_h}{n_h}}},
        \]
        and
        \[
  \alpha n_h^2/L \leq \frac{n_h}{2(1+\varepsilon)}.
        \]
        As such we have:
       \[
         \frac{p_{i,j}}{\delta_{i,j}} \leq
         \frac{1+\bigO{\sqrt{\frac{\log
         n_h\log\log
 n_h}{n_h}}}}{\frac{n_h\varepsilon}{1+\varepsilon}+\bigO{1/n_h}}.
       \]
     \end{itemize}
     The sum of precisions is therefore multiplied by
       \[
         d_h^2\left(1+\frac1\varepsilon\right)
         \left(\frac12+\frac1{2d_h}+\bigO{\sqrt{\frac{\log
         n_h\log \log n_h}{n_h}}}\right),
       \]
       which finishes the proof.
\end{proof}

We can now collect all the calls at each level to compute the global
complexity, for refining \cref{thm:complexity_fast_LLL}:
\begin{theorem}
\label{thm:sym}
Select an integer $f$ a power of $q=\bigO{\log f}$ and let
$n=\varphi(f)$. The complexity for reducing matrices $M$ of dimension
two over $\lL=\QQ[x]/\Phi_f(x)$ with $B$ the number of bits in the
input coefficients is heuristically
\[\tilde{O}\left(n^{2+\frac{\log((1/2+1/2q)(1+1/\varepsilon))}{\log
q}}B\right)\] and the first
column of the reduced matrix has coefficients bounded by
\[
	\exp\left(\bigO{n^{1+\frac{\log((1+\varepsilon)\frac{2q-1}{q})}{\log
	q}}}\right)\left|\norm_{\KK_h/\QQ}(\det
  M)\right|^{\frac{1}{2n}}.
\]
\end{theorem}
\begin{proof}
The proof is now exactly the same as for \cref{thm:complexity_fast_LLL}.
  We select a tower of cyclotomic subfields $\KK_h^\uparrow$ with $\KK_0=\QQ$,
  $[\KK_i:\QQ]=n_i$, $n_{i+1}/n_{i}=d_i=q$ for $i<h$ and $\KK_h=\lL$
  with $h=\log f/\log q$.
	According to \cref{rem:increase_alpha}, we can take
	\[ \alpha_i=\bigO{n_i\left((1+\varepsilon )\frac{2q-1}{q}\right)^{i}} \] and all
  our previous assumptions are fulfilled.

  The complexity at the level $i$ is $\bigO{q^5n_ip\log(Bn)}$ for precision
  $p$ but the sum on the precision over all calls is a:
  \[
      \bigO{B\prod_{j>i}
        \left(1+\frac{1}{\varepsilon}\right)\left(\frac12+\frac{1}{2q}+\bigO{\sqrt{
    \frac{\log n_{i}\log \log n_{i}}{n_i}}}\right)d_j^2},\]
    which simplifies in
    \[
      \bigO{B\left(\frac{n}{n_i}\right)^2
	\left(\frac{(1+\frac{1}{\varepsilon})(q+1)}{2q}\right)^{h-i}}.
    \]
  Summing over all $i$ gives the result.
\end{proof}

Selecting $\varepsilon=\log n$, and running the algorithm of~\cref{sec:fast_LLL_NF}
 on the output of the reduction analyzed in
\cref{sec:algorithm_complexity} gives:
\begin{corollary}
  Select an integer $f$ a power of $q=\bigO{\log f}$ and let
  $n=\varphi(f)$.
  The complexity for reducing matrices $M$ of dimension two over
  $\lL=\QQ[x]/\Phi_f(x)$ with $B$ the number of bits in the input
  coefficients is heuristically
  \[\tilde{\textrm{O}}\left(n^{2+\frac{\log(1/2+1/2q)}{\log
        q}}B\right)+n^{\bigO{\log
  \log n}}\] and the first column of the reduced matrix has coefficients
  bounded by
  \[
    2^{\tilde{\textrm{O}}(n)}\left|\norm_{\KK_h/\QQ}(\det
    M)\right|^{\frac{1}{2n}}.
  \]
\end{corollary}
Clearly, for $B=n^{\omega(1)}$, we can choose $\varepsilon=\omega(1)$ and
get a running time of \[n^{2+\frac{\log(1/2+1/2q)}{\log
q}+\littleO{1}}B.\]

We insist on the fact that for $B=n^x$, the above proof does not give an optimal running time.
This running time is given in~\cref{fig:compl-sympl}.
One can improve on the upper-bound by using a stronger (yet credible)
heuristic on $\Pi$, having only one reduction on each round where
$\norm_{\KK_h/\QQ}(R_{i,i}/R_{i+1,i+1})$ is maximized, an adaptive $\varepsilon$ and two different $q$ used.
Clearly, this algorithm can also be parallelized, but the maximum number of processor used is less than before.

\begin{figure}
  \begin{center}
    \begingroup
  \makeatletter
  \providecommand\color[2][]{
    \GenericError{(gnuplot) \space\space\space\@spaces}{
      Package color not loaded in conjunction with
      terminal option `colourtext'
    }{See the gnuplot documentation for explanation.
    }{Either use 'blacktext' in gnuplot or load the package
      color.sty in LaTeX.}
    \renewcommand\color[2][]{}
  }
  \providecommand\includegraphics[2][]{
    \GenericError{(gnuplot) \space\space\space\@spaces}{
      Package graphicx or graphics not loaded
    }{See the gnuplot documentation for explanation.
    }{The gnuplot epslatex terminal needs graphicx.sty or graphics.sty.}
    \renewcommand\includegraphics[2][]{}
  }
  \providecommand\rotatebox[2]{#2}
  \@ifundefined{ifGPcolor}{
    \newif\ifGPcolor
    \GPcolortrue
  }{}
  \@ifundefined{ifGPblacktext}{
    \newif\ifGPblacktext
    \GPblacktextfalse
  }{}
  \let\gplgaddtomacro\g@addto@macro
  \gdef\gplbacktext{}
  \gdef\gplfronttext{}
  \makeatother
  \ifGPblacktext
    \def\colorrgb#1{}
    \def\colorgray#1{}
  \else
    \ifGPcolor
      \def\colorrgb#1{\color[rgb]{#1}}
      \def\colorgray#1{\color[gray]{#1}}
      \expandafter\def\csname LTw\endcsname{\color{white}}
      \expandafter\def\csname LTb\endcsname{\color{black}}
      \expandafter\def\csname LTa\endcsname{\color{black}}
      \expandafter\def\csname LT0\endcsname{\color[rgb]{1,0,0}}
      \expandafter\def\csname LT1\endcsname{\color[rgb]{0,1,0}}
      \expandafter\def\csname LT2\endcsname{\color[rgb]{0,0,1}}
      \expandafter\def\csname LT3\endcsname{\color[rgb]{1,0,1}}
      \expandafter\def\csname LT4\endcsname{\color[rgb]{0,1,1}}
      \expandafter\def\csname LT5\endcsname{\color[rgb]{1,1,0}}
      \expandafter\def\csname LT6\endcsname{\color[rgb]{0,0,0}}
      \expandafter\def\csname LT7\endcsname{\color[rgb]{1,0.3,0}}
      \expandafter\def\csname LT8\endcsname{\color[rgb]{0.5,0.5,0.5}}
    \else
      \def\colorrgb#1{\color{black}}
      \def\colorgray#1{\color[gray]{#1}}
      \expandafter\def\csname LTw\endcsname{\color{white}}
      \expandafter\def\csname LTb\endcsname{\color{black}}
      \expandafter\def\csname LTa\endcsname{\color{black}}
      \expandafter\def\csname LT0\endcsname{\color{black}}
      \expandafter\def\csname LT1\endcsname{\color{black}}
      \expandafter\def\csname LT2\endcsname{\color{black}}
      \expandafter\def\csname LT3\endcsname{\color{black}}
      \expandafter\def\csname LT4\endcsname{\color{black}}
      \expandafter\def\csname LT5\endcsname{\color{black}}
      \expandafter\def\csname LT6\endcsname{\color{black}}
      \expandafter\def\csname LT7\endcsname{\color{black}}
      \expandafter\def\csname LT8\endcsname{\color{black}}
    \fi
  \fi
    \setlength{\unitlength}{0.0500bp}
    \ifx\gptboxheight\undefined
      \newlength{\gptboxheight}
      \newlength{\gptboxwidth}
      \newsavebox{\gptboxtext}
    \fi
    \setlength{\fboxrule}{0.5pt}
    \setlength{\fboxsep}{1pt}
\begin{picture}(8640.00,5760.00)
    \gplgaddtomacro\gplbacktext{
      \csname LTb\endcsname
      \put(726,704){\makebox(0,0)[r]{\strut{}$1.55$}}
      \put(726,1192){\makebox(0,0)[r]{\strut{}$1.6$}}
      \put(726,1681){\makebox(0,0)[r]{\strut{}$1.65$}}
      \put(726,2169){\makebox(0,0)[r]{\strut{}$1.7$}}
      \put(726,2657){\makebox(0,0)[r]{\strut{}$1.75$}}
      \put(726,3146){\makebox(0,0)[r]{\strut{}$1.8$}}
      \put(726,3634){\makebox(0,0)[r]{\strut{}$1.85$}}
      \put(726,4122){\makebox(0,0)[r]{\strut{}$1.9$}}
      \put(726,4611){\makebox(0,0)[r]{\strut{}$1.95$}}
      \put(726,5099){\makebox(0,0)[r]{\strut{}$2$}}
      \put(858,484){\makebox(0,0){\strut{}$1$}}
      \put(1679,484){\makebox(0,0){\strut{}$2$}}
      \put(2499,484){\makebox(0,0){\strut{}$3$}}
      \put(3320,484){\makebox(0,0){\strut{}$4$}}
      \put(4140,484){\makebox(0,0){\strut{}$5$}}
      \put(4961,484){\makebox(0,0){\strut{}$6$}}
      \put(5781,484){\makebox(0,0){\strut{}$7$}}
      \put(6602,484){\makebox(0,0){\strut{}$8$}}
      \put(7422,484){\makebox(0,0){\strut{}$9$}}
      \put(8243,484){\makebox(0,0){\strut{}$10$}}
    }
    \gplgaddtomacro\gplfronttext{
      \csname LTb\endcsname
      \put(4550,154){\makebox(0,0){\strut{}$B=n^x$}}
      \put(4550,5429){\makebox(0,0){\strut{}Complexity of lattice reduction is $n^y B$}}
      \csname LTb\endcsname
      \put(7256,4926){\makebox(0,0)[r]{\strut{}q=2}}
      \csname LTb\endcsname
      \put(7256,4706){\makebox(0,0)[r]{\strut{}q=4}}
      \csname LTb\endcsname
      \put(7256,4486){\makebox(0,0)[r]{\strut{}q=8}}
      \csname LTb\endcsname
      \put(7256,4266){\makebox(0,0)[r]{\strut{}q=16}}
      \csname LTb\endcsname
      \put(7256,4046){\makebox(0,0)[r]{\strut{}q=32}}
      \csname LTb\endcsname
      \put(7256,3826){\makebox(0,0)[r]{\strut{}q=64}}
      \csname LTb\endcsname
      \put(7256,3606){\makebox(0,0)[r]{\strut{}q=128}}
      \csname LTb\endcsname
      \put(7256,3386){\makebox(0,0)[r]{\strut{}q=256}}
    }
    \gplbacktext
    \put(0,0){\includegraphics{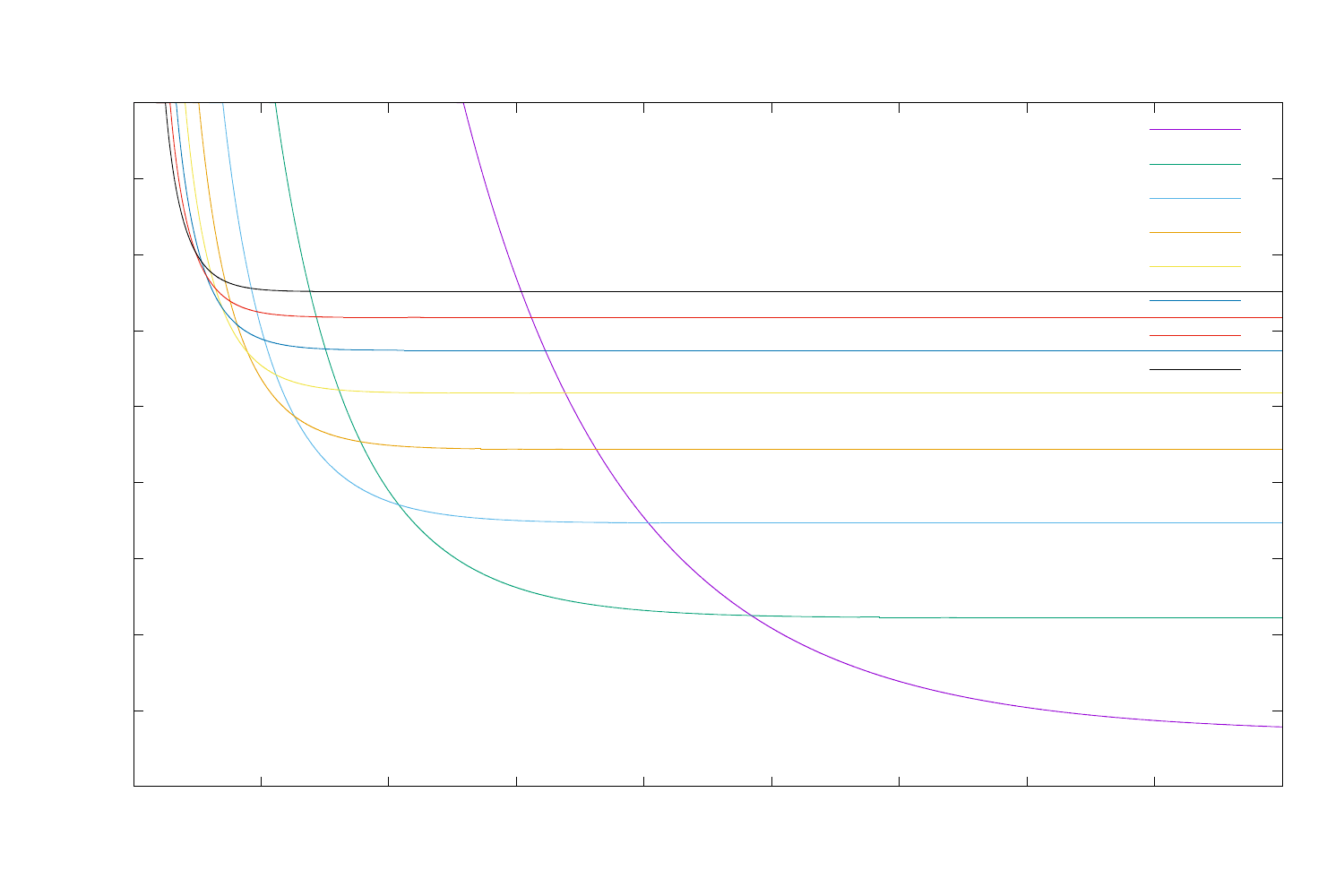}}
    \gplfronttext
  \end{picture}
\endgroup
   \end{center}
  \caption{Upper bound on the complexity of symplectic reduction}
  \label{fig:compl-sympl}
\end{figure}
 \section{Optimizations and Implementation}
\label{sec:implementation}

The algorithms detailed in \cref{sec:fast_LLL_NF},
\cref{sec:fastlll} and \cref{sec:symplectic} have been
implemented and tested. This section details various optimizations and
implementation choices, as well as gives an experimental assessment on the
heuristics used in the complexity proofs.

The first algorithm of~\cref{sec:fastlll} was used in the rational case
to solve NTRU instances in a previous paper by Kirchner and
Fouque~\cite{EC:KirFou17}, and found to perform as expected.

\subsection{On the choice of the base case}

Let $h>0$ be a non-negative integer. The
setting of the reduction is a tower of power-of-two cyclotomic fields
$\KK_h^\uparrow = \left(\QQ = \KK_0 \subset \KK_1 \subset \cdots \subset
\KK_h\right)$.

\subsubsection{Stopping the reduction before hitting $\ZZ$}

As stated in \cref{thm:complexity_fast_LLL}, the approximation factor
increases quickly with the height of the tower. However, if we
know how to perform a reduction over a number field above $\QQ$, say
$\KK_1$ for instance, directly, then there is no need to reduce up to
getting a $\ZZ$\nobreakdash-module and we instead stop at this level. Actually,
the largest the ring, the better the approximation factor becomes and the
more efficient is the whole routine. It is well-known that it is
possible to come up with a \emph{direct} reduction algorithm for an
algebraic lattice  when the underlying ring of integer is
\emph{norm-Euclidean}, as first mentioned by Napias in~\cite{Napias}.
The reduction algorithm over such a ring $\order_\KK$ can be done
exactly as for the classical \LLL{} algorithm, by replacing the norm
over $\QQ$ by the algebraic norm over $\KK$.
Hence a natural choice would be $\ZZ[x]/(x^n+1)$ with
$n\leq 8$ as these rings are proved to be norm-Euclidean.

\subsubsection{The ring $\ZZ[x]/(x^{16}+1)$}

However, it turns out that while $\KK = \ZZ[x]/(x^{16}+1)$ is not
norm-Euclidean, we can still use this as our base case. As such, we need to
slightly change the algorithm in case of failure of the standard algorithm.
Given $a,b$, we use the {\em randomized} unit rounding of $\sqrt{\{\mu\}}$
computed by~\cref{thm:unites} with $\mu=a/b$, which gives a unit $u$ such
that $u^2\{\mu\}$ is round.  We accept the change if \[
\norm_{\KK/\QQ}(a-b(\lfloor \mu \rceil+\lfloor u\{\mu\}\rceil
u^{-1}))<\norm_{\KK/\QQ}(a) \] and restart up to a hundred times if
it fails. \medskip

This algorithm restarts on average $0.7$ times and fails every $50000$
times. On failure, one can for example use a more complicated approach;
but as long as the number of bits is not gigantic, we can simply stop
there since the other reductions around the two Gram-Schmidt norms will
randomize everything and the algorithm can smoothly continue.
The terms $a,b$ tend to slowly accumulate a unit contribution when
$n\geq 4$, and it is therefore needed to rebalance them using randomized
rounding. For $n=16$, this happens on average every $50$ times.

\subsubsection{Comparison between the base fields}
We give in the~\cref{table:basecase} the properties of the various
possible base cases between the dimension 1 over $\QQ$---that is $\QQ$
itself---and 16, as described above.

\begin{table}[]
  \sffamily
  \caption{\small{Lattice reduction with root factor $\alpha$ in dimension $d$
    over $\ZZ$ gives an element of $\Lat$ of norm around $\alpha^{d/2}
    \covol(\Lat)^{1/d}$.
    After $k$ steps in the Euclidean algorithm with norm factor $\beta$,
    the norm of the elements is roughly divided by $\beta^k$.
  Both are for random inputs.}
  }
  \label{table:basecase}
  \setlength{\tabcolsep}{6pt}
  \vspace{4pt}
  \centering
  \renewcommand*{\arraystretch}{1.3}
  \begin{tabular}{ccc}
    Dimension & Root factor & Norm factor \\
    \midrule
    \rowcolor{black!10}
    1 & 1.031 & 4.6 \\
    2 & 1.036 & 7.1 \\
    \rowcolor{black!10}
    4 & 1.037 & 17 \\
    8 & 1.049 & 26 \\
    \rowcolor{black!10}
    16 & 1.11 & 24 \\
    \bottomrule
  \end{tabular}
\end{table}

\begin{remark}
  We need the base case to be (relatively) fast in our implementation.
  We followed the standard divide-and-conquer strategy: we first reduce
  the input matrix with half the precision, apply the transition matrix,
  and reduce the rest with about half the precision.
\end{remark}

\subsection{Decreasing the approximation factor}

In several applications, it is interesting to decrease the approximation factor.
Our technique is, at the lowest level of recursion, and when the number of
bits is low, to use a \LLL-type algorithm. Each time the reduction is finished,
we descend the matrix to a lower level where the approximation factor is lower.

Remark that the unit rounding is, at least theoretically, mandatory.
In particular, a swap when the basis is not reduced with the definition
in~\cite{IMA:KimLee17} may not lead to a reduction in potential
so that the proof of~\cite[Theorem 3]{IMA:KimLee17} is
incorrect.  We also point out that without a bound on the unit
contributions, we have no polynomial bound on the number of bits used in
their algorithm 3.

From a practical point of view, this does not seem to be a problem. If
this is the case, our algorithm
can be used every time we have a reasonable tower of number fields.

\subsection{Lifting a reduction}
\label{sec:lifting_problems}
One might expect that, as soon as the ideal generated by all the
$\norm_{\lL/\KK}(a_i)$ and $\norm_{\lL/\KK}(b_i)$ is $\order_\KK$, that
for most of the small $x\in
\order_\lL$, we
would have
\[
  \norm_{\lL/\KK}(\inner{a}{x})\order_\KK+\norm_{\lL/\KK}({\inner{b}{x})}\order_\KK=\order_\KK.
\]
There is, however, a profusion of counterexamples to this and the
algorithm often stumbles on them. This implies that the lift of
a short vector can actually be quite large, depending on the
norm of the ideal generated by the elements
$\norm_{\lL/\KK}({\inner{a}{x})}$ and $\norm_{\lL/\KK}({\inner{b}{x})}$.
A solution which practically works is to increase
the number of short vectors we consider in the lifting phase: instead of
lifting one vector, we lift multiple of them.  As such, the lift step
never causes problem when we are reducing a random lattice.
In our experiments with random lattices, the average number of lifted
vectors is around $1.5$.

When the lattice is not random, for example with a short planted
element, it sometimes completely fails: at each round in the algorithm,
the lift will return a long vector even if the recursive reduction found
plenty of short ones. While this may not be a problem for some
applications -- finding a short vector in a NTRU lattice implies an
ability to decrypt -- it is an important one for others.  Our proposed
solution to this difficulty is to use a pseudo-basis instead of a basis.
Indeed, it is a standard fact that the first element can be lifted into
a unimodular pseudo-basis~\cite[Corollary 1.3.5]{cohen2012advanced}.
Of course, we need to have a fast ideal arithmetic and to keep the ideals
of small norm, which is neither easy nor fast and will be the subject of
a future work.

\subsection{Other details}

The program was written in the interpreted language
\texttt{Pari/GP}~\cite{batut1998pari}.  It uses the native functions for
multiplying field elements, which is not at all optimal, and even more
so when we multiply matrices.  Only the recursive calls were
parallelized, and not the Gram-Schmidt orthogonalization nor the size
reduction, which limits the speed-up we can achieve in this way.
We used the Householder method for the QR decomposition.  The symplectic
optimization was used at each step, and was not found to change the
quality of the reduction\footnote{
Gama, Howgrave-Graham and Nguyen~\cite{EC:GamHowNgu06} found instead that it gave a
``smoother (better)'' basis, showing a significant difference in their Figure 1.
An other version of the paper does not include this comment, and their
(perplexing) Figure 1 shows no difference in the exponential decrease
of the Gram-Schmidt norms.
}.
We did not use the algorithm of~\cref{sec:fastlll} inside the recursion
of~\cref{sec:fast_LLL_NF}.
We chose a number of rounds of $d^2$ for all but the first level.
 \section{Applications}
\label{sec:applications}

\subsection{Attacks on multilinear maps}

In 2013, a construction for cryptographic multilinear maps was
announced~\cite{EC:GarGenHal13} with a heuristic security claim.
An implementation of an optimization of the scheme was later
published~\cite{AC:ACLL15}; however some of its uses, in particular
involving an encoding of zero, were broken~\cite{EC:HuJia16}.
Subsequently, subfield attacks showed that the previous choice of
parameters was unsafe~\cite{C:AlbBaiDuc16,ANTS:Cheon16,EC:KirFou17}, but
these attacks were only asymptotical due to the extremely
large dimension and length of the integers involved.

The improved scheme~\cite{AC:ACLL15} gives encoding of the form $u_i=e_i/z \Mod q$ where $\|e_i\|$ is around
\[ 28\me N^4 \log(N)^{3/2} \sqrt{\pi\log(8N)} \]
in the ring $\ZZ[x]/(x^N+1)$ with $N$ a power of two.
The attack, attributed to Galbraith, consists in computing $u_1/u_2=e_1/e_2$ and recovering short vectors in
\[ \begin{pmatrix} q & u_1/u_2 \\ 0 & \Id_N \end{pmatrix} \]
which is obviously solving a NTRU-like problem.

The present work revisits the results of the attacks presented
in~\cite{EC:KirFou17}: many instances can be broken even with a high
approximation factor. A simple instance is with $N=2^{16}$ and $q\approx
2^{6675}$, rated at the time at 56 bits of security~\cite[Table
1]{AC:ACLL15}.
We compute the norm of $e_1/e_2$ over $\ZZ[x]/(x^n+1)$ with $n=2^{11}$
and solve the lattice problem over this smaller field.
It took 13 core-days and 4 wall-time days to compute a solution.  There
are few running times of lattice reduction with high approximation
factor on hard instances in the literature.
It was reported in 2016~\cite[Table 6]{C:AlbBaiDuc16} that the same
problem with $n=2^8$ and $q\approx 2^{240}$ takes 120 (single-threaded)
hours with \textsf{fplll}~\cite{EC:NguSte05}.
As the complexity of their implementation is roughly proportional to
$n^4\log(q)^2$ we can estimate a running time of $40000$ years, or
$4000000$ times slower than the algorithm presented in this work.
This is the largest hard instance\footnote{
	There are easy instances with a larger dimension, for example in~\cite{EC:GamHowNgu06}.
	They considered a NTRU instance with degree $317$ and modulus $128$, and reduced it in $519$ seconds.
	The low modulus implies that we only have to reduce the middle dimension $90$ matrix, which \textsf{fplll} reduces in $0.2$ second.
}
of lattice reduction that we found in the literature.

\subsection{Gentry-Szydlo algorithm}

The fast reduction procedure for cyclotomic ideals can be used to build a fast
implementation of the Gentry-Szydlo algorithm~\cite{EC:GenSzy02}.
This algorithm retrieves, in polynomial time, a generator of a
principal ideal $f\order_\KK$ given its relative norm $f\conj{f}$ in
cyclotomic fields, or more generally in CM fields. This algorithm is a
combination of algebraic manipulations of ideals in the field and lattice
reduction.

\subsubsection{Gentry-Szydlo.}
\label{subsec:gs}
In this section, we briefly recall the crux of the Gentry-Szydlo
algorithm~\cite{EC:GenSzy02}. This algorithm aims at solving the following
problem, presented in its whole generality:

\begin{problem}[Principal ideal problem with known relative norm]
  Let $\lL$ be a CM-field, of conjugation $x\mapsto\conj{x}$, and denote by
  $\lL^+$ its maximal totally real subfield. Let $f\in\order_\lL$ and set
  $\mathfrak{f} = f\order_\lL$, the ideal spanned by this algebraic integer.
  \begin{description}
    \item[Input] The relative norm $\norm_{\lL^+/\QQ}(f) = f\conj{f}$
      and a $\ZZ$-basis of the ideal $\mathfrak{f}$.
    \item[Output] The element $f$.
  \end{description}
\end{problem}

We can use the reduction of an ideal as follows:
from $\mathfrak{f}$ and $f\conj{f}$ we start by reducing the
$\order_\lL$-lattice \[ \frac{f\order_\lL}{\sqrt{f\conj{f}}}, \]
of volume $\sqrt{|\Delta_\lL|}$ and find an element of the shape
$fx$ where $x\in \order_\lL$ and is small: $\|x\|=2^{\bigOtilde{n}}$.
Now we have that:
\[ \mathfrak{f}=\frac{f\conj{f}}{\conj{fx}}\cdot \conj{x}\order_\lL \]
We also have $x\conj{x}=\frac{fx\conj{fx}}{f\conj{f}}$ so that we have reduced
the problem to the smaller instance
$\left(\conj{x}\order_\lL, x\conj{x}\right)$.

For the sake of simplicity, we give here the outline of the remaining part of
the algorithm for a cyclotomic field of conductor a power of two. The
algorithm selects an integer $e$ such that $f^e \mod r$ is
known with a large $r$. Binary exponentiation with the above reduction
computes a $x\order_\lL$ with a short $x\in \order_\lL$ and such that
\[ f^e=Px \]
with $P$ known (and invertible) modulo $r$ and $q^k$.
Now we can deduce $x \Mod r$ and since $x$ is small, we know $x$.

The last step is to extract an $e$-th root modulo $q^k$.
We choose $q$ such that $q\order_\lL=\mathfrak{q}\mathfrak{\conj{q}}$ which always exists in power of two cyclotomic fields since $(\ZZ/2n\ZZ)^\times/\{-1,1\}$ is cyclic.
Extracting $e$-th root modulo $\mathfrak{q}$ is easy, as $e$ is smooth.
There are $\gcd(e,q^{n/2}-1)$ such roots, and we can choose $q$ such that for each $p|e$ with $p$ not a Fermat prime, $q^{n/2}\neq 1 \Mod p$.
If we choose $f \Mod \mathfrak{q}$ as a root, then we know $\conj{f} \Mod \conj{\mathfrak{q}}$, and we also know $f\conj{f}$ so we can deduce $f \Mod \conj{\mathfrak{q}}$.
As a result, we know $f \Mod q$ and Hensel lifting leads to $f \Mod q^k$.
For $k$ sufficiently large, we recover $f$.

We choose $e$ to be the smallest multiple of $2n$, such that $r$, the product of primes $p$ such that $2n|p-1|e$, is sufficiently large.
One can show~\cite{EPRINT:Kirchner16} that $\log e=O(\log n \log \log n)$ is
enough and heuristically taking $e$ as the product of $n$ and a primorial reaches
this bound.
\subsubsection{Faster multiplication using lattice reduction.}
The bottleneck of the Gentry-Szydlo algorithm is to accelerate
the ideal arithmetic. We represent ideals with a small family of
elements over the order of a subfield $\order_{\KK}$. One can represent
the product of two ideals using the family of all products of
generators. However, this leads to a blow-up in
the size of the family. A reasonable approach is simply to sample a bit
more than $[\lL:\KK]$ random elements in the product so that with
overwhelming probability the ideal generated by these elements is
the product ideal itself. It then suffices to reduce the corresponding
module to go back to a representation with smaller generators.

An important piece is then the reduction of an ideal itself. Our
practical approach is here to reduce a square matrix of dimension
$[\lL:\KK]$, and every two rounds to add a new random element with a
small Gram-Schmidt norm in the ideal at the last position.
With these techniques, the overall complexity of the Gentry-Szydlo now becomes
a $\tilde{O}(n^3)$.

In our experiment, we reduce up to $1.05^n$ (respectively $1.1^n$) the first ideal
to accelerate the powering with $n\leq 512$ (respectively $n=1024$).
The smallest $e$ such that this approximation works at the end was chosen.
The other reductions are done with an approximation factor of $2^{n/5}$ (respectively $2^{n/3}$).

\begin{table}[]
  \sffamily
  \caption{Implementation results}
  \label{table:implmentation}
  \setlength{\tabcolsep}{6pt}
  \vspace{4pt}
  \centering
  \renewcommand*{\arraystretch}{1.5}
  \begin{tabular}{cccc}
  Dimension & $e$ & Running time & Processor \\
    \midrule
    \rowcolor{black!10}
    256 & 15360 & 30 minutes & Intel i7-8650 (4 cores) \\
    512 & 79872 & 4 hours & Intel i7-8650 (4 cores) \\
    \rowcolor{black!10}
    1024 & 3194880 & 103 hours & Intel E5-2650 (16 cores) \\
    \bottomrule
  \end{tabular}
\end{table}
We emphasize that the implementation hardly used all cores: for example,
the total running time over all cores in the last case was 354 hours.

The runtime of the first implementation published~\cite{EC:BEFGK17} in
dimension 256 was 20 hours.
Assuming it is proportional to $n^{6}$ leads to an estimate of 10 years
for $n=1024$, or $800$ times slower than our algorithm. Our
practical results are compiled in~\cref{table:implmentation}.

There are applications in cryptography of this algorithm, such as when
some lattice-based cryptography has a
leak~\cite{EC:GenSzy02,CCS:EFGT17,C:AlbBaiDuc16},
for finding a generator of an ideal~\cite{EC:BEFGK17}, for solving a
norm equation~\cite{howgrave2004method} and
for solving geometric problems on
ideals~\cite{EC:GarGenHal13,EPRINT:Kirchner16}.

 \section{Conclusion} Through this article, we presented efficient
\LLL~variants to reduce lattices defined over the ring of integers  of
cyclotomic fields,  by exploiting the recursive structure of  tower of
cyclotomic subfields.
Our first algorithm has a complexity close to the number of swaps
$\bigO{n^2\cdot B}$  in \LLL~and the last one also exploits the symplectic
symmetries naturally present in such towers. In this last case, we
show that we can beat the natural lower bound on the number of
swaps required to perform a reduction. One caveat of our algorithms is that their approximation factors are worse than the classical \LLL~approximation factor. However, such algorithms can be useful for
some applications such as breaking graded encoding schemes or manipulating ideals, as in the Gentry-Szydlo algorithm.
We implemented all our algorithms and their performances are close to the complexities that we proved under some mild assumptions.
In particular, our implementation can use large base cases, that is all power of two cyclotomic fields of dimension $\leq 16$.

This work raises several questions.
First of all, on the need to rely on the introduced heuristics to prove the
complexity. It is possible to remove them by using the pseudo-basis
representation
of modules over Dedekind rings, and will be the matter of a subsequent
work.
Second, we can wonder about the actual complexity of the symplectic algorithm for low bitsize and on the eventuality of decreasing the approximation factor: is it possible to recover the original \LLL{} approximation factor while keeping the complexities of our fast variants?
Third, our lattice reduction algorithm suggests that the algorithms for reducing lattices on polynomial rings may not be optimal~\cite{giorgi2003complexity}, and in particular that an efficient algorithm with coarse-grain parallelism exists.
Another interesting research direction is to design a faster reduction for  lattices with a block-Toeplitz structure,  which appear in Coppersmith's algorithm~\cite{JC:Coppersmith97}.

Finally, using the symplectic structure we can remark that we can halve the
complexity of the \DBKZ{} algorithm when the block size is less than $n$.
We leave as an open problem the question of how to use similar techniques for
larger gains.

\subsection*{Acknowledgement} We thank Bill Allombert for his help in the
parallelization of the program.

 \bibliographystyle{abbrv}

\newpage
\appendix
\section{Bounding precision}
\label{app:precision}

In this section, we give details on the precision required in our
algorithms. We first indicate the loss of precision of elementary
operations, then look at the precision and complexity of the $QR$
decomposition, and finally the size-reduction procedure. The last part
indicates how to use fast matrix multiplication to reach the same goal.
We recall that $w$ is the number of bits in the words.

\subsection{Elementary operations}

\subsubsection{Fast computation of primitive roots of unity}
The fast Fourier transform algorithm needs a precise approximation of
the primitive roots of unity to be performed in fixed-point
arithmetic. In order to compute with high precision a primitive $f$-th
root of unity, one can use Newton's method where we start with
$1+6.3i/f$. The following lemma ensures that the convergence, in this
case, is at least quadratic.

\begin{lemma}
  Let $x\in\CC$ such that $|x|\geq 1-\frac{1}{2f}$, then by setting
  $x'=x-\frac{x^f-1}{fx^{f-1}}$ and with $\zeta^f=1$, we have:
  \[ |x'-\zeta|\leq f|x-\zeta|^2 \]
\end{lemma}
\begin{proof}
  Without loss of generality, by dividing everything by $\zeta$, we can assume $\zeta=1$.
  We then have the following equality:
  \[
    \frac{x'-1}{(x-1)^2}=
    \frac{(fx^{f-1}(x-1)-x^f+1)(x-1)^{-2}}{fx^{f-1}}=\frac{\sum_{k=1}^{f-1}
    kx^{k-1}}{fx^{f-1}}
  \]
  Applying the triangular inequality gives:
  \[
    \left|\frac{x'-1}{(x-1)^2}\right|\leq
    \frac{f(f-1)\max(1,|x|^{f-1})}{2f|x|^{f-1}}\leq
    \frac{1}{2}f\max(1,|x|^{1-f}).
  \]
  We can conclude by noticing that $(1-\frac{1}{2f})^{-f}\leq
  (1-1/6)^{-3}<2$.
\end{proof}
For $f\geq 128$, it is now easy to show that the sequence converges
towards $\exp(2i\pi/f)$; the finite number of remaining cases are easily
done by direct computations.

\subsubsection{A bound on the loss when iterating unitary matrices}

We now show the following elementary lemma on the iterations of
matrix-vector computations, which states that the error made when
computing chained matrix-vector multiplications can be controlled.
\begin{lemma}
  \label{lem:prec_composition_matrix}
  Let $A_i$ be a family of $k$ unitary matrices. Suppose that for each
  of these matrices $A_i$ there exists an algorithm $\mathcal{A}_i$ that
  given some vector $x$, outputs $A_ix$ within a certain vector of
  errors $e$ such that $\|e\|\leq \epsilon \|x\|$ with $\epsilon\leq
  \frac{1}{2k}$.  Then, the algorithm which computes $(\prod_i A_i)x$ by
  composing the algorithms $\mathcal{A}_i$ returns $(\prod_i A_i)x$
  within an error vector $e$ such that $\|e\|\leq 2k\epsilon\|x\|$.
\end{lemma}
\begin{proof}
  Let $B=\prod_{i=2}^k A_i$ and $Bx+e'$ the error committed using the algorithms $\mathcal{A}_i$.
  The algorithm $\mathcal{A}_1$ outputs $A_1(Bx+e')+e$, so that the
  error committed towards $A_1Bx$ is
  \[ \|A_1(Bx+e')+e-A_1Bx\|\leq \|e'\|+\|e\|\leq \|e'\|+\epsilon\|Bx+e\|  \]
  We now prove by induction that this error is less than $\left((1+\epsilon)^k-1\right)\|x\|$ with:
  \begin{align*}
    \|e'\|+\epsilon\|Bx+e\|\leq & \left((1+\epsilon)^{k-1}-1\right)\|x\|+
    \epsilon\left(\|x\|+\left((1+\epsilon)^{k-1}-1\right)\|x\|\right) \\
        = & \left((1+\epsilon)^k-1\right)\|x\|.
    \end{align*}
    The case $k=1$ is immediate and $(1+\epsilon)^k-1 < 2k\epsilon$ for $\epsilon<\frac{1}{2k}$ finishes the proof.

\end{proof}

\subsubsection{Analysis of the Discrete Fourier transform}

We now show how to efficiently compute a close approximation of
a Fourier transform. Indeed, the fast Fourier transform on
$2^n$ points correspond to a product of $n$ unitary matrices, so that
we can get $p$ bits of precision using a precision in $\bigO{p+\log n}$
by \cref{lem:prec_composition_matrix}.
Using this, we obtain an algorithm to multiply integers with $B$ bits with complexity $O(B/w\cdot \log(B/w))=O(B)$.

Bluestein's algorithm~\cite{bluestein1970linear} for Chirp-Z transform
reduces discrete Fourier transform in any size to the computation of
fast Fourier transform over power-of-two so that the same holds. Recall that
Inverse Fourier transform can also be computed from a discrete Fourier
transform.

All in all, we can evaluate the corresponding Fourier isomorphism and its
inverse:
\[ \RR[x]/(\Phi_f) \cong \CC^{\varphi(f)/2} \]
with limited loss in precision.

The complexity of this computation is a $\bigO{np+n\log n\cdot p/w}=\bigO{np}$ for $p=\Omega(w+\log n)$ with $n=\varphi(f)$. Indeed it
breaks down as:
\begin{itemize}
  \item Write the coefficients as polynomials with register-size
    coefficients and compute their Fourier transform with a cost of $\bigO{np}$
  \item Compute $\bigO{p/w}$ convolutions with Fourier transforms
    of size $\bigO{n}$
  \item Compute the inverse transform and
    propagate the carries for a running time of $\bigO{np}$.
\end{itemize}

(A modular implementation is probably faster if
$n$ is not tiny.)

In the general case, one would have to precompute the roots and use
product and remainder trees~\cite{moenck1972fast}.

\subsection{Householder orthogonalization}

The Householder orthogonalization algorithm transforms a complex matrix
$A$ into a product of $QR$, with $Q$ unitary and $R$ upper-triangular.
$Q$ is formed as a product of unitary reflections, which are all of the
type $\id-2v\conj{v}^t$ for certain vectors $\|v\|=1$.

The vector $v$ corresponding to the first symmetry is chosen so that the
first column of $R$ has only its first coordinate to be non-zero.
The algorithm then applies this unitary operation to the matrix $A$ and
recursively orthogonalize the bottom-right of this new matrix.

More precisely, denote by $a$ the first column of the matrix $A$.
As such, the first column of $R$ will be the vector
\[r = \left(-\|a\|\cdot \frac{a_1}{|a_1|},0,\dots,0\right)^t, \]
with the quotient $\frac{a_1}{|a_1|}$ set to $1$ if $a_1=0$.
Then with $v=\frac{a-r}{\|a-r\|}$ and $Q=\id-2v\conj{v}^t$, we have
that:
\[
  Qa=a-2\frac{(a-r)\conj{(a-r)}^ta}{\|a-r\|^2}=a-\frac{2(\|a\|^2-\conj{r}^ta)}{\|a-r\|^2}(a-r)
\]
We now use the fact that $\conj{a}^tr\in \RR$ and $\|r\|=\|a\|$ to get:
\[
  2(\|a\|^2-\conj{r}^ta)=\|a\|^2-\conj{r}^ta-\conj{a}^tr+\|r\|^2=\|a-r\|^2
\]
so that $Qa=r$.

The sign in the definition of $r$ implies that $\|a-r\|\geq \|a\|$ so
that we can compute $v$ with the precision used to handle $a$.

If we use $p>\omega(\log d)$ bits of precision, we can multiply by
$\id-2v\conj{v}^t$ with a relative error of $\bigO{d2^{-p}}$.
Using \cref{lem:prec_composition_matrix}, since we are performing $d$
symmetries, each column is computed with a
relative error of at most a $\bigO{d^22^{-p}}$.
Hence, with $\hat{Q}$ the matrix output by the algorithm, each column
of $\conj{Q}^tA$ has a relative error of $\bigO{d^22^{-p}}$ with respect
to the computed $R$.
This implies that there exists a matrix $A'$ where each column is $A$
within a relative error of $\bigO{d^22^{-p}}$, and whose $R$-factor in
the QR decomposition is the returned $R$.  Remark that the returned
$R_{i,i}$ may not be real. While this is usually not a problem, $R$ has
to be multiplied on the left by a diagonal unitary matrix to obtain {\em
the} QR-decomposition.\medskip

We define the conditional number of $A$ as $\kappa(A)=\|A\|\|A^{-1}\|$.
We can bound the stability of the QR decomposition~\cite{sun1991perturbation}:

\begin{theorem}\label{thm:stabqr}
  Given a matrix $A$, let $R$ be the $R$-factor of its QR decomposition.
  For the matrix $A+\delta A$, let $R+E$ be the $R$-factor of its QR
  decomposition.
  Then: \[ \|E\|\leq 3\kappa(A)\|\delta A\| \]
  provided that $\kappa(A)\frac{\|\delta A\|}{\|A\|}<1/10$.
\end{theorem}

\begin{proof}
  Let $A=QR$ be the QR-decomposition.
  Without loss of generality, we assume $\|A\|=1$.
  For a technical reason, we study the problem with $\delta A$ a linear function where $\delta A(1)$ is the wanted matrix, which means that other quantities such as $E$ are also functions.

  We now obtain:
  \[ \conj{\left(A+\delta A\right)}^t(A+\delta A)=\conj{A}^tA+\conj{\delta A}^tA+\conj{A}^t\delta A+\conj{\delta A}^t\delta A\]
  which is equal to:
  \[ \conj{\left(R+E\right)}^t(R+E)=\conj{R}^tR+\conj{E}^tR+\conj{R}^tE+\conj{E}^tE \]
  so we deduce:
  \[ \conj{E}^tR+\conj{R}^tE+\conj{E}^tE=\conj{\delta A}^tA+\conj{A}^t\delta A+\conj{\delta A}^t\delta A. \]
  We multiply by $\conj{A}^{-t}$ on the left and $A^{-1}$ on the right:
  \[ \conj{A}^{-t}\conj{E}^t\conj{Q}^t+QEA^{-1}+\conj{A}^{-t}\conj{E}^tEA^{-1}=\conj{A}^{-t}\conj{\delta A}^t+\delta AA^{-1}+\conj{A}^{-t}\conj{\delta A}^t\delta AA^{-1}. \]
  With $\rho=\|EA^{-1}\|$ and $\epsilon=\|\delta AA^{-1}\|$, we take the norm and get the inequality:
  \[ \rho-\rho^2\leq 2\epsilon+\epsilon^2 \]
  so that for $\epsilon<1/10$ we have $\rho\leq 3\epsilon$ if $\rho<1/2$.

  We now have to exclude the case $\rho>1/2$, which we do with a topological argument.
  It is clear from the algorithm that the QR-decomposition is continuous over invertible matrices.
  Since \[ \|A^{-1}(A+\delta A(t))-\id\|\leq \|A^{-1}\|\|\delta A(t)\|<1/2 \] for $0\leq t\leq 1$, we have that $A+\delta A$ is invertible and therefore $\rho$ is continuous over $[0;1]$.
  As $\rho(0)=0$ and $\rho([0;1])$ is connex, we get $\rho(1)<1/2$.

  Finally, $\|E\|\leq \|EA^{-1}\|\|A\|=\rho$ gives the result.
\end{proof}

Combining these results, we get:
\begin{theorem}
  Given a matrix $A$, we can compute the $R$-factor of its QR
  decomposition in time \[\bigO{\frac{d^3p}{w}+d^3+d^2p}\] with a
  relative error
  of \[ \bigO{\kappa(A)d^22^{-p}} \] if this is smaller than a constant.
\end{theorem}

We can, of course, decrease the $3$ in the exponent to a few matrix multiplications using aggregated Householder transformations and a divide-and-conquer algorithm, see~\cite[Subsection 18.4]{higham2002accuracy}.
This is also at the end of the appendix.

\subsection{Size-reduction}

We first consider the size-reduction for unitriangular matrices (i.e.
upper triangular matrices with ones on the diagonal).
Such a matrix $A$ is said to be size-reduced if both $A$ and $A^{-1}$
are small.

\begin{lemma}
  Let $A$ be a unitriangular matrix of dimension $d$ with coefficients
  in $\KK = \QQ[\zeta_f]$, such that its coefficients in the power basis
  are bounded in absolute value by $1$.
  Then $\|A\|\leq dn^{3/2}$ and $\|A^{-1}\|=(2n)^{\bigO{d}}$ with
  $n=\varphi(f)$.
\end{lemma}
\begin{proof}
  It is clear that $\|A_{i,j}\|\leq \sqrt{nf}\leq n^{3/2}$ so that
  $\|A\|\leq dn^{3/2}$.
  Now let $x$ be a column of $A^{-1}$.
  Consider a $i$ which maximizes $\|x_i\|(2n^{3/2})^i$.
  Then we have
  \[ 1\geq \|(Ax)_i\|\geq \|x_i\|-\sum_{j>i} \|A_{i,j}\|\|x_j\| \geq
  \|x_i\|\left(1-\sum_{j>i} \frac{n^{3/2}}{(2n^{3/2})^{j-i}} \right)>
  \|x_i\|/3
  \]
  and we obtain $\|x_i\|\leq 3$ which gives $\|x\|\leq
  3(2n^{3/2})^{d-1}\sqrt{d}$.
\end{proof}

We can finally prove our size-reduction theorem:
\begin{theorem}
  Let $A$ be a matrix of dimension $d$ with coefficients in
  $\KK=\QQ[\zeta_f]$, and $n=\varphi(f)$.
  We are given $p$, where $\|A\|,\|A^{-1}\|\leq 2^p$ and also
  $\sqrt{n\log n\log\log n}+d\log n<p$.
  In time  $\bigO{d^3np/w+d^2pn\log d}$,
  we can find an integral triangular matrix $U$ with $U_{i,i}\in
  \order_\KK^\times$ and a matrix $R+E$ such that $\|E\|\leq 2^{-p}$,
  with $R$ the R-factor of the QR decomposition of $AU$ and
  \[ \kappa(AU)\leq \left(\frac{\max_i \norm_{\KK/\QQ}(R_{i,i})}{\min_i
  \norm_{\KK/\QQ}(R_{i,i})}\right)^{1/n}2^{\bigO{\sqrt{n\log n\log\log n}+d\log n}}.\]
  We also have $\|U\|\in 2^{\bigO{p}}$
\end{theorem}
\begin{proof}
  In the canonical basis of $\KK$ repeated $d$ times, $A$ corresponds to a
  $d\times d$ block matrix, where each block is a diagonal complex matrix of
  size $n/2\times n/2$, so that the QR decomposition can be
  obtained from $n/2$ complex QR decompositions of dimension $d$.
  We can transform into (and from) this basis at a cost of
  $\bigO{d^2pn}$; and the same technique can be used with the
  size-reduction algorithm.

  The algorithm computes $R'$, the R-factor of the QR decomposition of
  $A$. Then we use~\cref{alg:size_reduce} on $R'$ which returns a
  $U$, and the algorithm returns $U$ and $R'U$.

  We have that
  \[ \|AU\|\leq d\sum_i \|R_{i,i}\|\leq d^2\|A\|
   \]
  so that $\|U\|\leq \|A^{-1}\|\|AU\|\leq d^22^{2p}$.
  As a result, we can use a precision of $\bigO{p}$ bits.

  Let $D$ be the diagonal of $R$.
  We have $\kappa(AU)=\kappa(R)\leq \kappa(D)\kappa(D^{-1}R)$.
  The reduction with units guarantees that \[ \kappa(D)\leq
  \left(\frac{\max_i \norm_{\KK/\QQ}(R_{i,i})}{\min_i
\norm_{\KK/\QQ}(R_{i,i})}\right)^{1/n}2^{\bigO{\sqrt{n\log n\log\log n}}}. \]
  The previous lemma gives $\kappa(D^{-1}R)=2^{\bigO{d\log n}}$.
\end{proof}

\subsubsection{On the reduction of well-conditioned matrices}
\label{sec:well_conditioned}

We finish this subsection with properties of lattices represented by a
well-conditioned matrix.
The following easy theorem indicates that if we want to reduce the
lattice generated by $A$, we can always truncate the matrix and work
with precision only $\bigO{\log(\kappa(A))}$.
The transition matrix which will be computed by the algorithm also needs
at most this precision.
Up to an irrelevant (small) quantity, this is of course a
\[\bigO{\log\left(\frac{\max_i \norm_{\KK/\QQ}(R_{i,i})}{\min_i
\norm_{\KK/\QQ}(R_{i,i})}\right)/n}.\]

\begin{theorem}
\label{thm:well_conditioned}
  Let $A$, $\delta A$ and $U$ an integer matrix such that $\|AU\|\leq
  \kappa\|A\|$, $\kappa(AU)\leq \kappa$ and \[ \frac{\|\delta
  A\|}{\|A\|}\leq \frac{\epsilon}{3\kappa^3} \]
  with $\epsilon<1/4$ and $\kappa\geq \kappa(A)$.
  Let $R$ be the $R$-factor of the QR-decomposition of $AU$ and $R+E$ be
  the one of $(A+\delta A)U$.
  Then $\|U\|\leq \kappa^2$ and
  \[ \frac{\|E\|}{\|A\|}\leq \epsilon. \]
\end{theorem}
\begin{proof}
  First $\|U\|\leq \|A^{-1}\|\|AU\|\leq \kappa\|A^{-1}\|\|A\|\leq
  \kappa^2$.
  Then $\|U\|\geq 1$ since it is integral so that $1\leq
  \|A^{-1}AU\|\leq \|A^{-1}\|\|AU\|$ and $\|AU\|\geq
  \frac{1}{\|A^{-1}\|}=\frac{\kappa(A)}{\|A\|}$.
  We deduce: \[ \frac{\|\delta AU\|}{\|AU\|} \leq
  \frac{\epsilon}{3\kappa^2}\]
  and applying the stability theorem we get:
  \[ \frac{\|E\|}{\|AU\|}\leq \frac{\epsilon}{\kappa}. \]
  Using the lower bound on $\|AU\|$ finishes the proof.
\end{proof}
In all \LLL{} algorithms, $\max_i \norm_{\KK/\QQ}(R_{i,i})$ is non-increasing with respect to the round number and
$\min_i \norm_{\KK/\QQ}(R_{i,i})$ is non-decreasing so that we can use the theorem
for all $U$ where $AU$ is size-reduced with
\[ \kappa\leq \left(\frac{\max_i \norm_{\KK/\QQ}(R_{i,i})}{\min_i
\norm_{\KK/\QQ}(R_{i,i})}\right)^{1/n}2^{\bigO{\sqrt{n\log n\log\log n}+d\log n}}. \]

Heuristically, for random lattices, we have $\|U\|\lesssim
\sqrt{\kappa(A)}$ and $\kappa(AU)$ depends only on the dimension so a
truncation of the R-factor of the QR-decomposition of $A$ with error
roughly $\|A\|/\kappa(A)$ is enough.
The precision needed is therefore on the order of $2\log(\kappa(A))$.

\subsection{Faster algorithms}
\label{app:ultra_fast}

We explain here algorithms running in time essentially equal to a matrix multiplication for all previous tasks.
They are only used in~\cref{sec:fastlll}.
We represent a matrix of real numbers by a matrix of integers and a denominator which is a power of two.
Multiplication of matrices, therefore, do not depend on {\em how} the multiplication is computed, as long as it is correct:
whether the corresponding algorithm for floating-point inputs is stable or not is not relevant here.

The QR-decomposition works as follows.
Given the matrix $\begin{pmatrix} A & B\end{pmatrix}$ with $n$ columns and $m\geq n$ rows, we first recursively compute the QR-decomposition of $A=Q_1R_1$.
We let $\conj{Q_1}^tB=\begin{pmatrix} B'_1 \\ B'_2\end{pmatrix}$ where $B'_1$ has as many rows as there are columns in $A$.
Then we compute the QR-decomposition of $B'_2=Q_2R_2$.
The QR-decomposition of the input is then
\[ \left(Q_1\begin{pmatrix}\Id & 0 \\ 0 & Q_2\end{pmatrix}\right)\begin{pmatrix}R_1 & B'_1 \\ 0 & R_2\end{pmatrix}. \]
Remark that
\[ (\Id+XY)(\Id+ZW)=\Id+XY+ZW+XYZW=\Id+\begin{pmatrix} X & Z+XYZ\end{pmatrix}\begin{pmatrix} Y \\ W\end{pmatrix} \]
so we represent all Q matrices in this way, and the base case is done as usual.

We now consider the complexity for a square matrix of dimension $d$.
At the $k$-th recursive levels, the matrices have at most $d/2^k+1$ columns and $d$ rows.
There are $\bigO{2^k}$ rectangular matrix products to be computed, each can be computed in $2^k+1$ products of square matrices of dimension $\leq d/2^k+1$.
The total complexity with $p$ bits of precision is
\[ \bigO{\sum_{k=1}^{1+\log d} 2^{2k}(d/2^k)^{\omega(d/2^k)}p/w+d^2p}=\bigO{\frac{d^{\omega}}{\omega-2}p/w+d^2p\log d}. \]

We can prove by induction for $p\geq \bigO{\log d}$ that for the Q computed
$\hat{Q}$, and for the Q matrix if it were computed exactly\footnote{This matrix is computed from erroneous inputs so that it need not be the Q part of the QR decomposition.} $\breve{Q}$, we have
\[ \| \conj{\breve{Q}}^t\hat{Q}-\Id\|=d^{\bigO{1}}2^{-p}. \]
As a result, $\conj{\breve{Q}}^t$ times the input matrix is with a relative error of $d^{\bigO{1}}2^{-p}$ the computed R;
so that the computed R corresponds to a QR-decomposition of the input matrix with a relative error in the input matrix of $d^{\bigO{1}}2^{-p}$.

We deduce using~\cref{thm:stabqr}:
\begin{theorem}
  Given a matrix $A$, we can compute the $R$-factor of its QR
  decomposition in time $\bigO{\frac{d^\omega p}{(\omega-2)w}+d^2p\log d}$ for $p\geq w+\log(d\kappa(A))$ with a relative error
  of $ 2^{-p} $.
\end{theorem}

We now show a fast size-reduction.
The best-known algorithm for minimizing the condition number of a unitriangular matrix was given by Seysen~\cite{seysen1993simultaneous}.
Using this approach replaces the $(2n)^{\bigO{d}}$ term by a $(2nd)^{\bigO{\log d}}$ in the final condition number.
We explain Seysen's size-reduction as it is both easier and better than the standard one.

It works as follows.
Given a matrix $\begin{pmatrix} A & B \\ 0 & C\end{pmatrix}$, we can assume that both $A$ and $C$ are size-reduced, after two recursive calls.
We then multiply it by
\[ \begin{pmatrix} \Id & -\lfloor A^{-1}B \rceil \\ 0 & \Id \end{pmatrix}\]
and return this matrix.

The result is thus
\[ \begin{pmatrix} A & B-A\lfloor A^{-1}B \rceil \\ 0 & C\end{pmatrix}. \]
and the top-right part is not much larger than $A^{-1}$.
The inverse of the result is
\[ \begin{pmatrix} A^{-1} & -\left(A^{-1}B-\lfloor A^{-1}B \rceil \right)C^{-1} \\ 0 & C^{-1} \end{pmatrix}\]
and the top-right part is not much larger than $C^{-1}$.

We first study an algorithm to invert unitriangular matrices.

 \begin{boxedAlgorithm}[algotitle={Invert}, label=alg:Invert]
    \begin{algorithm}[H]
      \Input{An unitriangular matrix $M$}
      \Output{An approximation of $M^{-1}$}
      \BlankLine
      \lIf{dimension=1}{\Return $1$}
      $\begin{pmatrix} A & B \\ 0 & C\end{pmatrix} \gets M$ \tcp{with a dimension almost halved}
      $A' \gets \algName{Invert}\left(A\right)$ \;
      $C' \gets \algName{Invert}\left(C\right)$ \;
     \Return $\begin{pmatrix} A' & -A'BC' \\ 0 & C' \end{pmatrix}$
    \end{algorithm}
  \end{boxedAlgorithm}

We first prove the performances of the inversion algorithm:

\begin{theorem}
    Given a unitriangular matrix $M$ of dimension $d$ with coefficients in $\KK=\QQ[\zeta_f]$, a field of dimension $n$, with $\|M\|,\|M^{-1}\|\leq 2^p$ and $p\geq w+\log(nd)$,
    \algName{Invert}{} returns a matrix $M'$ such that $\|M'-M^{-1}\|\leq 2^{-p}$
    with a running time of $\bigO{d^{\omega}np/w+d^2np}$.
\end{theorem}
\begin{proof}
    We use a precision $p'=1+2p+\lceil \log(d)\rceil =\bigO{p}$.

    We prove that $\|M'^{-1}-M\|\leq 2d^{0.5}2^{-p'}$ by induction on $d$.
    The case $d=1$ is easy, so we assume $d>1$.
    Let $E$ be such that the top-right part of $M'$ is $-A'BC'+E$, and also $A'^{-1}=A+\delta A$, $B'^{-1}=B+\delta B$.
    Then, we have:
    \[ M'^{-1}-M=\begin{pmatrix} \delta A & -A'^{-1}EC'^{-1} \\ 0 & \delta C\end{pmatrix}. \]
    We can guarantee $\|E\|\leq 2^{-p'-2p}$ with an intermediary bitsize $\bigO{p'}$.
    This leads to our intermediary result.

    Now let $M'^{-1}=M+F$.
    We get \[ M'=(M(\Id+M^{-1}F))^{-1}=(\Id+M^{-1}F)^{-1}M^{-1} \]
    and therefore $\|M'-M^{-1}\|\leq \|M^{-1}\|\|(\Id+M^{-1}F)^{-1}-\Id\| \leq 2^{-p}$.
\end{proof}

 \begin{boxedAlgorithm}[algotitle={Seysen-Size-Reduce}, label=alg:Seysen-SR]
    \begin{algorithm}[H]
      \Input{An unitriangular matrix $M$}
      \Output{An integer unitriangular transformation $U$, and $(AU)^{-1}$}
      \BlankLine
      \lIf{dimension=1}{\Return $1$}
      $\begin{pmatrix} A & B \\ 0 & C\end{pmatrix} \gets M$ \tcp{with a dimension almost halved}
      $U_1 \gets \algName{Seysen-Size-Reduce}\left(A\right)$ \;
      $U_2 \gets \algName{Seysen-Size-Reduce}\left(C\right)$ \;
      $A' \gets \algName{Invert}\left(AU_1\right)$ \;
      $W \gets \lfloor A'BU_2 \rceil $ \;
      \Return $\begin{pmatrix} U_1 & -U_1W \\ 0 & U_2 \end{pmatrix}$
    \end{algorithm}
  \end{boxedAlgorithm}

We finally have:
\begin{theorem}
  Given a unitriangular matrix $M$ of dimension $d$ with coefficients in
  $\KK=\QQ[\zeta_f]$, a field of dimension $n$, with
  $\|M\|,\|M^{-1}\|\leq 2^p$ and $p\geq w+\log(nd)\log(d)$.
  Then \algName{Seysen-Size-Reduce} returns an integer unitriangular matrix $U$ with
  $\|U\|\leq 2^{\bigO{p}}$ such that
  \[ \|MU\|,\|(MU)^{-1}\|\leq (n^{3/2}d)^{\lceil \log d\rceil} \]
  with a running time of $\bigO{d^{\omega}np/w+d^2np}$.
\end{theorem}
\begin{proof}
  We use a precision $p'=\bigO{p+\log(nd)\log(d)}=\bigO{p}$.
  We prove by induction on $d$ that $\|MU\|,\|(MU)^{-1}\|\leq (n^{3/2}d)^{\lceil\log
d \rceil}$.  Initialization is clear, so we assume $d>1$.
  We have that $MU$ is
  \[ \begin{pmatrix} AU_1 & BU_2-AU_1W \\ 0 & CU_2\end{pmatrix}. \]
  The top-right matrix is $AU_1((AU_1)^{-1}BU_2-W)$ and we have, with
  $A'-(A_1U)^{-1}=1+\delta A$:
  \[ \|(AU_1)^{-1}BU_2-W\|\leq \|\delta ABU_2\|+\|A'BU_2-W\|. \]
  The first term is bounded by $2^{\bigO{p}} \|\delta A\|$ and the
  second by $2dn^{3/2}/3$.
  We choose the precision so that the first term is at most $1/3$ and
  the result is proven, as $\|AU_1\|,\|CU_2\|\leq(n^{3/2}d)^{\lceil \log d\rceil-1}$.

    Next, the matrix $(MU)^{-1}$ is
    \[ \begin{pmatrix} (AU_1)^{-1} & -(AU_1)^{-1}(BU_2-AU_1W)(CU_2)^{-1} \\ 0 & (CU_2)^{-1}\end{pmatrix}. \]
    The top-right matrix is $((AU_1)^{-1}BU_2-W)(CU_2)^{-1}$.
    The first term was already bounded above, so $\|(CU_2)^{-1}\|\leq(n^{3/2}d)^{\lceil \log d\rceil-1}$ finishes the proof.

    Finally, we have $\|U\|=\|M^{-1}MU\|\leq \|M^{-1}\|\|MU\|\leq 2^p (n^{3/2}d)^{\lceil \log d\rceil}$.
\end{proof}
Note that it is mandatory to have $M$ well-conditioned if we want a $U$ which is not much larger than $M$.
This is also true for other variants of \LLL~(including \textsf{fplll}): outputting the transition matrix may lead to a slow-down by a factor of $n$.

 \section{Fast unit-rounding in cyclotomics fields}
\label{app:unites}

The goal of this section is to prove \cref{thm:unites}.
In particular we perform a novel analysis of the algorithm of
\cite{EC:CDPR16} to obtain a faster running time and we extend
their result for \emph{arbitrary} cyclotomic fields.

\subsection{Prime power-case}

As a starter, we prove that the techniques of \cite{EC:CDPR16} can be used for
unit-rounding in prime-power cyclotomic fields with quasi-linear
complexity. Formally we aim at proving the following:

\begin{theorem}\label{thm:unites_prime_power}
  Let $\KK$ be the cyclotomic field of prime power conductor $f$.
  There is a quasi-linear randomized algorithm that given any element in
  $x\in (\RR \otimes \KK)^\times$ finds a unit $u\in \order_\KK^\times$
  such that for any field embedding $\sigma:\KK \rightarrow\CC$ we have
  \[ \sigma\left({x}{u}^{-1}\right)= 2^{\bigO{\sqrt{f\log
  f}}}\norm_{\KK/\QQ}(x)^{\frac{1}{\varphi(f)}}.\]
\end{theorem}

Compared to~\cite{EC:CDPR16}, there are two differences
with the treatment proposed here: on the one hand we use fast
arithmetic of the involved objects---namely Fourier-based
multiplication in an abelian group-ring---and on the other hand we
increase the success probability by using a better bound by the
classical Berry-Esseen theorem, as it was hinted in their
seventh footnote.

\subsubsection{Recall on the probability notions used in the proof}

Before diving in the proof of~\cref{thm:unites_prime_power},
let us recall the basis notions of probability theory we are using,
namely subgaussians variables and the Berry-Esseen theorem.

\subsubsection*{On subgaussian random variables.}

The notion of subgaussian distribution goes back to the work of Kahane
in~\cite{Kah60}, and encompasses a large family of real distributions
with very convenient properties similar to the normal law.

\begin{definition}
  A real random variable $X$ is said to be $\tau$-\emph{subgaussian}
  for some $\tau>0$ if the following bound holds for all $s\in\RR$:
  \begin{equation}
    \label{eq:subgdef}
    \E\big[\exp(sX)\big] \leq \exp\Big(\frac{\tau^2 s^2}2\Big).
  \end{equation}
  A $\tau$-subgaussian probability distribution is in an analogous
  manner.
\end{definition}

\begin{lemma}
A $\tau$-subgaussian random variable $X$ satisfies
\[ \E[X] = 0. \]
\end{lemma}
\begin{proof}
	Follows from the Taylor expansion at 0 of $\E[\exp(sX)]=1+s\E[X]+\bigO{s^2}$.
\end{proof}

The main property of subgaussian distributions is that they satisfy a
Gaussian-like tail bound.
\begin{lemma}
  \label{lem:subgaussiantail}
  Let $X$ be a $\tau$-subgaussian distribution. For all $t>0$, we have
  \begin{equation}
    \label{eq:subgtailbound}
    \Pr[ X > t ] \leq \exp\Big(-\frac{t^2}{2\tau^2}\Big).
  \end{equation}
\end{lemma}
\begin{proof}
  Fix $t>0$. For all $s\in\RR$ we have, by Markov's inequality:
  \[
  \Pr[ X>t ] = \Pr[\exp(sX) > \exp(st)] \leq \frac{\E[\exp(sX)]}{\exp(st)}
  \]
since the exponential is positive. Using that $X$ is
$\tau$-subgaussian, \cref{eq:subgdef} gives:
\[ \Pr[ X>t ] \leq \exp\Big(\frac{s^2\tau^2}2 - st\Big) \]
and the right-hand side is minimal for $s=t/\tau^2$, entailing the
announced result.
\end{proof}

Many usual distributions over $\ZZ$ or $\RR$ are subgaussian. This is in
particular the case for distributions with finite supports and zero
mean.

\subsubsection*{The Berry-Esseen approximation theorem}
The Berry-Esseen theorem, or Berry-Esseen inequality, provides a
quantitative estimate of the rate of convergence towards the normal
distribution, as showing that the cumulative function (CDF) of the
probability distribution of the scaled mean of a random sample converges
to $\Phi$ at a rate inversely proportional to the square root of the
number of samples. More formally we have:

\begin{theorem}
  \label{thm:Berry_Esseen}
  There exists a positive $C<0.5$ such that if $X_1, X_2, \cdots, X_n$
  are independent and identically distributed random variables with
  zero mean, satisfying
  $\E(X_1^2) = \sigma^2 > 0$,
  $\E(|X_1|^3) = \rho$, and by setting
  \[Y_n = \frac{X_1 + X_2 + \cdots + X_n}{n}\]
  the sample mean, with $F_n$ the cumulative distribution function of
  $\frac{Y_n \sqrt{n}}{\sigma}$
  and $\Phi$ the cumulative distribution function of the standard normal
  distribution, then for all $x$ and $n$ we have,
  \[
    \left|F_n(x) - \Phi(x)\right| \le \frac{C \rho}{\sigma^3\sqrt{n}}
  \]
\end{theorem}

\subsubsection{Going back on the rounding problem}

We now fix a cyclotomic field $\KK=\QQ[\zeta_f]$ with prime
power-conductor $f$. We recall that in $\KK$, the cyclotomic units are
easily described:

\begin{lemma}[Lemma 8.1 of~\cite{Was97}]
  Let $f$ be a prime power, then the group of cyclotomic units is
  generated by $\pm \zeta_f$ and $\frac{\zeta_f^\alpha -1}{\zeta_f-1}$
  for $\alpha\in (\ZZ/f\ZZ)^\times$.
\end{lemma}

We first provide a convenient description of the cyclotomic units
as an orbit of the element $\zeta_f-1$ under the action of its Galois
group.

\subsubsection{Log-embedding and action of $(\ZZ/f\ZZ)^\times /
\{-1,+1\}$.}

Define the $\Log$ embedding to be the coefficient-wise
composition of the real logarithm with the absolute value of the Archimedean
embeddings:
\[
  \textrm{Log}:\left|
  \begin{array}{rcl}
     \KK & \longrightarrow & \RR^\frac{n}{2} \\
     \alpha & \longmapsto &
    \left[\log\left(|\sigma_i(\alpha)|\right)\right]_{i\in G}
  \end{array}
\right.,
\]
where the embeddings are paired by conjugates and listed by the group $G
= (\ZZ/f\ZZ)^\times / \{-1,+1\}$. The image of the unit
multiplicative group $\order_\KK^\times$ is a full rank lattice by
Dirichlet unit's theorem, and is called the \emph{Log-unit lattice}.

We first remark that the group-ring $\ZZ[(\ZZ/f\ZZ)^\times]$ acts on the
group $(\RR \otimes \KK)^\times$ in the following way: for any
$g=\sum_\alpha g_\alpha\alpha\in\ZZ[(\ZZ/f\ZZ)^\times]$ and $x\in(\RR \otimes \KK)^\times$,
\[
  g\cdot x=
      \prod_{\alpha \in (\ZZ/f\ZZ)^\times}
    \sigma_\alpha(x)^{g_\alpha},
\]
where $\sigma_\alpha$ maps $\zeta_f$ to $\zeta_f^\alpha$.
But $\sigma_\alpha$ acts as a permutation on the Archimedean embedding
so that the embedding in the Log-unit lattice \emph{commutes} with the
action of $\ZZ[(\ZZ/f\ZZ)^\times]$ in the following sense: \[\Log(g\cdot
x)=g\Log(x) \in \RR[G], \] for all $x\in(\RR \otimes \KK)^\times$.

Henceforth, the cyclotomic units can be described using this action,
as they correspond to the orbit of the element $\zeta_f-1$ by the
kernel, called the \emph{augmentation ideal}, of $g\mapsto \sum_\alpha g_\alpha$:
\begin{equation}
  \label{eq:description_of_units}
  \left\{ g\cdot (\zeta_f-1)\, |\, \sum_\alpha g_\alpha = 0 \right\}
\end{equation}

\subsubsection{An upper bound on the norm of $\Log(\zeta_f-1)$.}

We also have that $\Log(\zeta_f-1)$ is invertible, and with a small
inverse (for example $\|\Log(\zeta_f-1)\|=\bigO{n^3}$) so that we can compute
efficiently. Let us formalize this intuition. We first bound $\Log(\zeta_f-1)$:

\begin{lemma}
  \label{lem:estimate_log}
  We have $\|\Log(\zeta_f-1)\|_\infty\leq \log f$ and
  $\|\Log(\zeta_f-1)\|_2 =\bigO{\sqrt{f}}$.
\end{lemma}
\begin{proof}
  The coordinates are given by
  \[
    \Log(\zeta_f-1)_\alpha=
    \Log(|\zeta_f^\alpha-1|)=\log(|2\sin(\pi\alpha/f)|),
  \]
  for any $\alpha\in(\ZZ/f\ZZ)^\times$.
  Now, for $0\leq x \leq \frac12$ and $\alpha\in(\ZZ/f\ZZ)^\times$,
  we have $\sin(\pi x)\geq 2x$ and we can
  consider that $0\leq \frac{\alpha}{f}\leq \frac12$.
  We deduce that $\|\Log(\zeta_f-1)\|_\infty\leq
  \log\left(\frac{f}{4}\right)$ and
  \[ \|\Log(\zeta_f-1)\|_2^2\leq \sum_\alpha
    \log^2\left(\frac{f}{4\alpha}\right)\leq
  f\int_{0}^{\frac12} \log^2\left(\frac4x\right) dx,\] the latest
  integral being equal to $\frac92+\frac{3}{\ln 2}+\frac{1}{\ln^2 2}$
  entails the announced inequality.
\end{proof}

\begin{remark}
The multiplication in the group ring $\ZZ[G]$ is quasi-linear as $G$ is
a finite abelian group. Indeed, we can use Fourier transform to reduce
the multiplication to point-wise multiplications (see for
instance~\cite{maslen}).
\end{remark}

\subsubsection{Fast rounding in the Log-unit lattice}
We can now describe the rounding algorithm, which essentially
is a randomized coefficient-wise rounding using the orbital description
of \cref{eq:description_of_units}.

\begin{proof}[Proof of \cref{thm:unites_prime_power}]
  Without loss of generality, we can assume $\norm_{\KK/\QQ}(x)=1$.

  Then, using the description given by \cref{eq:description_of_units}
  the problem is thus reduced to searching a unit $u$ such that
  $\Log(u)\in \ZZ[G]$ which is close to
  $y=\frac{\Log(x)}{\Log(\zeta_f-1)}$ and such that $\sum_\alpha
  \Log(u)_\alpha=0$. The simplest idea consists in performing a
  coefficient wise rounding of the coefficients of the vector $y$.
  However, this approach does not succeed all the time, but we can take
  advantage of the two possible choices in the rounding to closest
  integers to randomize the rounding---that is to say, by randomizing
  the choice of floor or ceil instead of relying deterministically on
  the round function $\lfloor\cdot\rceil$. \medskip

  Formally, for $\alpha\neq 1$, we sample $z_\alpha$ following the
  unique distribution on the two elements set $\{\lfloor
  y_\alpha\rfloor,\lceil y_\alpha\rceil\}$
  with expectation $y_\alpha$.
  Then, $z_1$ is set at $-\sum_{\alpha\neq 1} z_\alpha$ to
  ensure $\sum_\alpha z_\alpha=0$.
  Clearly, $u=z\cdot (\zeta_f-1)$ verifies our requirements if
  \[ \left\|\Log(\zeta_f-1)(y-z)\right\|_\infty=\bigO{\sqrt{f\log f}}.\]

  The Berry-Esseen theorem indicates that $|y_1-z_1|\leq \sqrt{n}/\log
  n$ with probability $\Theta(1/\log n)$.
  The coordinates of \[ \Log(\zeta_f-1)(y-z-(y-z)_1\sigma_1) \]
  are subgaussians of parameter $\|\Log(\zeta_f-1)\|_2$.
  Therefore, using the estimation of \cref{lem:estimate_log}, we know that
  their absolute
  values can all be bounded by
  $\bigO{\sqrt{f\log f}}$ except with probability at most
  $\Theta\left(\frac1{\log^2 f}\right)$. Hence, our requirement is
  fulfilled with probability $\Omega\left(\frac1{\log n}\right)$. We have
  $\Log(u)=z\Log(\zeta_f-1)$ which can be computed in quasi linear time.
  Eventually a Fourier transform recovers $\sqrt{u\bar u}$, which is $u$ up
  to an irrelevant torsion\footnote{One can compute $u$ by simply removing
  the absolute values in the definition of $\Log$, and taking any
determination of complex logarithm. As we work inside a CM-field, this
technicality is not needed.}.
\end{proof}

\subsection{Extension to arbitrary cyclotomic fields}
\label{app:generalized_unites}

We now extend the result of \cref{thm:unites_prime_power} to
arbitrary cyclotomic fields, that is proving:

\begin{theorem}\label{thm:unites_general}
  Let $\KK$ be the cyclotomic field of conductor $f$.
  There is a quasi-linear randomized algorithm that given any element in
  $x\in (\RR \otimes \KK)^\times$ finds a unit $u\in \order_\KK^\times$
  such that for any field embedding $\sigma:\KK \rightarrow\CC$ we have
  \[ \sigma\left({x}{u}^{-1}\right)= 2^{\bigO{\sqrt{f\log
  f}}}\norm_{\KK/\QQ}(x)^{\frac{1}{\varphi(f)}}.\]
\end{theorem}

\subsubsection{Setting.}

Let us consider an integer $f$ and take its prime decomposition
$f=\prod_{i=1}^r p_i^{e_i}$. We set $q_i=p_i^{e_i}$ and we fix the
cyclotomic field $\KK=\QQ[\zeta_f]$ of conductor $f$.  Classically, the
Galois group of $\KK$ is equal to
$G=\faktor{(\ZZ/f\ZZ)^\times}{\{-1,1\}}$, whose elements are the
$\sigma_\alpha$, sending $\zeta_f$ to $\zeta_f^\alpha$ for
any $\alpha\in G$.

\subsubsection{Cyclotomic units and their generators.}

The cyclotomic units are defined as all the products
of $\pm \zeta_f$ and $\zeta_f^a-1$ which are units.
We let $\mathcal{Q}$ be the set of the $2^r$ possible products of the $q_i$.

A standard theorem of~\cite[Lemma 2.2]{kuvcera1992bases} reduces the
number of generators of the cyclotomic units:
\begin{theorem}
  The cyclotomic units are all the products of $\pm \zeta_f$ and $G\cdot
  (\zeta_f^a-1)$ which are units, when $a$ runs through $\mathcal{Q}$.
\end{theorem}
\begin{proof}
  Let $a\in \ZZ$, and define $k$ to be the product of all the $q_i$
  dividing $a$, so that by construction $k\in \mathcal{Q}$.
  Now, we have:
  \[ 1-\zeta_f^a=\prod_{i=0}^{\frac{a}{k}-1} 1-\zeta_f^{k+\frac{ifk}{a}}. \]
  Let $p_j|k+\frac{ifk}{a}$. Remark that $p_j|\frac{fk}{a}$, so that
  $p_j|k$, and by definition of $k$ we have $q_j|k$.
  We have therefore $q_j|\frac{fk}{a}$ and hence $\zeta_f^{k+ifk/a}-1 \in \pm
  G\cdot \zeta_f^k-1$.
\end{proof}

\begin{theorem}\label{thm:eval_character}
  Let $\chi$ be an even Dirichlet character of conductor $c\,|\,f$ with $c>1$ and
  $e\in \mathcal{Q}$.  Then if $c$ and $e$ are coprime, then \[
  |\chi(\Log(\zeta_f^e-1))|=\frac{\varphi(e)\sqrt{c}}{2\ln(2)}\left(\prod_{\substack{i\\
p_i|\frac{f}{e}}} |1-\chi(p_i)| \right)|L(1,\chi)| \]
  else it is 0.
\end{theorem}
\begin{proof}
  If $\gcd(c,e)>1$, we have $\sum_{\alpha\in
  (\ZZ/\gcd(c,e)\ZZ)^\times} \chi(\alpha)=0$ so the result is zero.
  We therefore assume for now on that $c$ and $e$ are coprime.

  We first compute: \[ \prod_{\substack{\beta\in G\\ \beta=1 \bmod
  c}}1-\zeta_{\frac{f}{e}}^{\beta}. \]
  Let $p_i|\frac{f}{ec}$ and $p_i|c$.
  Then:
  \begin{align*}
    \prod_{\substack{\beta\in G\\ \beta=1 \bmod c}}1-\zeta_{\frac{f}{e}}^{\beta}
    &=
    \prod_{\substack{\beta\in G\\ \beta=1 \bmod
cp_i}}\prod_{j=0}^{p_i-1} 1-\zeta_{\frac{f}{e}}^{\beta}\zeta_{p_i}^j \\&=
\prod_{\substack{\beta\in G\\ \beta=1 \bmod
cp_i}}1-\zeta_{\frac{f}{e}}^{p_i\beta}.  \end{align*}
  In the same way, we have if $p_i|\frac{f}{e}$ and $p_i\nmid c$, with
  $r^{-1}=\frac{f}{eq_i}\bmod p_i$:
  \begin{align*}
    \prod_{\substack{\beta\in G\\ \beta=1 \bmod c}}1-\zeta_{\frac{f}{e}}^{\beta}
    &=
    \prod_{\substack{\beta\in G\\ \beta=1 \bmod
        cq_i}}\prod_{\substack{j=0\\ j\neq -r \bmod p_i}}^{q_i-1}
        1-\zeta_{\frac{f}{e}}^{\beta}\zeta_{q_i}^{\beta j} \\&=
        \prod_{\substack{\beta\in G\\ \beta=1 \bmod
            cq_i}}\frac{1-\zeta_{\frac{f}{e}}^{\beta
    q_i}}{1-\zeta_{\frac{f}{e}}^{\beta(q_i-\frac{rf}{e})/p_i}} \\&=
    \prod_{\substack{\beta\in G\\ \beta=1 \bmod
    cq_i}}\frac{1-\zeta_{\frac{f}{eq_i}}^{\beta}}{1-\zeta_{\frac{f}{eq_i}}^{\frac{\beta}{p_i}}}.
\end{align*}
  In case $p_i|e$, we have $q_i|e$ and therefore
\[ \prod_{\substack{\beta\in G\\ \beta=1 \bmod
c}}1-\zeta_{\frac{f}{e}}^{\beta}=\prod_{\substack{\beta\in G\\ \beta=1 \bmod
cq_i}} \left(1-\zeta_{\frac{f}{e}}^{\beta}\right)^{\varphi(q_i)}. \]

We can now compute our sum:
\begin{align*}
   \sum_{\alpha\in G} \chi(\alpha)\log(|\zeta_f^{e\alpha}-1|)  &=
  \sum_{\alpha\in (\ZZ/c\ZZ)^\times/\{-1,1\}}
  \chi(\alpha)\log\left(\left|\sigma_\alpha\left(\prod_{\substack{\beta\in
          G
  \beta=1 \bmod c}}\zeta_{\frac{f}{e}}^{\beta}-1\right)\right|\right)  \\ &=
   \varphi(e)\bigg(\prod_{\substack{i\\ p_i|\frac{f}{e} \\ p_i\nmid c}}
   1-\chi(p_i) \bigg) \sum_{\alpha\in (\ZZ/c\ZZ)^\times/\{-1,1\}}
   \chi(\alpha)\log(|\zeta_c^\alpha-1|).
\end{align*}
  We finish by the standard computation (\cite[Theorem 4.9]{Was97}) of
  the term on the right with the Gauss sum:
  $\tau=\sum_{\alpha\in(\ZZ/c\ZZ)^\times}
  \conj\chi(\alpha)\zeta_c^\alpha$:
  \begin{align*} \sum_{\alpha\in (\ZZ/c\ZZ)^\times}
    \conj\chi(\alpha)\ln(|\zeta_c^\alpha-1|) &=
    \sum_{\alpha\in (\ZZ/c\ZZ)^\times}
    \conj\chi(\alpha)\ln(1-\zeta_c^\alpha)  \\ &=
    \sum_{\alpha\in (\ZZ/c\ZZ)^\times} \sum_{k=1}^\infty
    \conj\chi(\alpha)\frac{\zeta_c^{\alpha k}}{k} \\ &=
    \sum_{i=1}^\infty \frac{\tau\chi(k)}{k} = \tau L(1,\chi)
  \end{align*}
  and $\tau \conj \tau=c$.
\end{proof}

\begin{definition}
  The \emph{augmentation ideal} is the kernel of the form:
  $\left(\sum_\alpha x_\alpha \sigma_\alpha \to \sum_\alpha
  x_\alpha\right)$ over $\ZZ[G]$.
\end{definition}

With this definition we can complete the description of the cyclotomic
units:
\begin{theorem}{\cite[Lemma 2.4]{kuvcera1992bases}}
  The cyclotomic units are generated by:
  \begin{itemize}
    \item The pair $\pm \zeta_f$,
    \item the $G \cdot \zeta_f^a-1$ for
      all $a\in \mathcal{Q}$ such that $\frac{f}{a}$ is not prime power,
    \item the orbit of $\zeta_f^{f/q_i}-1$ by the action of the
      augmentation ideal.
  \end{itemize}
\end{theorem}
\begin{proof}
  Note first that for any $a\in \mathcal{Q}$,
$(1-\sigma_\alpha)\cdot (\zeta_f^a-1) \in \order_\KK$.
  Next, we prove that an element $u$ generated by the $\zeta_f^a-1$ is a
  unit if $\norm_{\KK/\QQ}(u)=1$.
	We remark that 	\[ \varphi(f)\cdot u= \norm_{\KK/\QQ}(u)\left(
    \left(\sum_\alpha
	1-\sigma_\alpha\right)\cdot u\right)=\left(\sum_\alpha 1-\sigma_\alpha\right)
\cdot u \]
  so that it is a unit.  The converse is clear.  Finally
  $\norm_{\KK/\QQ}(1-\zeta_f^a)$ is easily
  computed to be $p_i^{\varphi(a)}$ if $a=f/q_i$ and $1$ else using the
  equations at the beginning of the proof of
  \cref{thm:eval_character}.
\end{proof}

\subsubsection{Construction of an ``orthogonal'' basis}
We now define the family $(b_i)_{1\leq i\leq |\mathcal{Q}|}$ by setting
$b_i=\Log(\zeta_f^a-1)$ where the $a\in \mathcal{Q}$ are taken in
decreasing order.
We can define some Gram-Schmidt orthogonalization on this family with
the relations:
\[
b_i^*=b_i-\sum_{j<i}
\frac{\inner{b_i}{b_j^*}}{\inner{b_j^*}{b_j^*}}b_j^*=b_i-\sum_{j<i}
b_ib_j^*{(b_j^*)^\dagger}
\]
where the dagger is the Moore-Penrose pseudo-inverse.
As such, $\chi(b_i)=\chi(b_i^*)$ if $\chi(b_j^*)=0$ for all $j<i$, and
is equal to zero elsewhere.
As $L(1,\chi)\neq 0$, we have for all $\chi\neq 1$ that $\chi(b_i^*)\neq
0$ iff $\rad(\frac{f}{e})|c|\frac{f}{e}$ where $c$ is the conductor of the character
$\chi$.
Furthermore, in this case, the term $\prod_{p_i|\frac{f}{e}}
(1-\chi(p_i))$ is
one.  We can now give our decoding algorithm, assuming again that the
cyclotomic units have a finite index:

\begin{proof}[Proof of \cref{thm:unites_general}]
  We let $b_i=\Log(\zeta_f^e-1)$ and recall that for all $\chi$ with
  conductor not coprime with $e$ we have $\chi(b_i)=0$.
  We remark that if $\frac{f}{e}$ is a prime power, we have $b_i^*=b_i$ and as a
  result $\|b_i^* \|_\infty\leq \log(\frac{f}{e}).$
  Also, we have for all $i$ that $\|b_i^*\|\leq
  \|b_i\|=\bigO{\sqrt{\varphi\left(\frac{f}{e}\right)}}$ using the same technique.  The
  algorithm consists in using Babai reduction with our generating
  family, with the modification described above to round with respect to
  the augmentation ideal when we have to.
  More precisely, for any $y\in \ZZ[G]b_i^*$, we compute $z$ a
  randomized rounding of $y/b_i^*$ in the same way as in the previous
  section.
  If $\frac{f}{e}$ is a prime power, the rounding is $z-\sum_\alpha z_\alpha
  \sigma_1$, else it is $z$.
  If $|\sum_\alpha z_\alpha|\geq
  \frac{\sqrt{\frac{f}{e}}}{\log(\frac{f}{e})}$ in case where
  $\frac{f}{e}$ is a prime power, we restart the rounding.
  We then continue in the same way with $i-1$. The analysis is as
  before.
  The randomized rounding produces an error with subgaussian coordinates
  with parameter $\bigO{\sqrt{\sum_{e\in X}
  \varphi(\frac{f}{e})}}=\bigO{\sqrt{f}}$.
  The correction for the prime power adds an error bounded by $\sum_i
  \log(q_i)\sqrt{q_i}/\log(q_i)=O(\sqrt{f})$.
  Hence, the bound on the output holds.  The running time is
  quasi-linear since we can work at each step with the ring \[
    \ZZ\left[\left(\ZZ\bigg/\left(\frac{f}{e}\right)\ZZ\right)^\times\right].\]
\end{proof}
Remark that the running time is also quasi-linear if we work with the
input and output in the logarithm space.
Note that $\KK^+=\QQ[\zeta_f+\conj\zeta_f]$ has the same units, up to
torsion. As such, the same theorem is true for $\KK^+$.
It has the following algorithmic implication.
Given an ideal $\ideal{a} \subset \order_{\KK^+}$, as the class group
order of $\KK^+$ is usually small, it is simple to find an ideal $\alpha
\order_{\KK^+}\subset \ideal{a}$ with low norm.
From there, we can compute a generator $\alpha$ in quantum polynomial
time and using the above theorem on $\alpha$, we have found quickly an
element in $\ideal{a}$ with approximation factor $2^{O(\sqrt{f\log
f})}$.

\subsection{BDD on the unit lattice}
The following theorem has deep
implications in arithmetic. One part is due to
Landau~\cite{landau1927dirichletsche}, another to Dirichlet.
\begin{theorem}
  Let $\chi$ be a character of conductor $c>1$.
  If $\chi^2=1$ ($\chi$ is quadratic) we have
  $|L(1,\chi)|=\Omega(1/\sqrt{c})$, else
  $|L(1,\chi)|=\Omega(1/\log(c))$.
\end{theorem}
Note that under the Generalized Riemann Hypothesis we can take
$|L(1,\chi)|=\Omega(1/\log \log c)$ and for most characters we have
$|L(1,\chi)|=\bigO{1}$.
This justifies our previous assumptions.  We let $\tau(f)=\prod_i 1+e_i$
be the number of divisors of $f$; we have the well-known bound $\tau(f)=
f^{\bigO{1/\log \log f}}$.  We can now prove our BDD\footnote{The usual
  definition of BDD is about the worst case decoding distance. The
implied worst case bound is too large to be useful, but with high
probability we can decode large Gaussian noise, which is enough for
current applications.} theorem:
\begin{theorem}
  Given $\KK=\QQ[\zeta_f]$, there are $\varphi(f)/2$ (explicit) elements
$r_i$ of norm \[ \bigO{\frac{\sqrt{\tau(f)}}{n}\log(n)} \] in $\RR[G]$
with the following property.
  Let $x \in (\RR \otimes \KK)^\times$ be such that there is a
  cyclotomic unit $u$ with for all $i$, $ |\inner{r_i}{\Log(x/u)}|<1/3$.
  Then, given $x$ we can find $u$ up to a power of $\zeta_f$ in
  quasi-linear time.
\end{theorem}
\begin{proof}
  The algorithm is similar to the previous one.
  We first scale $x$ to get $\norm_{\KK/\QQ}(x)=1$.
  For decreasing $i$, we compute $z$ the (deterministic) rounding
  $\Log(x)/b_i^*$ where we force $\sum_\alpha z_\alpha=0$ if
  $b_i=\Log(\zeta_f^e-1)$ with $\frac{f}{e}$ a prime power,
  and we then divide $x$ by $z \cdot \zeta_f^e-1$.  We first bound
  $\|({b_i^*})^\dagger\|_2$ where $b_i=\zeta_f^e-1$.
  Thanks to our previous computations and the character orthogonality
  relation, we have
  \[ \left\| ({b_i^*})^\dagger-\frac{\sum_\alpha
      \sigma_\alpha}{\sum_\alpha
      (b_i^*)_\alpha} \right\|^2 =\frac{1}{|G|}\sum_{\chi} \frac{4\ln^2
  2}{c\varphi(e)^2} \frac{1}{|L(1,\chi)|^2} \]
  where $\chi$ has a conductor $c>1$ with
  $\rad(\frac{f}{e})|c|\frac{f}{e}$.
  The Chinese Remainder theorem implies that:
  \[ \sum_\chi \frac{1}{c}=\frac{1}{2}\prod_{p_i|\frac{f}{e}}
    \sum_{k=1}^{e_i}
  \frac{(p_i-1)p_i^{k-1}}{p_i^k}=\frac{1}{2} \prod_{p_i|\frac{f}{e}}
e_i(1-1/p_i) \]
  with the same assumptions on $\chi$.
  We have at most $2^{r+1}$ quadratic characters, so we get:
  \[ \left\| ({b_i^*})^\dagger-\frac{\sum_\alpha
      \sigma_\alpha}{\sum_\alpha
      (b_i^*)_\alpha} \right\|^2 \leq
      \bigO{\frac{2^r+\log^2\left(\frac{f}{e}\right)\prod_{p_i|\frac{f}{e}}
  e_i}{\varphi(f)\varphi(e)^2}}\]
  which is in $\bigO{n^{-1}\log^2(n)\tau(f)}$.  Now each non-zero
  coefficient of $z$ in the algorithm can be expressed as an inner
  product between an element of ${G}({b_i^*})^\dagger$ and $\Log(x/u)$, which is of
  unit norm.
  This leads to a $r$ vector for each coefficient, and with the given
  condition this guarantees that $\Log(u)$ is exactly recovered.
\end{proof}

This implies that given any generator of the ideal
$\alpha\order_\KK$ where $\alpha$ is sampled from a large discrete
Gaussian, we can recover $\alpha$ in quasi-linear time;
see~\cite[Section 5]{EC:CDPR16}.  The practical average length of the $r_i$ is of
course on the order of $\sqrt{\frac{\prod_i e_i}{n}}$.

 \section{The symplectic structure in all number fields}
\label{app:sympallnf}
In~\cref{sec:symplectic}, we described how to obtain a symplectic structure when $\lL=\KK[X]/(X^d+a)$.
We show here the general case, with $\lL=\KK[X]/f(X)$.
We first give a simple construction which recovers the one given above but has losses in the general case; and then describe a general construction without losses.

\subsection{The dual integer construction}

We have the following lemma, proved in~\cite[Chapter III, Proposition 2.4]{Neukirch}:
\begin{lemma}
	Let $a_i=X^i$ and $\sum_i b_i Y^i=\frac{f(Y)}{Y-X}$.
	Then $\tr_{\lL/\KK}(a_ib_j/f'(X))$ is equal to 1 if $i=j$ and $0$ else.
\end{lemma}

This suggests taking as a $\KK-$basis for $\lL^2$ the $(a_i,0)$ followed by the $(0,b_i)$.
With the notations of~\cref{sec:symplectic}, we now define $J'_{\lL}$ as \[ \tr_{\lL/\KK}(J_{\lL}/f'(X)). \]
It follows from the lemma that in our basis, this is represented by the Darboux matrix:
\[ \begin{pmatrix} 0 & \Id_d \\ -\Id_d & 0 \end{pmatrix} \]
and, as usual, we can reverse the order of the second part of the basis to obtain the wanted matrix.

We can convert efficiently a number $z\in \lL$ in the basis of $b_i$.
Clearly, the coefficients are given by all the $\tr_{\lL/\KK}(z/f'(X)\cdot X^i)$.
We then simply evaluate $z/f'(X)$ on all roots of $f$ using a remainder tree, and follow by a Vandermonde matrix-vector multiplication, which is also a multipoint evaluation~\cite{moenck1972fast}.
In particular, we do not need to compute the $b_i$.

There is however a loss with this basis: the algorithm tries to minimize the size of the coefficients in our basis of $\lL^2$ instead of the canonical norm.

\subsection{The orthogonal construction}

We want to build an orthogonal $\RR \otimes \KK$-basis of $\RR \otimes \lL$.
We assume for simplicity (only) that $\lL$ (and therefore $\KK$) is a totally real field.
Hence, with $\KK=\QQ[Y]/g(Y)$, we have that all roots $r_i$ of $g$ are real, and when we evaluate all coefficients of $f$ on $r_i$, the resulting polynomial has real roots $r_{i,j}$.

We then define the $j$-th element of the basis as being the element of $\lL$ which, when we evaluate on $(X-r_{i,k},Y-r_i)$, we obtain 1 if $j=k$ and 0 else.
This is clearly an orthogonal basis for the canonical norm, and in this case, it is also its dual.
Hence, using twice this basis leads again to the Darboux matrix for $J'_{\lL}=\tr_{\lL/\KK}(J_{\lL})$.
Exactly the same construction works for totally imaginary $\KK$ (and therefore $\lL$).

The general case can be done in the same way, by taking care of ramified places.
 \section{Reduction with linear algebra}
\label{app:reduction}

We shall prove that lattice reduction is no easier than linear algebra on a large field $\ZZ/p$.
We start by defining the problems.

\begin{definition}[Lattice reduction]
	The problem of lattice reduction consists in, given an integer matrix $A$ of dimension $d$ with $\|A\|,\|A^{-1}\|\leq 2^B$, outputting
	a matrix $AU$ with $U$ a unimodular integer matrix such that with $QR=AU$ the QR-decomposition, we have for all $i$:
	\[ R_{i,i}\leq 2R_{i+1,i+1}. \]
\end{definition}

\begin{definition}[Kernel problem]
	The kernel problem consists in, given a square matrix $A$ of dimension $d$ over $\ZZ/p$, outputting a matrix $K$ such that $AK=0$ and the number of columns of $K$ is $\dim \ker A$.
\end{definition}

\begin{theorem}
	If one can solve the lattice problem in dimension $2d$ with parameter $B$, one can solve the kernel problem in dimension $d$ for any prime $p\leq 2^{B/2-d}d^{-1}$ with the same complexity, up to a constant.
\end{theorem}
\begin{proof}
	Let $A$,$p$ be the input of the kernel problem.
	The matrix
	\[ L=\begin{pmatrix} pd2^{2d-1}\cdot p\Id_d & pd2^{2d-1} A \\ 0 & \Id_d\end{pmatrix} \]
	is given to the lattice reduction oracle.
	The output is of the form
	\[ \begin{pmatrix} 0 & * \\ K & *\end{pmatrix} \]
	where we maximize the number $k$ of columns of $K$.
	The reduction returns this matrix $K$.

	We have $\|L\|\leq d^22^{2d}p^2\leq 2^B$ and $\|L^{-1}\|\leq 2d$ which is also less than $2^B$ since $p\geq 2$.
	It is clear that vectors in $L\ZZ^{2d}$ of the form $\begin{pmatrix} 0 \\ x\end{pmatrix}$ are exactly the integer solutions of $Ax=0 \Mod p$.
	We let $QR$ be the QR-decomposition of $AU$.
	Let $K'$ be a basis of $\ker A$, where entries are integers smaller than $p$.
	Then, since $U$ is unimodular, there is an integer matrix $V$ such that \[ AUV=\begin{pmatrix} 0 \\ K'\end{pmatrix}.\]
	If $V$ has no nonzero entries $V_{i,j}$ with $i>k$, then it is clear that the output is correct.
	Hence, we consider $v$ a column of $V$ where it is not the case, and let $i$ be maximal with $v_i\neq 0$.
	First, we have $\|AUv\|\leq \sqrt{d}p$.
	Second, as $Q$ is orthogonal, we have $\|AUv\|=\|Rv\|\geq R_{i,i}$.
	Third, the definition of $k$ implies that $R_{k+1,k+1}\geq d2^{2d-1}p$.
	As the lattice is reduced and $i>k$, we have $R_{i,i}\geq R_{k+1,k+1}2^{1-2d}$.
	We conclude that:
	\[ \sqrt{d}p\geq \|AUv\|\geq R_{k+1,k+1}2^{1-2d}\geq dp \]
	which is a contradiction.
\end{proof}

As we expect the kernel problem to have a complexity of $\Omega(d^\omega B/\log B+d^2B)$, we can expect the same for the lattice reduction problem.
The reduction can of course be extended with other rings, and also to compute a span.

\end{document}